\documentclass[bj,preprint]{imsart}

\RequirePackage[OT1]{fontenc}
\RequirePackage{amsthm,amsmath,natbib}

\startlocaldefs
\numberwithin{equation}{section}
\theoremstyle{plain}
\newtheorem{theorem}{Theorem}[section]
\newtheorem{example}{Example}[section]
\newtheorem{lemma}{Lemma}[section]
\newtheorem{proposition}{Proposition}[section]

\newtheorem{definition}{Definition}[section]

\theoremstyle{remark}

\endlocaldefs

\usepackage{silence}
\usepackage{array}
\usepackage{amsmath}
\usepackage{amsfonts,amssymb}
\usepackage{cancel}
\usepackage{etoolbox}
\usepackage{dsfont} 
\usepackage{fontawesome} 
\usepackage{graphicx}
\graphicspath{{images/}}
\usepackage{caption}
\usepackage{subcaption}
\makeatletter
\@ifpackageloaded
    {hyperref}
    {}
    {\usepackage{hyperref}}
\hypersetup{
    breaklinks=true,
    colorlinks=true,
    linkcolor=blue,
    urlcolor=gray,
    citecolor=blue}
\makeatother
\usepackage[utf8x]{inputenc}
\usepackage{multirow}
\usepackage{natbib}
\usepackage{pgffor}
\usepackage[usenames,dvipsnames]{xcolor}

\usepackage{adjustbox}
\usepackage{multirow}
\usepackage{tabu}

\usepackage{algorithm, algorithmic}

\usepackage{silence}
\newcommand{\DPPy}{\textsf{DPPy}}
\newcommand{\Jab}{J_{\mathbf{a},\mathbf{b}}}
\newcommand{\tildeJab}{T_{\mathbf{a},\mathbf{b}}}
\newcommand{\mueq}{\mathrm{\mu_{eq}}}
\def\atomSpace{\bbR_>^N}
\def\weightSpace{S_N}

\newcommand{\CommaBin}{\mathbin{\raisebox{0.5ex}{,}}}

\newcommand*{\Algoref}[1]{Algorithm~\ref{#1}}

\newcommand*{\Exref}[1]{Example~\ref{#1}}
\newcommand*{\Figref}[1]{Figure~\ref{#1}}
\newcommand*{\Lemref}[1]{Lemma~\ref{#1}}
\newcommand*{\Proporef}[1]{Proposition~\ref{#1}}
\newcommand*{\Secref}[1]{Section~\ref{#1}}

\newcommand*{\Thref}[1]{Theorem~\ref{#1}}

\makeatletter
\patchcmd{\NAT@test}{\else \NAT@nm}{\else \NAT@nmfmt{\NAT@nm}}{}{}
\DeclareRobustCommand\citepos
  {\begingroup
   \let\NAT@nmfmt\NAT@posfmt
   \NAT@swafalse\let\NAT@ctype\z@\NAT@partrue
   \@ifstar{\NAT@fulltrue\NAT@citetp}{\NAT@fullfalse\NAT@citetp}}
\let\NAT@orig@nmfmt\NAT@nmfmt
\def\NAT@posfmt#1{\NAT@orig@nmfmt{#1's}}
\makeatother

\foreach \x in {a,...,z}{
	\expandafter\xdef\csname bf\x
		\endcsname{
			\noexpand\ensuremath{\noexpand\mathbf{\x}}
		}
}

\foreach \x in {A,...,Z}{
	\expandafter\xdef\csname bf\x
		\endcsname{
			\noexpand\ensuremath{\noexpand\mathbf{\x}}
			}
	\expandafter\xdef\csname cal\x
		\endcsname{
			\noexpand\ensuremath{\noexpand\mathcal{\x}}
			}
	\expandafter\xdef\csname bb\x
		\endcsname{
			\noexpand\ensuremath{\noexpand\mathbb{\x}}
			}
}

\foreach \x/\y in {alpha/\alpha, beta/\beta, gamma/\gamma, delta/\delta, epsilon/\epsilon, zeta/\zeta, eta/\eta, theta/\theta, iota/\iota, kappa/\kappa, lambda/\lambda, mu/\mu, nu/\nu, xi/\xi, pi/\pi, rho/\rho, sigma/\sigma, tau/\tau, upsilon/\upsilon, phi/\phi, chi/\chi, psi/\psi, omega/\omega }{
	\expandafter\xdef\csname bf\x
		\endcsname{
			\noexpand\ensuremath{\noexpand\boldsymbol{\y}}
			}
}

\renewcommand{\hat}{\widehat}

\newcommand{\oplim}[3]{%
\operatornamewithlimits{#1}\limits_{\substack{#2}}^{\substack{#3}}}

\newcommand{\argmax}{\operatornamewithlimits{argmax}}
\newcommand{\rddots}{\reflectbox{$\ddots$}}

\newcommand{\e}{\operatorname{\mathrm{e}}}

\newcommand*{\diff}{\mathrm{d}}

\let\originalleft\left
\let\originalright\right
\renewcommand{\left}{\mathopen{}\mathclose\bgroup\originalleft}
\renewcommand{\right}{\aftergroup\egroup\originalright}
\newcommand{\lrb}[1]{\left[ #1 \right]}
\newcommand{\lrp}[1]{\left( #1 \right)}
\newcommand{\lrcb}[1]{\left\{ #1 \right\}}
\newcommand{\lrsp}[1]{\left\langle #1 \right\rangle}
\newcommand{\lrabs}[1]{\left| #1 \right|}
\newcommand{\lrnorm}[1]{\left\| #1 \right\|}

\makeatletter
\newcommand{\degree}{\operatorname{deg}}

\newcommand{\indic}{\mathds{1}}

\newcommand{\Beta}[1][\@nil]{%
  \def\tmp{#1}%
   \ifx\tmp\@nnil%
   		\operatorname{Beta}%
    \else%
      \operatorname{Beta}\lrp{#1}%
    \fi}

\newcommand{\Gama}[1][\@nil]{%
  \def\tmp{#1}%
   \ifx\tmp\@nnil%
   		\Gamma%
    \else%
      \Gamma\lrp{#1}%
    \fi}
\newcommand{\Gauss}[1][\@nil]{%
  \def\tmp{#1}%
   \ifx\tmp\@nnil%
   		\calN%
    \else%
      \calN\lrp{#1}%
    \fi}
\newcommand{\Dir}[1][\@nil]{%
  \def\tmp{#1}%
   \ifx\tmp\@nnil%
   		\operatorname{Dir}%
    \else%
      \operatorname{Dir}\lrp{#1}%
    \fi}

\makeatother

\renewcommand{\top}{\mathsf{\scriptscriptstyle T}}
\makeatletter
\newcommand{\Tr}[1][\@nil]{%
  \def\tmp{#1}%
   \ifx\tmp\@nnil%
   		\operatorname{Tr}%
    \else%
      \operatorname{Tr}\lrb{#1}%
    \fi}
\newcommand{\rank}[1][\@nil]{%
  \def\tmp{#1}%
   \ifx\tmp\@nnil%
   		\operatorname{rank}%
    \else%
      \operatorname{rank}\lrp{#1}%
    \fi}
\newcommand{\diag}[1][\@nil]{%
  \def\tmp{#1}%
   \ifx\tmp\@nnil%
   		\operatorname{diag}%
    \else%
      \operatorname{diag}\lrp{#1}%
    \fi}
\makeatother

\begin{document}

\begin{frontmatter}

\title{Fast sampling from $\beta$-ensembles}
\runtitle{Fast sampling from $\beta$-ensembles}

\begin{aug}
    \author{%
    \fnms{Guillaume}
    \snm{Gautier$^{*,}$}%
        \thanksref{a,b}%
    }
    \author{%
    \fnms{Rémi}
    \snm{Bardenet}%
        \thanksref{a}%
    }
    \and
    \author{%
    \fnms{Michal}
    \snm{Valko}%
        \thanksref{b,c}%
    }
    \address[a]{Univ.\,Lille, CNRS, Centrale Lille, UMR 9189 - CRIStAL, 59651 Villeneuve d'Ascq, France.\\[-0.4em]}
    \address[b]{INRIA Lille - Nord Europe,  40, avenue Halley 59650, Villeneuve d'Ascq, France.\\[-0.4em]}
    \address[c]{DeepMind Paris, 14 Rue de Londres, 75009  Paris, France.\\[0.2em]}

    E-mails:
    $^*$\href{mailto:guillaume.gga@gmail.com}{guillaume.gga@gmail.com};
    \href{mailto:remi.bardenet@gmail.com}{remi.bardenet@gmail.com};
    \href{mailto:valkom@deepmind.com}{valkom@deepmind.com}.

    \runauthor{G.\,Gautier, R.\,Bardenet, and M.\,Valko}
\end{aug}

\begin{abstract}
    We study sampling algorithms for $\beta$-ensembles with time complexity less than cubic in the cardinality of the ensemble.
    Following \cite{DuEd02}, we see the ensemble as the eigenvalues of a random tridiagonal matrix, namely a random Jacobi matrix.
    First, we provide a unifying and elementary treatment of the tridiagonal models associated to the three classical Hermite, Laguerre and Jacobi ensembles.
    For this purpose, we use simple changes of variables between successive reparametrizations of the coefficients defining the tridiagonal matrix.
    Second, we derive an approximate sampler for the simulation of $\beta$-ensembles, and illustrate how fast it can be for polynomial potentials.
    This method combines a Gibbs sampler on Jacobi matrices and the diagonalization of these matrices.
    In practice, even for large ensembles, only a few Gibbs passes suffice for the marginal distribution of the eigenvalues to fit the expected theoretical distribution.
    When the conditionals in the Gibbs sampler can be simulated exactly, the same fast empirical convergence is observed for the fluctuations of the largest eigenvalue.
    Our experimental results support a conjecture by \citet[][]{KrRiVi16}, that the Gibbs chain on Jacobi matrices of size $N$ mixes in $\mathcal{O}(\log N)$.
\end{abstract}



\end{frontmatter}

\begingroup
\allowdisplaybreaks


\section{Introduction} 
\label{sec:introduction}

    $\beta$-ensembles are probability distributions of the form
    \begin{equation}
    \label{eq:joint_distribution_beta_ensemble}
        \lrabs{\Delta(x_1,\dots,x_N)}^{\beta}
            ~ Z^{-1} \prod_{n=1}^N e^{-V(x_n)} \diff x_n, \quad x_1,\dots,x_N \in \bbR,
    \end{equation}
    where $\Delta(x_1,\dots,x_N) = \prod_{i<j}(x_j-x_i)$ is the Vandermonde determinant, $\beta>0$ is akin to an inverse temperature in statistical physics, and $V:\bbR\rightarrow \bbR$ is called the \emph{potential}.
    Loosely speaking, one can think of \eqref{eq:joint_distribution_beta_ensemble} as representing the position of $N$ particles living on the real line, confined by the potential $V$, and repelling each other through the Vandermonde determinant.
    As this interpretation suggests, $\beta$-ensembles arise as models in statistical physics \citep[Chapters 1 to 3]{For10}.
    They are also famous as the distribution of the eigenvalues of some of the classical models of random matrices.
    The particular values $\beta\in\{1,2,4\}$ respectively appear when considering specific random matrices with real, complex, or quaternionic Gaussian entries; see e.g., \citet{For10} again or \citet*[Chapter 4]{AnGuZe09}.

    The case $\beta=2$ is of particular interest, since the distribution of $\{x_1,\dots,x_N\}$ then becomes a particular determinantal point process (DPP), called an \emph{orthogonal polynomial ensemble} \citep[OPE,][]{Kon05}.
    Originally introduced as models in fermionic optics by \citet{Mac75}, DPPs are comparatively easier to analyze than other repulsive distributions.
    Moreover, there exists a generic algorithm to sample from DPPs \citep{HKPV06}.
    Together with their analytic tractability, the existence of sampling algorithms has sparked the study of Monte Carlo integration using an OPE for the quadrature nodes \citep*{BaHa19,GaBaVa19b,BeBaCh19}.

    Besides Monte Carlo integration, numerical procedures to generate samples from $\beta$-ensembles are also needed to establish conjectures in statistical physics or random matrix theory.
    For instance, using a tailored version of the generic DPP sampler of \citet{HKPV06}, \citet{OlNaTr15} explore so-called \emph{universality properties} in random matrix theory, and make conjectures on the law of $\max x_i$ when $\beta=2$ and $V$ is a polynomial of degree $4$.
    \citet{ChFe18} rather use Hamiltonian Monte Carlo to approximately sample from various \emph{Coulomb gases}, including \eqref{eq:joint_distribution_beta_ensemble} with $\beta=2$ and $V(x)=x^4/4$, and investigate their limiting features when $N\rightarrow\infty$.
    From a different perspective, \citet{LiMe13} view \eqref{eq:joint_distribution_beta_ensemble} as the equilibrium distribution for the Dyson Brownian motion associated to the potential $V$.
    When $\beta=2$, they generate approximate samples by discretizing the corresponding stochastic differential equation.

    Algorithms to sample from $\beta$-ensembles come in three different guises, which we describe in increasing order of complexity.
    First, when $\beta>0$ and $V$ is the negative logarithm of a Gaussian, gamma, or beta pdf, we speak of the Hermite, Laguerre, and Jacobi $\beta$-ensemble, respectively.
    \citet{DuEd02} showed that the Hermite and Laguerre $\beta$-ensembles can be characterized as the eigenvalue distribution of a random tridiagonal matrix with easy-to-sample independent entries.
    This gives a $\calO(N^2)$ sampling algorithm.
    \citet{DuEd02} expected the same to hold for the Jacobi $\beta$-ensemble, which was later proved by \citet{KiNe04}.

    Second, when $\beta=2$, the generic projection DPP sampler of \citet{HKPV06} applies.
    That there actually exists an exact sampler is maybe surprising, and it is a particular feature of DPPs among interacting particle systems.
    The procedure remains costly, though.
    It has at least a cubic cost in $N$, with the total cost further depending on rejection sampling subroutines, the cost of which is case-dependent and has been left uninvestigated.
    Additionally, it is required in this procedure to numerically evaluate the first $N$ orthonormal polynomials $p_k, k=0,\dots,N-1$ with respect to $e^{-V(x)}\diff x$.
    This is traditionally done using their recurrence relation
    \begin{equation}
        \label{eq:intro_recurrence_relation_orthogonal_polynomials_wrt_reference_measure}
        \sqrt{b_{k-1}} p_{k-1}(x)
        + a_kp_k(x)
        + \sqrt{b_k} p_{k+1}(x)
        = xp_k(x),
    \end{equation}
    see e.g., \citet{Gau04}.
    In the Hermite, Laguerre, and Jacobi case, the recurrence coefficients $a_k,b_k$ are known, but as we just saw, these three cases are already covered by a computationally more efficient tridiagonal matrix model.
    When the coefficients in \eqref{eq:intro_recurrence_relation_orthogonal_polynomials_wrt_reference_measure} are not known, one can either rely on the Stieltjes algorithm \citep[Section 2.2]{Gau04} or numerically solve a Riemann-Hilbert problem \citep{Olv11}.
    The latter is theoretically only an $\calO(N)$ overcost.

    A third algorithm is Markov chain Monte Carlo \citep[MCMC, see e.g.,][]{RoCa04}, which is in principle valid for any $\beta>0$ and any $V$ that gives a well-defined distribution in \eqref{eq:joint_distribution_beta_ensemble}.
    MCMC only requires to evaluate the pdf in \eqref{eq:joint_distribution_beta_ensemble} pointwise and up to a constant, but it only delivers approximate samples of \eqref{eq:joint_distribution_beta_ensemble}, in the sense that it outputs a sample from a Markov chain with \eqref{eq:joint_distribution_beta_ensemble} as its limiting distribution.
    The issue is that the performance of MCMC samplers --~the mixing time of the Markov chain~-- deteriorates when $N\gg 1$, which is typically the regime of interest for conjectures in random matrix theory or statistical physics.
    Hybric Monte Carlo \citep[HMC,][]{DKPR87, Nea11} is an MCMC sampler that has demonstrated good mixing in high-dimensional problems, provided one can evaluate the gradient of the pdf in \eqref{eq:joint_distribution_beta_ensemble}.
    For $\beta$-ensembles with $\beta=2$, \citet{ChFe18} provide empirical evidence that the output of HMC successfully reproduces known limiting features of the large $N$ regime, and they raise new conjectures.
    The main limitation of this approach is the large number of MCMC iterations required by HMC: \citet{ChFe18} require at least $10^4$ iterations and are restricted to $N\leq 50$.

    In this paper, we further investigate fast samplers of $\beta$-ensembles.
    Our contributions are twofold.
    First we gather existing tools from different communities to give an elementary proof of the tridiagonal models for the Hermite, Laguerre, and Jacobi $\beta$-ensembles.
    This proof crucially relies on successive reparametrizations of the recurrence coefficients in \eqref{eq:intro_recurrence_relation_orthogonal_polynomials_wrt_reference_measure} and unifies the treatment of tridiagonal models for the three classical $\beta$-ensembles, pioneered with two different methods by \citet{DuEd02} and \citet{KiNe04}.
    We take no credit for the originality of the proof: the credit should go -- among others cited below -- to \cite{DeNa12}, who studied distributions on the space of moments, and recognized these three $\beta$-ensembles as corresponding to natural distributions over moments.
    We rather take credit for a stand-alone and elementary version of this unifying proof, using only basic facts on orthogonal polynomials and linear algebra.

    Our second contribution is an MCMC sampler that applies to polynomial potentials. For $V$ of degree at most 6, we give experimental evidence that the resulting Markov chain mixes extremely fast, which confirms an intuition of \citet[Section 2]{KrRiVi16}.
    On a variety of potentials, we demonstrate that our simple
    Gibbs Markov kernel yields a much cheaper (although approximate) sampler than the exact procedure of \citet[for $\beta=2$]{HKPV06}. Importantly, our Markov kernel outperforms the HMC approach of \citet{ChFe18} in the particular case of $\beta$-ensembles.
    To give an idea, we are able to reproduce known features of \eqref{eq:joint_distribution_beta_ensemble} for values of $N$ in the hundreds, using only a few Gibbs sweeps, totaling a few seconds on a modern laptop: it takes roughly $10$s for $N=200$ points and less than a minute for $N=1000$ points.
    That such a basic Gibbs kernel can outperform HMC may seem surprising.
    The key is that we exploit the structure of $\beta$-ensembles by defining a Markov chain on the recurrence coefficients of orthogonal polynomials.
    These recurrence coefficients are defined similarly to \eqref{eq:intro_recurrence_relation_orthogonal_polynomials_wrt_reference_measure}, but this time using the orthogonal polynomials with respect to a random discrete measure, the support of which is the $\beta$-ensemble.
    Intuitively, in that new parametrization, the interaction between variables is short-range compared to the interaction among particles in \eqref{eq:joint_distribution_beta_ensemble}, and Gibbs sampling thus becomes easier.
    In this sense, our MCMC kernel extends the tridiagonal models of the three classical $\beta$-ensembles.
    Finally, we note that all experiments can be reproduced using our \DPPy\, toolbox \citep*[][\href{https://github.com/guilgautier/DPPy}{https://github.com/guilgautier/DPPy}]{GPBV19}, which features all samplers described here.

    The rest of the paper is organized as follows.
    In \Secref{sec:tridiagonal_models_for_the_three_classical_beta_ensembles}, we survey existing results on tridiagonal models for $\beta$-ensembles.
    Known exact sampling results actually take the form of diagonalizing random Jacobi matrices, that is, tridiagonal matrices whose coefficients are the recurrence coefficients of a sequence of orthogonal polynomials.
    We introduce the necessary background on orthogonal polynomials in \Secref{sec:atomic_measures_moments_and_jacobi_matrices}.
    In \Secref{sec:making_the_change_of_variable}, we perform the change of variables between the points of a $\beta$-ensemble augmented with weights and the entries of a Jacobi matrix.
    In \Secref{sec:proving_the_three_classical_tridiagonal_models}, we give an elementary proof of the known results on tridiagonal models.
    Finally, in \Secref{sec:gibbs_sampling}, we demonstrate a simple MCMC scheme based on a Gibbs kernel, to sample Jacobi matrices corresponding to $\beta$-ensembles with polynomial potentials.


\section{Classical $\beta$-ensembles and their tridiagonal models} 
\label{sec:tridiagonal_models_for_the_three_classical_beta_ensembles}

    The Hermite, Laguerre and Jacobi $\beta$-ensembles were originally defined for $\beta\in \lrcb{1, 2, 4}$, as the eigenvalue distribution of some random full matrices; see e.g., \cite{AnGuZe09}.
    The latter matrices are symmetrizations of matrices filled with i.i.d.\,real, complex, or quarternionic Gaussian variables when $\beta$ is respectively $1,2,$ and $4$.
    In this section, we recall the seminal results of \citet{DuEd02} and \citet{KiNe04} regarding the construction of real-symmetric tridiagonal random matrices, whose eigenvalues follow the classical Hermite, Laguerre and Jacobi $\beta$-ensembles.
    These results actually allow any $\beta \in (0, +\infty)$, and can be interpreted as samplers with $\calO(N^2)$ time complexity, by simply diagonalizing the proposed tridiagonal matrices.

    We note $\bfa \triangleq  \lrp{a_{1}, \dots, a_{N}} \in \mathbb{R}^N$, $\bfb \triangleq \lrp{b_{1}, \dots, b_{N-1}} \in(0,+\infty)^{N-1}$, and define the tridiagonal matrix

    \begin{equation}
    \label{eq:jacobi_matrix_a_b}
        \Jab
        \triangleq
        \begin{bmatrix}
        a_{1}           & \sqrt{b_{1}}  &                   & (0)           \\
        \sqrt{b_{1}}    & a_{2}         & \ddots            &               \\
                        & \ddots        & \ddots            & \sqrt{b_{N-1}}\\
        (0)             &               & \sqrt{b_{N-1}}    & a_{N}
        \end{bmatrix}.\\[1em]
    \end{equation}
    Such a matrix is called a \emph{Jacobi matrix}.
    As we will see in \Secref{sec:atomic_measures_moments_and_jacobi_matrices}, Jacobi matrices naturally arise in the study of orthogonal polynomials.

    To build the random tridiagonal matrix model for the Hermite $\beta$-ensemble, \citet{DuEd02} started from to the original random full matrix model defining the Hermite ensemble with $\beta=1$.
    More specifically, they considered the symmetric part of a random matrix filled with i.i.d. unit Gaussians, and applied Householder transformations to reduce it to tridiagonal form, as in, e.g., \citet[Section 5.4.8]{GoVL13}.

    \begin{theorem}[\citealp{DuEd02}, II C, for $\mu=0$ and $\sigma=1$]
    \label{th:hermite_beta_ensemble}
        The Hermite $\beta$-ensemble, defined as \eqref{eq:joint_distribution_beta_ensemble} with potential $V(x) = \frac{1}{2\sigma^2}(x - \mu)^2$, corresponds to the eigenvalue distribution of the tridiagonal matrix $J_{\bfa, \bfb}$ in \eqref{eq:jacobi_matrix_a_b}, with entries drawn independently as
        \begin{equation}
            \label{eq:hermite_beta_ensemble_param_distrib}
            a_n \sim \Gauss[\mu, \sigma^2],
            \quad
            \text{and}
            \quad
            b_n \sim \Gama[\frac{\beta}{2}(N-n), \sigma^2].
        \end{equation}
    \end{theorem}

    For the Laguerre $\beta$-ensemble, \citet{DuEd02} used the same linear algebra techniques starting from the original full matrix model defining the Laguerre $\beta$-ensemble for $\beta=1$.
    The latter corresponds to the eigenvalue distribution of the covariance matrix $XX^{\top}$ of i.i.d. $\Gauss[0,I]$ vectors.
    More specifically, they reduced the matrix $X$ to bidiagonal form, see, e.g., \citet[Section 8.3.1]{GoVL13}.

    \begin{theorem}[\citealp{DuEd02}, III B, for $k=\frac{\beta}{2}(M-N+1)$ and $\theta=2$]
    \label{th:laguerre_beta_ensemble}
        \text{ }\\
        The Laguerre $\beta$-ensemble, defined as \eqref{eq:joint_distribution_beta_ensemble} with potential $V(x)= -(k-1)\log(x) + \frac{x}{\theta}$, corresponds to the eigenvalue distribution of the tridiagonal matrix $J_{\bfa, \bfb}$ in \eqref{eq:jacobi_matrix_a_b} parametrized by
        \begin{equation}
        \label{eq:laguerre_beta_ensemble_param}
            \begin{aligned}
                &a_1
                    = \xi_1,
                \quad\text{and}\quad
                a_{n}
                    = \xi_{2n-2}+\xi_{2n-1},
                \quad\text{for } 2\leq n \leq N,
                \quad\text{and}\\
                &b_{n}
                    = \xi_{2n-1}\xi_{2n},
                \quad\text{for } 1\leq n \leq N-1,
            \end{aligned}
        \end{equation}
        with independent coefficients
        \begin{equation}
        \label{eq:laguerre_beta_ensemble_param_distrib}
            \begin{aligned}
                \xi_{2n-1} &\sim \Gama[\frac{\beta}{2}(N-n)+k, \theta],
                    \quad \text{and} \quad
                \xi_{2n} &\sim \Gama[\frac{\beta}{2}(N-n), \theta].
            \end{aligned}
        \end{equation}
    \end{theorem}

    \cite{DuEd02} left the construction of a tridiagonal model for the Jacobi $\beta$-ensemble as an open problem.
    \citet{KiNe04} found such a model as a byproduct of their study the Circular $\beta$-ensemble.
    The latter ensemble is originally defined, for $\beta\in \lrcb{1, 2, 4}$, as the eigenvalue distribution of orthogonal, unitary and symplectic matrices drawn uniformly at random from the corresponding Haar measures.
    First, \citet{KiNe04} applied Householder transformations to reduce to quindiagonal form a unitary matrix drawn uniformly at random.
    Second, they projected the resulting eigenvalues onto the real line to obtain the tridiagonal model for the Jacobi $\beta$-ensemble.

    \begin{theorem}[\citealp{KiNe04}, Theorem 2]
    \label{th:jacobi_beta_ensemble}
        The Jacobi $\beta$-ensemble, defined as \eqref{eq:joint_distribution_beta_ensemble} with potential $V(x)= -[(a-1)\log(x) + (b-1)\log(1-x)]$, corresponds to the eigenvalue distribution of the tridiagonal matrix $J_{\bfa, \bfb}$ in \eqref{eq:jacobi_matrix_a_b} parametrized by
        \begin{equation}
        \label{eq:jacobi_beta_ensemble_param}
            \begin{aligned}
                a_1
                    &= c_1,
                &a_{n}
                    = (1-c_{2n-3})c_{2n-2} + (1-c_{2n-2})c_{2n-1},
                \quad\text{for } 2\leq n \leq N,\\
                b_1
                    &= c_1(1-c_1)c_2,
                &b_{n}
                    = (1-c_{2n-2})c_{2n-1}(1-c_{2n-1})c_{2n},
                \quad\text{for } 2\leq n \leq N-1,
            \end{aligned}
        \end{equation}
        with independent coefficients
        \begin{equation}
        \label{eq:jacobi_beta_ensemble_param_distrib}
            \begin{aligned}
                c_{2n-1}
                    &\sim \Beta[\frac{\beta}{2}(N-n)+a,
                                \frac{\beta}{2}(N-n)+b],
                    \quad\text{and}\\
                c_{2n}
                    &\sim \Beta[\frac{\beta}{2}(N-n),
                                \frac{\beta}{2}(N-n-1)+a+b].
            \end{aligned}
        \end{equation}
    \end{theorem}

    Observe how the stars align for these three special $\beta$-ensembles: Hermite, Laguerre, and Jacobi. The coefficients in successive parameterizations of the Jacobi matrix $\Jab$ are independent with easy-to-sample distributions.
    From a practical point of view, for any $\beta>0$, the computation of the eigenvalues of these random real-symmetric tridiagonal matrices can be seen as a $\mathcal{O}(N^2)$ sampler for each of the model; see \citet{CoRo13} for practical approaches to diagonalizing such matrices that can even run in quasi-linear time.

    Studying distributions over the space of moments, \citet{DeNa12} elegantly derived the three classical tridiagonal models as the supports of random atomic measures corresponding to natural moment distributions.
    On our side, we provide a unified treatment of these three classical models using a more pedestrian, sampling-motivated approach. To do this, we consider an atomic measure $\mu=\sum_{n=1}^{N} \omega_n \delta_{x_n}$, whose support points are distributed as a target $\beta$-ensemble, and take the Jacobi matrix $\Jab$ in \eqref{eq:jacobi_matrix_a_b} with coefficients the recurrence coefficients \eqref{eq:intro_recurrence_relation_orthogonal_polynomials_wrt_reference_measure} of the orthonormal polynomials w.r.t.\;$\mu$.
    We shall see in \Secref{sec:atomic_measures_moments_and_jacobi_matrices} that the recurrence coefficients are a suitable reparametrization of the atomic measure $\mu$.
    In particular, the support of $\mu$ actually coincides with the eigenvalues of $\Jab$, so that a tridiagonal model for the support of $\mu$ follows from knowing how to randomize $\Jab$.

    The first step of our proof will be to rederive \Thref{th:joint_distribution_recurrence_coefficients_Tr_J}, which allows changing variables from the nodes and weights of an atomic measure $\mu$ to the recurrence coefficients defining the Jacobi matrix $\Jab$.
    Note that the specific choice of distribution on the weights is simply of mathematical convenience.
    \begin{theorem}[\citealp{KrRiVi16}, Proposition 2]
    \label{th:joint_distribution_recurrence_coefficients_Tr_J}
        Consider a random atomic measure $\mu=\sum_{n=1}^{N} \omega_n \delta_{x_n}$, with nodes and weights independently distributed according to a $\beta$-ensemble with potential $V$ \eqref{eq:joint_distribution_beta_ensemble} and a Dirichlet $\Dir[\beta/2]$, respectively.
        Otherly put, the joint distribution of $(x_{1}, \dots, x_{N}, w_{1}, \dots, w_{N})$ is proportional to
        \begin{equation}
        \label{eq:joint_distribution_x_w}
            \lrabs{\Delta \lrp{x_{1}, \dots, x_{N}}}^{\beta}
            \e^{- \sum\limits_{n=1}^{N} V(x_n)}
            \diff x_{1:N}
            \prod_{n= 1}^N
                w_n^{\frac{\beta}{2}-1}
                \indic_{w_n \geq 0}
                \indic_{\sum_{n=1}^{N} w_n = 1}
            \diff w_{1:N-1}.
        \end{equation}
        Then, the recurrence coefficients $(a_{1}, \dots, a_{N}, b_{1}, \dots, b_{N-1})$ of $\mu$ have joint distribution proportional to
        \begin{equation}
        \label{eq:joint_distribution_a,b_Tr-J}
            \prod_{n=1}^{N-1}
                b_n^{\frac{\beta}{2}(N-n)-1}
            \e^{- \Tr[V(\Jab)]}
                \diff a_{1:N}
                \diff b_{1:N-1}.
        \end{equation}
    \end{theorem}
    In \Secref{sec:making_the_change_of_variable}, we first re-prove that the change of variables underlying \Thref{th:joint_distribution_recurrence_coefficients_Tr_J} is valid. Then, in \Secref{sec:proving_the_three_classical_tridiagonal_models}, we obtain the three classical tridiagonal models of Theorems~\ref{th:hermite_beta_ensemble}, \ref{th:laguerre_beta_ensemble}, and \ref{th:jacobi_beta_ensemble} as instances of this result, using further smart-but-simple changes of variables.
    Before delving into the proof, we first survey how Jacobi matrices naturally appear in the theory of orthogonal polynomials.



\section{Atomic measures, moments and Jacobi matrices} 
\label{sec:atomic_measures_moments_and_jacobi_matrices}

    Throughout this section, we let $\mu = \sum_{n=1}^N w_n \delta_{x_n}$ be a discrete probability measure on $\bbR$ with $N$ distinct atoms $x_{1}, \dots, x_{N}$ and positive weights $\omega_{1}, \dots, \omega_{N}$.
    We further denote its moments by
    $$ m_k \triangleq \sum_{n=1}^N w_n x_n^k,\quad k\geq 0.$$

    \subsection{Orthogonal polynomials and Jacobi matrices} 
    \label{sub:orthogonal_polynomials_jacobi_matrices}

        This section closely follows \citet[Section 1.3]{Sim11}, to which we refer for details.
        Applying the Gram-Schmidt procedure in $L^2(\mu)$ to the monomials $(x\mapsto x^k)_{k=0}^{N-1} $ yields monic polynomials $(P_k)_{k=0}^{N-1}$ with $\deg P_k=k$ and
        \begin{equation}
        \label{eq:orthogonality_atomic_measure}
            \lrsp{P_k, P_\ell}_{\mu}
            \triangleq
                \sum_{n=1}^{N} w_n P_k(x_n) P_\ell(x_n) = 0,
                \quad k\neq \ell.
        \end{equation}
        These polynomials are called the \emph{monic orthogonal polynomials} (monic OPs, in short) with respect to $\mu$.
        We \emph{define} the $N$-th monic OP as
        $$
        P_N(x) =\prod_{n=1}^{N} (x-x_n).
        $$
        Since $\Vert P_N \Vert_\mu \triangleq  \lrsp{P_N, P_N}_{\mu} = 0$, $P_N$ is the zero vector of $L^2(\mu)$: it is orthogonal to all $P_k$ with $k\leq N-1$.

        Furthermore, for any $n<N$, since $\lrsp{xP_n,P_k}_\mu = \lrsp{P_n,xP_k}_\mu = 0$ for $k< n-1$, the polynomial $xP_n$ can be uniquely expressed using only $P_{n-1}$, $P_n$ and $P_{n+1}$.
        This is usually phrased as follows.
        The monic OPs satisfy a three-term recurrence relation involving two sequences of \emph{recurrence coefficients}, namely
        \begin{gather}
        \label{eq:3-terms_recurrence_monic_orthogonal_poly}
        \begin{aligned}
            &P_{- 1} \equiv 0, P_0 \equiv 1 \text{ and} \\
            &x P_{n}(x) = b_{n}P_{n-1}(x) + a_{n+1}P_{n}(x) + P_{n+1}(x),
            \quad \forall 0 \leq n < N,
        \end{aligned}
        \end{gather}
        where $\bfa = a_{1:N} = \lrp{a_n} \in \mathbb{R}^N$, and $\bfb = b_{1:N-1} =\lrp{b_n}\in(0,+\infty)^{N-1}$.
        These relations can be written in matrix form as
        \begin{equation}
        \label{eq:3-terms_recurrence_monic_orthogonal_poly_matrix}
            \begin{bmatrix}
                a_{1}   & 1             &                   & (0)   \\
                b_{1}   & a_{2}         & \ddots            &       \\
                        & \ddots        & \ddots            & 1     \\
                (0)     &               & b_{N-1}           & a_{N}
            \end{bmatrix}
            \begin{bmatrix}
                P_0(x)        \\
                \vdots        \\
                P_{N-2}(x)    \\
                P_{N-1}(x)
            \end{bmatrix}
            = x
            \begin{bmatrix}
                P_0(x)        \\
                \vdots        \\
                P_{N-2}(x)    \\
                P_{N-1}(x)
            \end{bmatrix}
            -
            \begin{bmatrix}
                0           \\
                \vdots      \\
                0           \\
                P_{N}(x)
            \end{bmatrix}.
        \end{equation}
        From \eqref{eq:3-terms_recurrence_monic_orthogonal_poly_matrix}, it is clear that the roots of $P_N$ are also eigenvalues of the tridiagonal matrix $\tildeJab$ appearing on the left-hand side.
        Roots and eigenvalues actually coincide since $P_N$ has $N$ distinct roots by definition.

        A lot more can be said on the links between OPs and their recurrence coefficients.
        For instance, \Proporef{propo:norms_of_Pk} will be of use later on.

        \begin{proposition}
        \label{propo:norms_of_Pk}
        The squared norms of the monic polynomials $(P_n)_{n=0}^{N-1}$ can be expressed as
        \begin{equation}
        \label{eq:monic_square_norm}
            \lrnorm{P_0}_{\mu}^2 = 1,
                \quad
            \lrnorm{P_k}_{\mu}^2 = \prod_{n=1}^k b_n,
                \quad \forall 1\leq k \leq N-1.
        \end{equation}
        \end{proposition}

        \begin{proof}
            For $k=0$, $\lrnorm{P_0}_{\mu}^2 = \sum_{n=1}^N w_n = 1$.
            Then, for any $1\leq k\leq N-1$,
            \begin{align*}
                \lrsp{\eqref{eq:3-terms_recurrence_monic_orthogonal_poly},
                  P_{k-1}}_{\mu}
                &\Longleftrightarrow
                \lrsp{x P_{k}, P_{k-1}}_{\mu}
                = \lrsp{b_{k}P_{k-1}, P_{k-1}}_{\mu}\\
                &\Longleftrightarrow
                \lrsp{P_{k}, x P_{k-1}}_{\mu}
                  = b_{k} \lrsp{P_{k-1}, P_{k-1}}_{\mu} \\
                &\Longleftrightarrow
                \lrsp{P_{k}, x^n}_{\mu}
                  = b_{k} \lrnorm{P_{k-1}}_{\mu}^2 \\
                &\Longleftrightarrow
                \lrnorm{P_{k}}_{\mu}^2
                  = b_{k} \lrnorm{P_{k-1}}_{\mu}^2,
            \end{align*}
            and a simple recursion provides $\lrnorm{P_{k}}_{\mu}^2 = \prod_{n=1}^k b_n>0$.
        \end{proof}

        Denoting by $D = \diag\lrp{\lrnorm{P_0}, \dots, \lrnorm{P_{N-1}}}$, \Proporef{propo:norms_of_Pk} yields $\Jab = D^{-1} \tildeJab D$, where we recall that the Jacobi matrix $\Jab$ was defined in \eqref{eq:jacobi_matrix_a_b}.
        This yields the following proposition.
        \begin{proposition}
          \label{propo:atoms_are_eigenvalues}
          The atoms of $\mu$, coincide with the eigenvalues of $\Jab$, where the coefficients of the matrix are taken to be the recurrence coefficients of the monic OPs with respect to $\mu$.
        \end{proposition}

        \Proporef{propo:atoms_are_eigenvalues} already gives a tentative $\calO(N^2)$ sampling algorithm for $\beta$-ensembles: find a distribution over Jacobi matrices such that the eigenvalues form the desired $\beta$-ensemble.
        This is precisely what the tridiagonal models of \citet{DuEd02} do; see \Thref{th:hermite_beta_ensemble}.
        To give a complete elementary proof, we need to perform a change of variables from the atoms and weights of $\mu$ to the recurrence coefficients.
        The rest of this section introduces the tools needed for this change of variables, which is then performed in \Secref{sec:making_the_change_of_variable}.

        So far, we have explained how to obtain a Jacobi matrix from an atomic measure with finite support.
        The reverse construction is also possible and elementary.
        This is called Favard's theorem for atomic measures with finite support.
        To save space and because our proof would be a simple copy of Simon's book, we only give a reference.
        We have used the same notation as Simon throughout this section, for ease of reference.
        \begin{theorem}[\citealp{Sim11}, Theorem 1.3.3]
            \label{th:favard}
            Let
            \begin{equation}
                \label{eq:def_atom_and_weight_spaces}
                \atomSpace
                    \triangleq
                        \lrcb{x_1,\dots,x_N \in \bbR
                                \mid x_{1} > \dots > x_{N}}
                \quad\text{and}\quad
                \weightSpace
                    \triangleq
                        \lrcb{\omega_{1}, \dots, \omega_{N-1} > 0
                                \mid \sum_{n=1}^{N-1} \omega_n < 1}.
            \end{equation}
            Favard's map
            \begin{equation}
            \label{eq:favard_map}
                \psi:
                \begin{array}{rcl}
                \bbR_>^{N} \times S_N
                    &\longrightarrow
                    &\bbR^N \times (0, +\infty)^{N-1}\\[0.2em]
                (x_{1:N}, w_{1:N-1})
                    &\longmapsto
                    &(a_{1:N}, b_{1:N- 1})
                \end{array}
            \end{equation}
            linking the nodes and weights of $\mu=\sum_{n=1}^{N} w_n \delta_{x_n}$ with the entries of the corresponding Jacobi matrix $\Jab$ defined in \eqref{eq:jacobi_matrix_a_b},
                is one-to-one and onto.
        \end{theorem}

        Note that whenever $w_{1:N-1}\in \weightSpace$, we always set $w_N = 1-\sum_{n=1}^{N-1} \omega_n$, so that $\mu$ is a probability measure.
        As a side remark, the weights $w_{1:N}$ of $\mu$ can also be expressed using evaluations of the monic OPs on the support of $\mu$ \citep[Proposition 1.3.1]{Sim11}: for all $n=1,\dots,N$,
        \begin{equation}
        \label{eq:weights_CD_kernel_K_N}
        w_n = \frac{1}{K_N(x_n, x_n)}\CommaBin
        \quad\text{with }
        K_N(x,y)
          = \sum_{k=0}^{N-1}
            \frac{P_k(x) P_k(y)}
               {\lrnorm{P_k}_{\mu}^2}\cdot
        \end{equation}
        These weights are reminiscent of Gaussian quadrature \citep[Section 1.4.2]{Gau04}, where the OPs are usually w.r.t.
        a non-atomic measure.

    \subsection{Orthogonal polynomials and moments} 
    \label{sub:orthogonal_polynomials_moments_and_spectral_measures}

        We know from \Thref{th:favard} that the change of variables $\psi$ is a bijection.
        In order to prove that $\phi$ is a $C^1$-diffeomorphism and compute its Jacobian in \Secref{sec:making_the_change_of_variable}, we pause to introduce an intermediate parametrization through moments.
        Intuitively, the moments are responsible for the Vandermonde determinant in \eqref{eq:joint_distribution_beta_ensemble}.

        The monic orthogonal polynomials $(P_n)_{n=0}^N$ w.r.t.\,$\mu$ can also be expressed in terms of the moments $(m_k)$ of $\mu$.
        First, define the following \emph{moment matrices}, see, e.g., \citet[Equation 1.4.3]{DeSt97}.
        \begin{definition}
            \label{def:moment_matrices}
            Let
            \begin{align}
                \underline{H}_{2n}
                = \lrb{m_{i+ j}}_{i,j=0}^{n}
                &=  \begin{bmatrix}
                    m_0         & \cdots    & m_{n} \\
                    \vdots          & \rddots   & \vdots \\
                    m_{n}       & \cdots    & m_{2n}
                  \end{bmatrix}
                \label{eq:moment_matrix_H2N_under}\\
                \underline{H}_{2n+1}
                = \lrb{m_{i+j+1}}_{i,j=0}^{n}
                &=  \begin{bmatrix}
                    m_1     & \cdots    & m_{n+1}   \\
                    \vdots  & \rddots   & \vdots    \\
                    m_{n+1} & \cdots    & m_{2n+1}
                  \end{bmatrix}
                \label{eq:moment_matrix_H2N+1_under}\\
                \overline{H}_{2n+1}
                = \lrb{m_{i+j} - m_{i+j+1}}_{i,j=0}^{n}
                &=  \begin{bmatrix}
                    m_0 - m_1       & \cdots    & m_{n} - m_{n+1}\\
                    \vdots          & \rddots   & \vdots \\
                    m_{n} - m_{n+1} & \cdots    & m_{2n} - m_{2n+1}
                  \end{bmatrix}.
                \label{eq:moment_matrix_H2N+1_over}
            \end{align}
            where $H$ stands for \href{https://en.wikipedia.org/wiki/Hankel_matrix}{\emph{Hankel} matrix}.
        \end{definition}

        The determinant of the Vandermonde matrix
        \begin{equation}
        \label{eq:def_vandermonde_matrix}
            \Delta\lrp{x_{1}, \dots, x_{n}}
            \triangleq
            \begin{bmatrix}
            1           & \cdots    & 1         \\
            x_1         & \cdots    & x_n       \\
                    & \vdots    &           \\
            x_1^{n-1}   & \cdots    & x_n^{n-1}
            \end{bmatrix},
        \end{equation}
        which appears in the definition of $\beta$-ensembles \eqref{eq:joint_distribution_beta_ensemble}, comes out naturally when taking the determinant of moment matrices associated to discrete measures.

        \begin{lemma}
        \label{lem:moment_matrices_vandermonde}
            It holds that
            \begin{equation}
            \label{eq:moment_matrices_vandermonde_determinant_first}
            \lrabs{\underline{H}_{2n-2}}
              \begin{cases}
                >0,
                  &\text{for any } 1\leq n \leq N,\\
                = \lrabs{\Delta\lrp{x_{1}, \dots, x_{N}}}^2
                  \prod_{n=1}^{N} w_n,
                  & \text{for } n=N, \\
                = 0,
                  & \text{for } n>N.
              \end{cases}
            \end{equation}
            Moreover
            \begin{equation}
            \label{eq:moment_matrices_vandermonde_determinant_second}
            \lrabs{\underline{H}_{2N-1}}
              = \lrabs{\underline{H}_{2N-2}}
                \prod_{n=1}^{N} x_n
                \quad \text{and} \quad
                \lrabs{\overline{H}_{2N-1}}
              = \lrabs{\underline{H}_{2N-2}}
                \prod_{n=1}^{N} (1-x_n).
            \end{equation}
        \end{lemma}
        \begin{proof}
            For any $1\leq n\leq N$, the Cauchy-Binet formula yields
            \begin{align}
                \lrabs{\underline{H}_{2n-2}}
                &= \lrabs{\sum_{k=1}^N w_k x_k^{i+j}}_{i,j=0}^{n-1}
                \nonumber
                \\
                &=  \lrabs{
                \begin{bmatrix}
                    1           & \cdots    & 1         \\
                    x_1         & \cdots    & x_N       \\
                    \vdots      & \vdots    & \vdots    \\
                    x_1^{n-1}   & \cdots    & x_N^{n-1}
                \end{bmatrix}
                \begin{bmatrix}
                    w_1 &           &       \\
                        & \ddots    &       \\
                        &           & w_N
                \end{bmatrix}
                \begin{bmatrix}
                    1       & x_1       & \cdots    & x_1^{n-1}     \\
                    \vdots  & \vdots    & \cdots    & \vdots    \\
                    1       & x_N       & \cdots    & x_N^{n-1}     \\
                \end{bmatrix}}
                \label{eq:moment_matrix_H2n-2_factorization}
                \\
                &= \sum_{\lrcb{i_1,\dots,i_n}\subset [N]}
                    \lrabs{\Delta\lrp{x_{i_1}, \dots, x_{i_n}}}^2
                    \prod_{k=1}^n w_{i_k}
                    > 0.
                    \nonumber
            \end{align}
            The particular case $n=N$ yields
            \begin{equation*}
                    \lrabs{\underline{H}_{2N-2}}
                    = \lrabs{\Delta\lrp{x_{1}, \dots, x_{N}}}^2
                        \prod_{n=1}^N w_{n}.
            \end{equation*}
            In the same vein, the two other determinants are obtained starting from
            \begin{equation*}
            \lrabs{\underline{H}_{2N-1}}
              = \lrabs{\sum_{n=1}^N w_n x_n^{i+j} x_n}_{i,j=0}^{N-1}
            \quad\text{and}\quad
              \lrabs{\overline{H}_{2N-1}}
              = \lrabs{\sum_{n=1}^N w_n x_n^{i+j} (1-x_n)}_{i,j=0}^{N-1}.
            \end{equation*}
            For $n>N$, \eqref{eq:moment_matrix_H2n-2_factorization} clearly shows that $\underline{H}_{2n-2}$ is rank deficient.
        \end{proof}
        Moment matrices also provide an alternative description of orthogonal polynomials.
        \begin{proposition}
          The monic polynomials $\lrp{P_n}_{n=0}^{N}$ orthogonal with respect to $\mu = \sum_{n=1}^{N} w_n \delta_{x_n}$ admit the following expression
          \begin{equation}
          \label{eq:monic_Pn_moment_matrix}
            P_0=1
            \quad \text{and} \quad
            P_{n}(x)
            = \frac{1}{\lrabs{\underline{H}_{2n-2}}}
            \lrabs{
              \begin{array}{cccc}
              \multicolumn{3}{c}{
                \multirow{3}{*}{
                  \raisebox{0mm}{
                    \Large\mbox{$\underline{H}_{2n-2}$}
                  }
                }
              }
                & 1
              \\
                &
                &
                & \raisebox{2pt}{\vdots}
              \\
              m_{n}
                & \cdots
                & m_{2n-1}
                & x^{n}
              \end{array}
            },
            \quad \forall 1\leq n \leq N.
          \end{equation}
          Besides,
          \begin{equation}
          \label{eq:monic_square_norm_moments}
            \lrnorm{P_0}_{\mu}^2=1
            \quad \text{and} \quad
            \lrnorm{P_n}_{\mu}^2
            = \frac{\lrabs{\underline{H}_{2n}}}
                   {\lrabs{\underline{H}_{2n-2}}}
            \CommaBin
            \quad \forall 1\leq n \leq N.
          \end{equation}
          In particular, $P_N(x) = \prod_{n=1}^{N} (x-x_n)$ is the zero vector of $L^2(\mu)$.
        \end{proposition}

        \begin{proof}
            The previous \Lemref{lem:moment_matrices_vandermonde} validates the definition of $(P_n)_{n=0}^{N} $ as a sequence of monic polynomials with $\deg P_n = n$ since the denominator $\lrabs{\underline{H}_{2n-2}}>0$.
            They are also mutually orthogonal.
            To see this, let $1\leq n\leq N$, then
            \begin{equation*}
            \lrsp{P_{n}, x^k}_{\mu}
            =\frac{1}{\lrabs{\underline{H}_{2n-2}}}
                \begin{vmatrix}
                  m_0           & \cdots    & m_{n-1}   & m_k           \\
                  \vdots        &           & \vdots    & \vdots        \\
                  m_{n}         & \cdots    & m_{2n-1}  & m_{n+ k}
                \end{vmatrix}
              = 0,
              \quad \forall k < n.
            \end{equation*}
            Moreover, $\forall 1\leq n\leq N$,
            \begin{equation*}
                \lrnorm{P_n}_{\mu}^2
                = \lrsp{P_n,P_n}_{\mu}
                = \lrsp{P_n,x^n}_{\mu}
                = \frac{\lrabs{\underline{H}_{2n}}}
                     {\lrabs{\underline{H}_{2n-2}}}\cdot
            \end{equation*}
            Then, \Lemref{lem:moment_matrices_vandermonde}
            yields
            $\lrnorm{P_N}_{\mu}^2
            = \frac{\lrabs{\underline{H}_{2N}}}
                 {\lrabs{\underline{H}_{2N-2}}}
            = 0$.
            Thus, the distinct support points of $\mu$ are zeros of $P_N$.
            But the latter is monic with $\deg P_N=N$, hence $P_N = \prod_{n=1}^{N} (x-x_n)$.
        \end{proof}
        The next result further relates moment matrices and the recurrence coefficients.
        \begin{lemma}
        \label{lem:H_2N-2_Vandermonde_prod_b}
          The moment matrix $\underline{H}_{2N-2}$ associated to $\mu=\sum_{n=1}^N w_n \delta_{x_n}$ has determinant
          \begin{equation}
          \label{eq:H_2N-2_Vandermonde_prod_b}
              \lrabs{\underline{H}_{2N-2}}
            = \lrabs{\Delta\lrp{x_1, \dots, x_N}}^2
              \prod_{n=1}^N w_n
            = \prod_{n=1}^{N-1} b_{n}^{N-n}.
          \end{equation}
        \end{lemma}
        \begin{proof}
          The first equality was established in \Lemref{lem:moment_matrices_vandermonde}.
          The second results from a simple recursion combining Equations~\ref{eq:monic_square_norm} and \ref{eq:monic_square_norm_moments}.
          For any $1\leq k \leq N-1$
          \begin{equation}
          \label{eq:monic_square_norm_moments_prod_bn}
              \lrnorm{P_k}_{\mu}^2
                = \frac{\lrabs{\underline{H}_{2k}}}
                       {\lrabs{\underline{H}_{2k-2}}}
                = \prod_{n=1}^k b_n
              \implies
              \lrabs{\underline{H}_{2k}}
              = \prod_{\ell=1}^k \prod_{n=1}^\ell b_n
              = \prod_{n=1}^k b_n^{k+1-n}.
          \end{equation}
        \end{proof}
        From the point of view of sampling a $\beta$-ensemble, \Lemref{lem:H_2N-2_Vandermonde_prod_b} already hints what tridiagonal models can achieve: if we see the $\beta$-ensemble as the support of a random atomic measure, which is parametrized by its recurrence coefficients, then the complex interaction term that is the Vandermonde determinant in \eqref{eq:joint_distribution_beta_ensemble} gets replaced by a simple product of powers of $b_n$s.
        This intuition, formalized in \Thref{th:joint_distribution_recurrence_coefficients_Tr_J}, requires to make explicit the change of variables between the nodes and weights of $\mu$ and the recurrence coefficients of the corresponding Jacobi matrix $\Jab$.



\section{Making the change of variables} 
\label{sec:making_the_change_of_variable}

    To compute the Jacobian of Favard's map $(x_{1:N}, \omega_{1:N})\mapsto(a_{1:N}, b_{1:N-1})$, defined in \Thref{th:favard}, we first compute the Jacobian of the moment map $(x_{1:N}, \omega_{1:N})\mapsto(m_{1:2N-1})$, and then use the lattice path construction of \citet{Har18} to express the Jacobian of $(m_{1:2N-1})\mapsto (a_{1:N}, b_{1:N-1})$.
    We mention that the overall Jacobian has already been derived, in a more concise style, by \citet{FoRa06} and \citet{KrRiVi16}.
    Our contribution in this section is to give all details while remaining as elementary as possible.
    In particular, we only rely on Favard's theorem for atomic measures, and the proof of \Thref{th:joint_distribution_recurrence_coefficients_Tr_J} boils down to checking that the changes of variables are $C^1$-diffeomorphisms.

    Let $\phi:\atomSpace\times\weightSpace\rightarrow \mathbb{R}^{2N-1}$ map a set of $N$ distinct atoms and $N-1$ positive weights to their moments $(m_k)$.
    Let $\mathcal{M}\subset \mathbb{R}^{2N-1}$ be the image of $\phi$.
    \begin{proposition}[From atomic measures to moments]
        \label{propo:jacobian_x,w_moments}
        $\calM\subset \mathbb{R}^{2N-1}$ is open, $\phi$ is a $C^1$-diffeomorphism from $\atomSpace\times\weightSpace$ onto $\calM$, and
        \begin{equation}
            \label{eq:jacobian_x,w_moments}
            \lrabs{
                    \frac{  \partial m_{1:2N-1}
                        }{  \partial x_{1:N}, w_{1:N-1} }
                  }
            = \lrabs{\Delta(x_1,\dots,x_N)}^4
                \prod_{n=1}^N w_n
            = \frac{\lrabs{\underline{H}_{2N-2}}^2}
                   {\prod_{n=1}^N w_n},
        \end{equation}
        where the Hankel matrix $\underline{H}_{2N-2}$ is defined by \eqref{eq:moment_matrix_H2N_under}.
    \end{proposition}
    \begin{proof}
        Moments define monic OPs; see \Proporef{eq:monic_Pn_moment_matrix}.
        By Favard's \Thref{th:favard}, monic OPs in turn define the atoms and weights of $\mu$ uniquely.
        Thus, $\phi$ is injective.
        Moreover $\atomSpace\times\weightSpace\subset\mathbb{R}^{2N-1}$ is open, and $\phi$ is $C^1$.
        By the classical inverse function theorem, see e.g., \citet[Corollary 4.2.2]{Car71}, it is thus enough to show that the Jacobian of $\phi$ never vanishes.

        The $i$-th moment of $\mu$ can be written in two forms
        \begin{equation}
            m_i
            = \sum_{j= 1}^{N} w_j x_j^i
            = \sum_{j= 1}^{N- 1} w_j \lrp{x_j^i - x_N^i} + x_N^i,
        \end{equation}
        so that
        \begin{equation}
            \frac{\partial m_i}{\partial x_j}
            =
              iw_jx_j^{i-1}
            \quad
            \text{and}
            \quad
            \frac{\partial m_i}{\partial w_j}
            =
              x_j^{i} - x_N^{i}.
        \end{equation}
        Thus,
        \begin{align}
            \lrabs{
              \frac{    \partial m_{1:2N- 1}
                }{  \partial x_{1:N}, w_{1:N- 1} }
              }
            &= \lrabs{
            \lrb{
              \lrb{
                \frac{\partial m_i}{\partial x_j} ~
                \frac{\partial m_i}{\partial w_j}}_{j= 1}^{N- 1}
              \quad
              \lrb{
                \frac{\partial m_i}{\partial x_N}
                }
            }_{i= 1}^{2N- 1}
            }\nonumber\\
            &= \lrabs{
              \lrb{
                \lrb{
                  iw_jx_j^{i-1} ~~ x_j^{i} - x_N^{i}
                  }_{j= 1}^{N- 1}
                  \quad
                  iw_Nx_N^{i-1}
                }_{i= 1}^{2N- 1}
              }\nonumber\\
            &= \lrabs{
              \lrb{
                \lrb{
                  ix_j^{i-1} ~~ x_j^{i} - x_N^{i}
                  }_{j= 1}^{N- 1}
                  \quad
                  ix_N^{i-1}
                }_{i= 1}^{2N- 1}
              }
            \times \prod_{n=1}^N w_n
            \nonumber\\
            &= \lrabs{
              \lrb{
                \lrb{(i-1)x_j^{i-2}
                   ~~
                   x_j^{i-1} - x_N^{i-1}
                }_{j= 1}^{N- 1}
                \quad
                  (i-1)x_N^{i-2}
                  ~~
                  x_N^{i-1}
              }_{i= 1}^{2N}
            }
            \times \prod_{n=1}^N w_n
            \nonumber\\
            &= 
            \lrabs{
              \lrb{
                (i-1)x_j^{i-2} ~~ x_j^{i-1}
                }_{i=1,j= 1}^{2N,N}
              }
            \prod_{n=1}^N w_n.
            \label{e:confluentVandermonde}
        \end{align}
        The last equality is obtained by adding the last column to all other even columns.
        The determinant in \eqref{e:confluentVandermonde} is called a \emph{confluent} Vandermonde determinant.
        Its value is given, e.g., by \citet[Corollary 1, with $\eta_i \equiv 2$]{HaGi80}
        \begin{equation*}
          \lrabs{
              \lrb{
                (i-1)x_j^{i-2} ~~ x_j^{i-1}
                }_{i=1,j= 1}^{2N,N}
              }
          = \prod_{1\leq i<j \leq N} (x_j-x_i)^{2\times 2}
          = \lrabs{\Delta(x_1,\dots,x_N)}^4.
        \end{equation*}
        In particular, \eqref{e:confluentVandermonde} never vanishes on $\atomSpace\times \weightSpace$.
    \end{proof}


    Let us now consider the map
    \begin{equation}
        \label{eq:def_map_moment_recurrence_coef}
        \rho:\calM \rightarrow \bbR^N \times (0, +\infty)^{N-1},
    \end{equation}
    that takes moments $m_{1:2N-1}$ and returns the recurrence coefficients $(a_{1:N}, b_{1:N- 1})$.

    \begin{proposition}[From recurrence coefficients to moments]
        \label{propo:jacobian_a,b_moments}
        $\rho$ is a $C^1$-diffeomorphism from $\calM$ onto $\bbR^N \times (0, +\infty)^{N-1}$, and
        \begin{equation}
        \label{eq:jacobian_a,b_moments}
            \lrabs{
                    \frac{  \partial m_{1:2N-1}
                        }{  \partial a_{1:N}, b_{1:N-1} }
                  }
            = \prod_{n= 1}^{N- 1} b_{n}^{2(N-n)-1}
            = \frac{\lrabs{\underline{H}_{2N-2}}^2}
                   {\prod_{n= 1}^{N- 1} b_{n}},
        \end{equation}
        where the Hankel matrix $\underline{H}_{2N-2}$ is defined by \eqref{eq:moment_matrix_H2N_under}.
    \end{proposition}

    \begin{proof}
        Using \Thref{th:favard} and \Proporef{propo:jacobian_x,w_moments}, $\rho = \psi\circ \phi^{-1}$, so that $\rho$ is bijective.
        As in the proof of \Proporef{propo:jacobian_x,w_moments}, we apply the inverse function theorem \citep[Corollary 4.2.2]{Car71}, but this time to $\rho^{-1}$.
        We first note that $\bbR^N \times (0, +\infty)^{N-1}\subset \bbR^{2N-1}$ is open.
        It is thus enough to show that $\rho^{-1}$ is $C^1$ and that its Jacobian never vanishes.
        To this end, we borrow an elegant lattice path representation of the recurrence relations for OPs from \citet[Equation 1.8]{Har18}.
        This allows us to express the successive moments as polynomials in the recurrence coefficients.

        To provide intuition, we first compute the first few moments by hand, recursively applying the recurrence relation \eqref{eq:3-terms_recurrence_monic_orthogonal_poly}. It comes
        \begin{align}
            m_{1}
                = \lrsp{xP_0, P_0}
                &= 1    \cdot  \cancel{\lrsp{P_1, P_0}}
                  + a_{1} \cdot \lrsp{P_0, P_0}
                  + 0
                =\textcolor{red}{ a_1}\,,
            \nonumber\\[5pt]
            m_{2}
                = \lrsp{x^2P_0, P_0}
                &=
                        1       \cdot \lrsp{xP_1, P_0}
                        + a_{1} \cdot \lrsp{xP_0, P_0}
                        + 0\nonumber\\
                &= 1 \cdot
                    \lrp{
                        1       \cdot \cancel{\lrsp{P_2, P_0}}
                        + a_{2} \cdot \cancel{\lrsp{P_1, P_0}}
                        + b_{1} \lrsp{P_0, P_0}
                    }\nonumber\\
                   &\quad
                   +  a_1
                   \cdot \lrp{
                        1       \cdot \cancel{\lrsp{P_1, P_0}}
                        + a_{1} \cdot \lrsp{P_0, P_0}
                        + 0
                    }\nonumber\\
                &= \textcolor{blue}{1\cdot b_1} +  a_1\cdot a_1\,,
            \nonumber\\[5pt]
            m_{3}
                = \lrsp{x^3P_0, P_0}
                &=  1 \cdot \lrsp{x^2P_1, P_0}
                    + a_{1} \cdot \lrsp{x^2P_0, P_0}
                    + 0\nonumber\\
                &= 1
                    \cdot \lrp{
                        1       \cdot \lrsp{xP_2, P_0}
                        + a_{2} \cdot \lrsp{xP_1, P_0}
                        + b_{1}\cdot    \lrsp{xP_0, P_0}
                    }\nonumber\\
                &\quad
                   +  a_1
                    \cdot \lrp{
                        1       \cdot \lrsp{xP_1, P_0}
                        + a_{1} \cdot \lrsp{xP_0, P_0}
                        + 0
                    }\nonumber\\
                &=1\cdot1
                    \cdot
                    \lrp{
                        1       \cdot \cancel{\lrsp{P_3, P_0}}
                        + a_{3} \cdot \cancel{\lrsp{P_2, P_0}}
                        + b_{3} \cdot \cancel{\lrsp{P_1, P_0}}
                    }\nonumber\\
                &\quad
                  +1\cdot a_2
                    \cdot
                    \lrp{
                        1       \cdot \cancel{\lrsp{P_2, P_0}}
                        + a_{2} \cdot \cancel{\lrsp{P_1, P_0}}
                        + b_{1}\cdot    \lrsp{P_0, P_0}
                    }\nonumber\\
                &\quad
                  +1\cdot b_1
                    \cdot
                    \lrp{
                        1       \cdot \cancel{\lrsp{P_1, P_0}}
                        + a_{1} \cdot \lrsp{P_0, P_0}
                        + 0
                    }\nonumber\\
                &\quad
                  +  a_1\cdot1
                    \cdot \lrp{
                        1       \cdot \cancel{\lrsp{P_2, P_0}}
                        + a_{2} \cdot \cancel{\lrsp{P_1, P_0}}
                        + b_{1}\cdot    \lrsp{P_0, P_0}
                    }\nonumber\\
                &\quad
                  +  a_1 \cdot  a_1
                       \times\lrp{
                            1       \cdot \cancel{\lrsp{P_1, P_0}}
                            + a_{1} \cdot \lrsp{P_0, P_0}
                            + 0
                        }\nonumber\\
                &= \textcolor{red}{1\cdot a_2\cdot b_1}
                    + a_1\cdot1\cdot b_1
                    + 1\cdot b_2\cdot a_1
                    + a_1\cdot a_1\cdot a_1.
            \label{eq:development_of_m3}
        \end{align}

        More generally, when computing $m_k =  \lrsp{x^k P_0, P_0}$, the recursive application of the recurrence relation \eqref{eq:3-terms_recurrence_monic_orthogonal_poly} allows to decrease the power of $x$ from $k$ to $0$ until each term in the development is proportional to the inner product of $P_0=1$ with another monic OP.
        The only nonzero such inner product is $\lrsp{P_0,P_0}=1$.
        Consequently, each nonzero term in the final development of $m_k$ corresponds to a path of length at most $k$ that leaves from the lower left corner of the graph in \Figref{fig:Jacobian_moments_rec-coefs_paths} and ends up on the bottom row.
        In between, the path has to remain above the bottom row, and can only move North-East, East, or South-East.
        Each edge corresponds to picking one of the three terms in the recurrence relation \eqref{eq:3-terms_recurrence_monic_orthogonal_poly}.
        For example, the development of $m_3$ in \eqref{eq:development_of_m3} corresponds to three such paths, shown in green in \Figref{fig:Jacobian_moments_rec-coefs_paths}.
        The product of the coefficients along each path forms the resulting term in the development.

        \begin{figure}[ht]
            \centering
            \begin{subfigure}[b]{0.9\textwidth}
                \centering
                \includegraphics[scale=0.5]{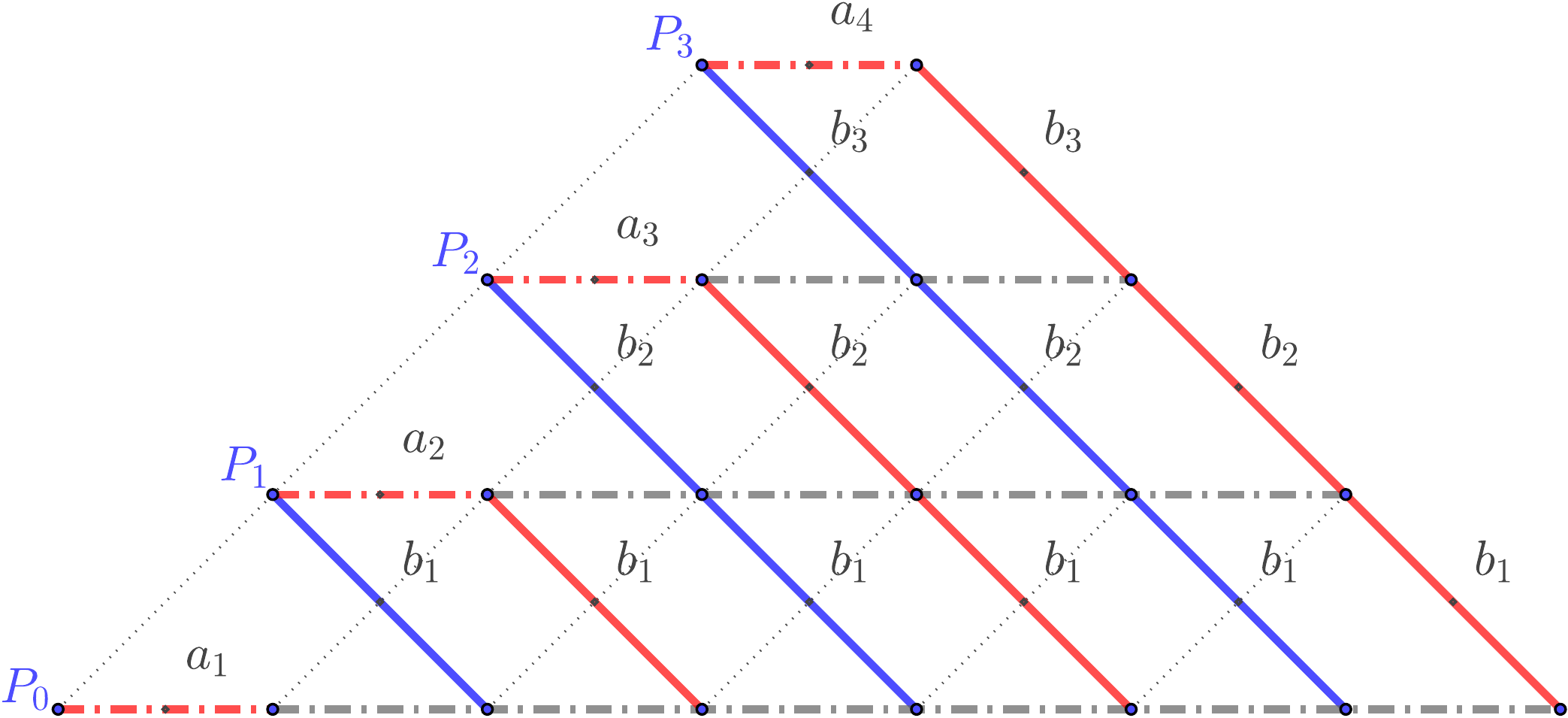}
                \caption{The North-East, East, South-East edges associated to weights $1$, $a_n$, $b_n$ are respectively represented as dashed, dash-dotted and solid lines.
                Note that on each dashed and dash-dotted line, the weight is constant.}
                \label{fig:Jacobian_moments_rec-coefs_paths}
            \end{subfigure}\\
            \begin{subfigure}[b]{0.23\textwidth}
                \centering
                \includegraphics[width=\textwidth]{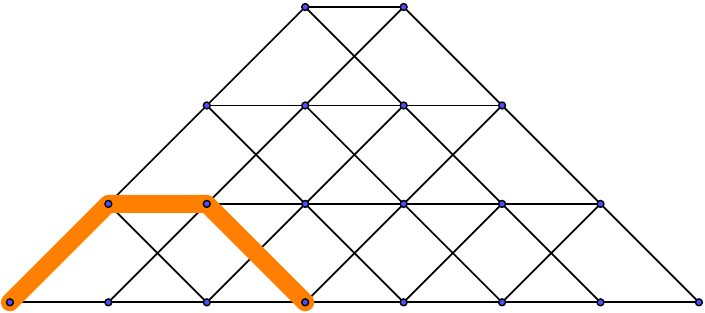}
                \caption{$1\cdot a_2\cdot b_1$}
                \label{fig:path_urd}
            \end{subfigure}
            \begin{subfigure}[b]{0.23\textwidth}
                \centering
                \includegraphics[width=\textwidth]{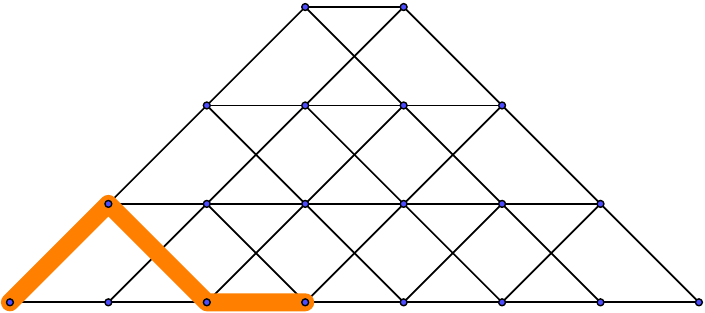}
                \caption{$1\cdot b_2\cdot a_1$}
                \label{fig:path_rud}
            \end{subfigure}
            \begin{subfigure}[b]{0.23\textwidth}
                \centering
                \includegraphics[width=\textwidth]{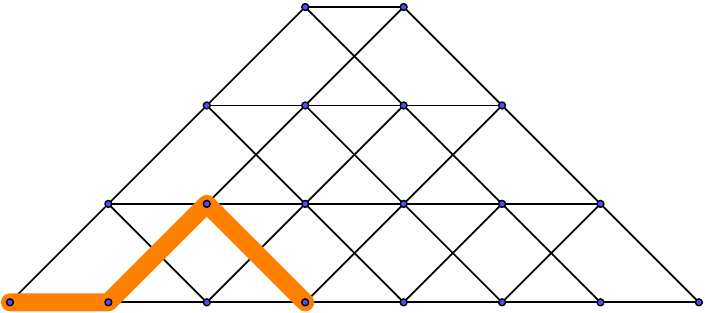}
                \caption{$a_1\cdot1\cdot b_1$}
                \label{fig:path_udr}
            \end{subfigure}
            \begin{subfigure}[b]{0.23\textwidth}
                \centering
                \includegraphics[width=\textwidth]{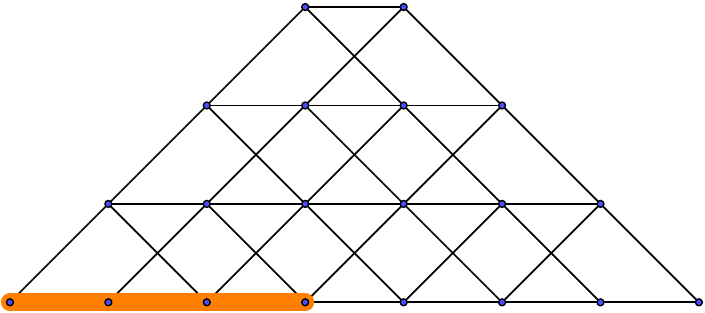}
                \caption{$a_1\cdot a_1\cdot a_1$}
                \label{fig:path_rrr}
            \end{subfigure}
            \caption{The lattice path of \citet{Har18} used to compute $m_n=\lrsp{x^nP_0, P_0}$ is displayed in~\subref{fig:Jacobian_moments_rec-coefs_paths}.
            The paths used for the computation of $m_3$ \eqref{eq:development_of_m3} are highlighted in~\subref{fig:path_urd}-\subref{fig:path_rrr} with the corresponding weight as caption.}
            \label{fig:lattice_path}
        \end{figure}

        In the end, odd moments $m_{2i-1}$, resp.\;even moments $m_{2i}$, are the sum of the weights of the paths below the $i$-th \textcolor{red}{red}, respectively \textcolor{blue}{blue} path, counting from the bottom left.
        More precisely,
        \begin{equation}
        \label{eq:moments_rec_coef_paths}
            m_{2i-1}
                = \textcolor{red}{
                    a_{i} \prod_{k= 1}^{i-1}  b_{k}
                    }
                    + f_1( a_{1:i-1}, b_{1:i-2})
            \quad
                \text{and}
            \quad
            m_{2i}
                = \textcolor{blue}{
                    \prod_{k= 1}^{i}  b_{k}
                    }
                    + f_2( a_{1:i}, b_{1:i-1}).
        \end{equation}
        Thus, the Jacobian is the determinant of a triangular matrix
        \begin{equation*}
            \lrabs{
                    \frac{  \partial m_{1:2N-1}
                        }{  \partial a_{1:N}, b_{1:N-1} }
                  }
            =   \begin{vmatrix}
                    \begin{bmatrix}
                        \frac{\partial m_{2i-1}}{\partial  a_j}
                            &\frac{\partial m_{2i-1}}{\partial  b_{j}} \\
                        \frac{\partial m_{2i}}{\partial  a_j}
                            &\frac{\partial m_{2i}}{\partial  b_{j}}
                    \end{bmatrix}
                    & \begin{bmatrix}
                        \frac{\partial m_{2i-1}}{\partial  a_N}  \\
                        \frac{\partial m_{2i}}{\partial  a_N}
                    \end{bmatrix}
                    \\
                    \lrb{
                        \frac{\partial m_{2N-1}}{\partial  a_j}
                        \enspace
                        \frac{\partial m_{2N-1}}{\partial  b_{j}}
                    }
                    &
                        \frac{\partial m_{2N-1}}{\partial  a_N}\\
                \end{vmatrix}_{i,j= 1}^{N-1}
            =\prod_{i= 1}^N
                \frac{\partial m_{2i-1}}{\partial a_{i}}
              \prod_{i= 1}^{N-1}
                \frac{\partial m_{2i}}{\partial b_{i}}
                \cdot
        \end{equation*}
        The formulation \eqref{eq:moments_rec_coef_paths} yields
        \begin{equation*}
            \frac{\partial m_{2i-1}}{\partial a_{i}}
                = \prod_{k= 1}^{i-1}  b_{k}
            \quad
                \text{and}
            \quad
            \frac{\partial m_{2i}}{\partial b_{i}}
                = \prod_{k= 1}^{i-1}  b_{k}.
        \end{equation*}
        Finally, we obtain
        \begin{align*}
            \lrabs{
                    \frac{  \partial m_{1:2N-1}
                        }{  \partial a_{1:N}, b_{1:N-1} }
                  }
            =   \prod_{i= 1}^N
                    \prod_{k= 1}^{i-1}  b_{k}
                \prod_{i= 1}^{N-1}
                    \prod_{k= 1}^{i-1}  b_{k}
            =   \frac{
                    \lrb{
                        \prod_{i= 1}^{N}
                            \prod_{k= 1}^{i-1}  b_{k}
                       }^2
                     }{
                        \prod_{k= 1}^{N-1}  b_{k}
                     }
            =   \frac{
                    \lrb{
                        \prod_{n= 1}^{N-1} b_{n}^{N- n}
                       }^2
                     }{
                        \prod_{n= 1}^{N-1} b_{n}
                     }\CommaBin
        \end{align*}
        which does not vanish since all $b_n$s are positive by construction.
        Finally, the last equality in \eqref{eq:jacobian_a,b_moments} follows from \Lemref{lem:H_2N-2_Vandermonde_prod_b}.
    \end{proof}

    Propositions~\ref{propo:jacobian_x,w_moments} and \ref{propo:jacobian_a,b_moments} now allow us to conclude that Favard's map $\psi=\rho\circ\phi$ (cf.\,\Thref{th:favard}) is a $C^1$-diffeomorphism, and compute its Jacobian.
    \begin{proposition}
    \label{propo:change_of_variables_favard}
    Favard's map $\psi$ is a $C^1$-diffeomorphism from $\atomSpace\times\weightSpace$ onto $\bbR^N \times (0, +\infty)^{N-1}$, and
        \begin{equation}
        \label{eq:jacobian_x,w_a,b}
            \lrabs{
                    \frac{  \partial x_{1:N}, w_{1:N-1}
                        }{  \partial a_{1:N}, b_{1:N-1} }
                  }
            = \prod_{n=1}^{N-1} b_{n}^{-1} \prod_{n=1}^N w_n.
        \end{equation}
    \end{proposition}

    We now have all the ingredients to give an explicit proof of \Thref{th:joint_distribution_recurrence_coefficients_Tr_J}, of which the three classical tridiagonal models of \Secref{sec:proving_the_three_classical_tridiagonal_models} will be seen to be corollaries.
    \begin{proof}[Proof of \Thref{th:joint_distribution_recurrence_coefficients_Tr_J}]
        For simplicity we drop the indicator functions and rewrite the density of the nodes and weights as
        \begin{align*}
            \eqref{eq:joint_distribution_x_w}
                &= \lrp{\lrabs{\Delta \lrp{x_{1}, \dots, x_{N}}}^{2}
                    \prod_{n=1}^N
                        w_n}^{\frac{\beta}{2}}
                \e^{-\Tr[V(\diag(x_{1}, \dots, x_{N}))]}
                \prod_{n= 1}^N
                    w_n^{-1}
                \diff x_{1:N}
                \diff w_{1:N-1}
      \end{align*}

      Combining \Lemref{lem:H_2N-2_Vandermonde_prod_b}, and the fact that $x_{1}, \dots, x_{N}$ are the eigenvalues of $\Jab$, the change of variables provided by \Proporef{propo:change_of_variables_favard} yields

      \begin{align*}
            \eqref{eq:joint_distribution_x_w}
      &= \lrp{\prod_{n=1}^{N-1}
                    b_n^{N-n}}^{\frac{\beta}{2}}
                \e^{-\Tr[V(\Jab)]}
                \prod_{n= 1}^N
                    w_n^{-1}
                \lrabs{
                    \frac{  \partial x_{1:N}, w_{1:N-1}
                        }{  \partial a_{1:N}, b_{1:N-1} }
                  }
                \diff a_{1:N}
                \diff b_{1:N-1}\\
            &=
                \prod_{n=1}^{N-1}
                b_n^{\frac{\beta}{2}(N-n)-1}
                \e^{-\Tr[V(\Jab)]}
                \diff a_{1:N}
                \diff b_{1:N-1},
        \end{align*}
    where the last equality follows from \eqref{eq:jacobian_x,w_a,b}.
    \end{proof}



\section{Proving the three classical tridiagonal models} 
\label{sec:proving_the_three_classical_tridiagonal_models}

    \Thref{th:joint_distribution_recurrence_coefficients_Tr_J} gives the distribution over recurrence coefficients, from which one has to sample, in order for the atoms of the corresponding atomic measure to follow a given $\beta$-ensemble.
    When the potential of the $\beta$-ensemble is taken among three particular forms, the recurrence coefficients turn out to be independent with simple distributions.
    In particular, the recurrence coefficients are much simpler to sample than the complex joint distribution \eqref{eq:joint_distribution_beta_ensemble} of the atoms.

    \subsection{The H$\beta$E and its tridiagonal model} 
    \label{sub:the_hermite_ensemble_and_its_tridiagonal_model}

        The tridiagonal model associated to the Hermite $\beta$-ensemble, cf.\,\Thref{th:hermite_beta_ensemble}, follows from a direct application of \Thref{th:joint_distribution_recurrence_coefficients_Tr_J} and the following immediate lemma.
        \begin{lemma}
        \label{lem:trace_J_J^2-ab}
            Let $\Jab$ be a Jacobi matrix as defined by \eqref{eq:jacobi_matrix_a_b}, with eigenvalues $x_{1}, \dots, x_{N}$.
            It holds that
            \begin{equation}
            \label{eq:trace_J_J^2-ab}
                \sum_{n=1}^N x_n
                    = \Tr \Jab
                    = \sum_{n=1}^N a_n
                \quad \text{and} \quad
                \sum_{n=1}^N x_n^2
                    = \Tr \Jab^2
                    = \sum_{n=1}^N a_n^2 + 2\sum_{n=1}^{N-1} b_n.
            \end{equation}
        \end{lemma}
        \begin{proof}[Proof of \Thref{th:hermite_beta_ensemble}]
            Starting from \Thref{th:joint_distribution_recurrence_coefficients_Tr_J} it remains to express the term $\Tr V(\Jab)$, where $V(x)= \frac{(x-\mu)^2}{2 \sigma^2} =  \frac{1}{2 \sigma^2}(x^2 - 2\mu x + \mu^2)$.
            To this end, \Lemref{lem:trace_J_J^2-ab} yields
            \begin{equation*}
                \Tr V(\Jab)
                    = \frac{1}{2 \sigma^2} \lrb{\Tr \Jab^2 - 2\mu \Tr \Jab
                        + N \mu^2}
                    =\frac{1}{2 \sigma^2} \sum_{n=1}^N (a_n- \mu)^2
                        + \frac{1}{\sigma^2} \sum_{n=1}^{N-1} b_n.
            \end{equation*}
            Finally, we can plug this expression back into \eqref{eq:joint_distribution_x_w} to see that the entries of $\Jab$  are independently distributed, with joint distribution proportional to
            \begin{equation}
                \prod_{n=1}^{N-1} b_{n}^{\frac{\beta}{2}(N-n)-1}
                    \e^{- \frac{1}{\sigma^2} b_n}
                    \diff b_n
                \prod_{n=1}^{N}
                    \e^{-  \frac{1}{2\sigma^2}(a_n- \mu)^2 }
                    \diff a_n.
            \end{equation}
        \end{proof}

        Note that when $\mu$ is still supported on $\mathbb{R}$, but the potential is a more general polynomial, the recurrence parameters are no longer independent, but the interaction remains short range.
        This is what we later exploit in \Secref{sec:gibbs_sampling}, where we derive a fast approximate sampler for various $\beta$-ensembles with polynomial potentials.


    \subsection{The L$\beta$E and its tridiagonal model} 
    \label{sub:the_l_beta_e_and_its_tridiagonal_model}

        When the target $\beta$-ensemble is supported on $(0, +\infty)$, there is a natural reparametrization of the recurrence coefficients of $\mu$, which allows to express other quantities than those in \Lemref{lem:trace_J_J^2-ab}.
        This leads to the tridiagonal model of \Thref{th:laguerre_beta_ensemble} for the Laguerre $\beta$-ensemble.

        The reparametrization, denoted by $\xi_{1}, \dots, \xi_{2N-1} > 0$, arose in the work of \citet{Sti1894} on continued fractions; see also \citet[Equation 9.12 in Corollary of Theorem 9.1; Equation 2]{Chi78,Chi71} in the context of three-term recurrence relations.
        To introduce these new parameters, first note that since $\mu$ is now supported on $(0,+\infty)$, the recurrence relation  \eqref{eq:3-terms_recurrence_monic_orthogonal_poly} implies
        \begin{equation}
            \label{eq:on_0_infinity_a_n_is_positive}
            a_n
                = \lrnorm{P_{n-1}}^{-2} \lrsp{xP_{n-1}, P_{n-1}}
                = \lrnorm{P_{n-1}}^{-2} \int_0^\infty x P_{n-1}^2(x) \mu(\diff x)
                >0,
                \quad n=1,\dots,N.
        \end{equation}
        Now, we set
        \begin{equation}
        \label{eq:def_ab_xi}
            \begin{aligned}
                a_1
                    = \xi_1,~
                &a_{n}
                    = \xi_{2n-2}+\xi_{2n-1},
                \quad\text{for } 2\leq n \leq N,\\
                \quad\text{and}\quad
                &b_{n}
                    = \xi_{2n-1}\xi_{2n},
                \quad\text{for } 1\leq n \leq N-1.
            \end{aligned}
        \end{equation}
        Equivalently, the new parameters correspond to the Cholesky factorization
        \begin{equation}
            \label{eq:cholesky_factorization_jacobi_matrix_ab_xi}
            \Jab = \Xi~\Xi^{\top},
            \quad\text{where}\quad
            \Xi
            =
            \begin{pmatrix}
                \sqrt{\xi_1}
                    &               &               & (0)\\
                \sqrt{\xi_2}
                    & \sqrt{\xi_3}  &               & \\
                    & \ddots        &\ddots         & \\
                (0)
                    &           & \sqrt{\xi_{2N-2}}     & \sqrt{\xi_{2N-1}}
            \end{pmatrix}.
        \end{equation}
        Note that this bidiagonal transformation is reminiscent of the  construction of the tridiagonal model for the L$\beta$E, where \citet{DuEd02} bidiagonalize a random Gaussian matrix.
        The following proposition shows that the change of variables replacing the recurrence coefficients by $\xi_{1:2N-1}$ is valid.
        \begin{proposition}
            \label{propo:jacobian_xi_a,b}
            Consider $\mu$ supported on $(0,+\infty)$, then the corresponding Jacobi matrix \eqref{eq:jacobi_matrix_a_b} factorizes uniquely as $\Jab = \Xi~\Xi^{\top}$, where $\Xi$ is given by \eqref{eq:cholesky_factorization_jacobi_matrix_ab_xi}.
            Moreover, the mapping
            \begin{equation}
            \label{eq:map_xi_a,b}
                \begin{array}{rcl}
                (\xi_1,\dots, \xi_{2N-1})
                    &\longmapsto
                    &(a_{1:N}, b_{1:N- 1}),
                \end{array}
            \end{equation}
            defined by \eqref{eq:def_ab_xi}
            is a $C^1$-diffeomorphism of $(0, +\infty)^{2N-1}$ onto itself, and its Jacobian reads
            \begin{equation}
            \label{eq:jacobian_xi_a,b}
                \lrabs{
                        \frac{  \partial a_{1:N}, b_{1:N-1}
                            }{  \partial \xi_{1:2N- 1}}
                      }
                = \prod_{i=1}^{N- 1} \xi_{2i-1}.
            \end{equation}
        \end{proposition}
        \begin{proof}
            Given that $\Jab$ is symmetric with positive eigenvalues, the Cholesky factorization $\Jab=\Xi~\Xi^{\top}$ is unique, see, e.g., \citet[Theorem 4.2.7]{GoVL13}.
            Moreover, since $\Jab$ is tridiagonal, the factor $\Xi$ can only be bidiagonal.
            Hence, the mapping \eqref{eq:map_xi_a,b} is injective (and even bijective) and $C^1$ because it is polynomial.
            Finally, by definition of the transformation \eqref{eq:def_ab_xi}, the Jacobian reads as the determinant of a triangular matrix
            \begin{equation*}
                \lrabs{
                        \frac{  \partial a_{1:N}, b_{1:N-1}
                            }{  \partial \xi_{1:2N- 1}}
                      }
                =   \begin{vmatrix}
                        \begin{bmatrix}
                            \frac{\partial a_{i}}{\partial \xi_{2j-1}}
                                & \frac{\partial a_{i}}{\partial \xi_{2j}}\\
                            \frac{\partial b_{i}}{\partial \xi_{2j-1}}
                                & \frac{\partial b_{i}}{\partial \xi_{2j}}
                        \end{bmatrix}_{i,j= 1}^{N-1}
                        & \begin{bmatrix}
                            \frac{\partial a_{i}}{\partial \xi_{2N-1}}\\
                            \frac{\partial b_{i}}{\partial \xi_{2N-1}}
                          \end{bmatrix}_{i=1}^{N-1}
                        \\
                        \begin{bmatrix}
                            \frac{\partial  a_N}{\partial \xi_{2j-1}}
                            \enspace
                                \frac{\partial  a_N}{\partial \xi_{2j}}
                        \end{bmatrix}_{j= 1}^{N-1}
                        &
                            \frac{\partial  a_N}{\partial \xi_{2N-1}}\\
                    \end{vmatrix}
                = \prod_{i= 1}^N
                    \underbrace{\frac{\partial a_{i}}{\partial \xi_{2i-1}}}_{=1}
                  \prod_{i= 1}^{N-1}
                    \underbrace{\frac{\partial b_{i}}{\partial \xi_{2i}}}_{=\xi_{2i-1}}\cdot
            \end{equation*}
        \end{proof}
        For our purpose, the Cholesky factorization \eqref{eq:cholesky_factorization_jacobi_matrix_ab_xi} is ideal to express the key quantities that appear in the L$\beta$E.
        The proof of the corresponding tridiagonal model, cf.\,\Thref{th:laguerre_beta_ensemble}, follows from a direct application of \Thref{th:joint_distribution_recurrence_coefficients_Tr_J} and the following immediate lemma.
        \begin{lemma}
        \label{lem:sum-prod_x_xi}
            Let $\Jab = \Xi~\Xi^{\top}$ as in \eqref{eq:cholesky_factorization_jacobi_matrix_ab_xi} and note $x_{1}, \dots, x_{N}$ its eigenvalues.
            Then,
            \begin{equation}
            \label{eq:sum-prod_x_xi}
                \sum_{n=1}^N x_n
                    = \Tr \Jab
                    = \sum_{n=1}^{2N-1} \xi_{n}
                \quad \text{and} \quad
                \prod_{n=1}^N x_n
                    = \det \Jab
                    = \prod_{n=1}^N \xi_{2n-1}.
            \end{equation}
        \end{lemma}
        \begin{proof}[Proof of \Thref{th:laguerre_beta_ensemble}]
            Applying \Lemref{lem:sum-prod_x_xi} to the $V(x)= -(k-1)\log(x) + \frac{x}{\theta}$ yields
            \begin{align}
                \exp\lrb{- \Tr V(\Jab)}
                    &=\lrp{\det \Jab}^{k-1}
                        \exp\lrp{-\frac{1}{\theta}\Tr \Jab}
                    \nonumber\\
                    &\oplim{=}{}{\eqref{eq:sum-prod_x_xi}} \prod_{n=1}^{N} \xi_{2n-1}^{k-1}
                        \exp\lrp{-\frac{1}{\theta}
                                 \sum_{n=1}^{2N-1} \xi_{n}}.
                    \label{eq:exp_-_Tr_V_Jab_xi}
            \end{align}
            Starting from \eqref{eq:joint_distribution_a,b_Tr-J}, \Proporef{propo:jacobian_xi_a,b} gives the joint distribution of the underlying $\xi_{1:2N-1}$ parameters as proportional to
            \begin{align}
                &\prod_{n=1}^{N-1}
                    b_{n}^{\frac{\beta}{2}(N-n)-1}
                    \e^{- \Tr V(\Jab)}
                    \diff a_{1:N}
                    \diff b_{1:N-1}
                \nonumber\\
                &
                \overset{\eqref{eq:def_ab_xi}}{=}
                \prod_{n=1}^{N-1}
                    \lrp{\xi_{2n-1}\xi_{2n}}^{\frac{\beta}{2}(N-n)-1}
                \e^{- \Tr V(\Jab)}
                \lrabs{
                \frac{\partial a_{1:N}, b_{1:N-1}}
                    {\partial \xi_{1:2N-1}}
                }
                \diff \xi_{1:2N-1}
                \nonumber
                \\
                &
                \overset{\eqref{eq:jacobian_xi_a,b}}{=}
                \prod_{n=1}^{N-1}
                    \xi_{2n-1}^{\frac{\beta}{2}(N-n)-\cancel{1}}
                    \xi_{2n}^{\frac{\beta}{2}(N-n)-1}
                \e^{- \Tr V(\Jab)}
                \cancel{\prod_{n=1}^{N-1} \xi_{2n-1}}
                \diff \xi_{1:2N-1}
                \label{eq:joint_distribution_xi_Tr-J}
                \\
                &
                \overset{\eqref{eq:exp_-_Tr_V_Jab_xi}}{=}
                \prod_{n=1}^{N}
                    \xi_{2n-1}^{\frac{\beta}{2}(N-n)}
                \prod_{n=1}^{N-1}
                    \xi_{2n}^{\frac{\beta}{2}(N-n)-1}
                \Big(\prod_{n=1}^N \xi_{2n-1}\Big)^{k-1}
                \e^{- \frac{1}{\theta}\sum_{n=1}^{2N-1} \xi_n}
                \diff \xi_{1:2N-1}
                \nonumber
            \end{align}

        \end{proof}

        In the next section, we introduce another reparametrization of the recurrence coefficients, this time when $\mu$ is supported in a compact interval: the canonical moments of \citet{DeSt97}.


    \subsection{The J$\beta$E and its tridiagonal model} 
    \label{sub:the_j_beta_e_and_its_tridiagonal_model}

        Finding a tridiagonal model for the J$\beta$E was left as an open problem by \citet[IV B]{DuEd02}.
        The latter was addressed by \citet[Theorem 2]{KiNe04} in their study of the quindiagonal model associated to the circular ensemble.
        However, the authors acknowledged that lifting the points on the unit circle to apply their result represents a winding detour to prove the J$\beta$E.
        Besides, the Jacobian required by this method was obtained by indirect means by \citet[Lemma 4.3]{KiNe07}.
        Subsequently, \citet[Theorem 2]{FoRa06} obtained the Jacobian more directly.

        We can actually prove the tridiagonal model of \Thref{th:jacobi_beta_ensemble} by reparametrizing the Jacobi matrix $\Jab$ again, this time using \emph{canonical moments} \citep[Chapter 1]{DeSt97}.
        In essence, the result can be found in the work of \citet{GaRo10} and \citet{DeNa12}, but we rephrase it as just another consequence of \Thref{th:joint_distribution_recurrence_coefficients_Tr_J}.

        Before formally introducing them, let us mention that canonical moments and their complex counterpart were successfully used to investigate the connection between randomized moments problems, orthogonal polynomials, and optimal design \citep{DeSt97} and in random matrix theory \citep{GaRo10,GaNaRo16}.
        In particular, canonical moments can be thought of as a reparametrization of the moments, where $0<c_n<1$ represents the relative position of the $n$-th moment $m_n$ in the range of all possible moments associated to measure with compatible previous moments $m_{1}, \dots, m_{n-1}$, see \citet{DeSt97}.

        Throughout this section, we assume that the $N$-atomic measure $\mu$ is supported on $(0,1)$.
        In particular, with $(\xi_n)$ the parameters introduced in Section~\ref{sub:the_l_beta_e_and_its_tridiagonal_model}, it comes
        \begin{equation*}
            \label{eq:on_0_1_a_n_is_in_0_1}
            0   < \xi_{2n-2} + \xi_{2n-1}
                = a_n
                = \lrnorm{P_{n-1}}^{-2} \lrsp{xP_{n-1}, P_{n-1}}
                < 1,
                \quad n=2,\dots,N.
        \end{equation*}
        Similarly, $\xi_1=a_1\in (0,1).$
            This implies $0<\xi_n<1$ for all $1\leq n\leq 2N-1$.
        Following the work of \citet{Wal40} on chain sequences and continued fractions, we introduce a new parametrization of the recurrence coefficients.

        \begin{lemma}[Wall]
        \label{lem:wall}
        Assume $\mu$ is supported on $(0,1)$, there exist a sequence $(c_n)\in(0,1)^{\mathbb{N}}$ such that
        \begin{equation}
            \label{eq:def_xi_c}
                \xi_1 = c_1
                \quad \text{and} \quad
                \xi_n = (1-c_{n-1})c_n,
                \quad \forall 2\leq n \leq 2N-1.
            \end{equation}
        \end{lemma}

        We do not prove Lemma~\ref{lem:wall} and refer to \citet[Theorem 6.1]{Wal40}; see also \citet[Chapter 3]{Chi78} for more details on chain sequences.
        We simply note that defining $(c_n)$ in \eqref{eq:def_xi_c} is straightforward, the nontrivial part of the lemma is that $0<c_n<1$ for all $n$.
        We also note that the $c_n$s are today known as the \emph{canonical moments} of $\mu$; see the monograph of \citet{DeSt97}.

        The following proposition shows that the change of variables replacing $\xi$ by $c$ is valid.
        \begin{proposition}
        \label{propo:jacobian_c_xi}
            Consider $\mu$ supported on $(0,1)$, then the corresponding Jacobi matrix \eqref{eq:jacobi_matrix_a_b} can be parametrized in terms of the canonical moments following \eqref{eq:def_ab_xi} and \eqref{eq:def_xi_c}.
            Moreover, the mapping
            \begin{equation}
            \label{eq:map_c_xi}
                \begin{array}{rcl}
                (c_1,\dots, c_{2N-1})
                    &\longmapsto
                    &(\xi_1,\dots, \xi_{2N-1}),
                \end{array}
            \end{equation}
            defined by \eqref{eq:def_xi_c}
            is a $C^1$-diffeomorphism of $(0, 1)^{2N-1}$ onto itself, and its Jacobian reads
            \begin{equation}
            \label{eq:jacobian_c_xi}
                \lrabs{
                    \frac{  \partial \xi_{1:2N- 1}
                        }{  \partial c_{1:2N-1}}
                  }
                = \prod_{n=1}^{2N-2} (1-c_n).
            \end{equation}
        \end{proposition}
        \begin{proof}
            The map \eqref{eq:map_c_xi} is a bijection by definition and Lemma~\ref{lem:wall}, and $C^1$ because it is polynomial.
            Then, by definition of the transformation \eqref{eq:def_xi_c}, the Jacobian is the determinant of a triangular matrix
            \begin{equation*}
                \lrabs{
                        \frac{  \partial \xi_{1:2N- 1}
                            }{  \partial c_{1:2N-1}}
                      }
                = \prod_{n= 1}^{2N-1}
                    \frac{  \partial \xi_{n}
                            }{  \partial c_{n}
                        }
                = 1 \cdot \prod_{n=2}^{2N-1} (1-c_{n-1})
                = \prod_{n=1}^{2N-2} (1-c_{n}).
            \end{equation*}
        \end{proof}
        For our purpose, the canonical moment parametrization is ideal to express the key quantities that appear in the J$\beta$E.
        The proof of the corresponding tridiagonal model in \Thref{th:jacobi_beta_ensemble} is again a direct application of \Proporef{propo:jacobian_c_xi} and the following lemma.
        \begin{lemma}
            \label{lem:prod_x,1-x_c}
            It holds that
            \begin{equation}
            \label{eq:prod_x,1-x_c}
                \prod_{n=1}^N x_n
                    = \det \Jab
                    = \prod_{n=1}^N c_{2n-1}
                        \prod_{n=1}^{N-1} (1-c_{2n})
                \quad \text{and} \quad
                \prod_{n=1}^N (1-x_n)
                    = \det[I_N - J]
                    = \prod_{n=1}^{2N-1} (1-c_{n}).
            \end{equation}
        \end{lemma}
        \begin{proof}
            First, combine the result of \Lemref{lem:sum-prod_x_xi} and the definition of the canonical moments in \eqref{eq:def_xi_c} to get
            \begin{equation}
                \prod_{n=1}^N x_n
                    = \xi_1 \prod_{n=2}^N \xi_{2n-1}
                    = c_1 \prod_{n=2}^N (1-c_{2n-2})c_{2n-1}
                    = \prod_{n=1}^N c_{2n-1} (1-c_{2n}).
            \end{equation}
            Then, \Lemref{lem:moment_matrices_vandermonde} yields
            \begin{equation}
            \label{eq:prod_1-x_ratio_moment_matrices}
                \prod_{n= 1}^N (1-x_n)
                =
                \frac{\lrabs{\overline{H}_{2N-1}}}
                     {\lrabs{\underline{H}_{2N-2}}}
                \cdot
            \end{equation}
            The denominator can be expressed in terms of the $\xi_{1:2N-1}$ parameters
            \begin{equation}
                \lrabs{\underline{H}_{2N-2}}^2
                    \oplim{=}{}{\eqref{eq:H_2N-2_Vandermonde_prod_b}}
                        \prod_{n=1}^{N-1} b_n^{N-n}
                    \oplim{=}{}{\eqref{eq:def_ab_xi}}
                        \prod_{n=1}^{N-1} \lrb{\xi_{2n-1}\xi_{2n}}^{N-n}.
            \end{equation}
            For the numerator, we follow \citet[Theorem 1.4.10]{DeSt97} who introduced additional quantities $\gamma_{1:2N-1}$ to get
            \begin{equation}
            \lrabs{\overline{H}_{2N-1}}
                = \gamma_{1}^N
                    \prod_{n=1}^{N- 1}
                        \lrb{\gamma_{2n}\gamma_{2n+1}}^{N-n},
            \end{equation}
            where
            \begin{equation}
            \label{eq:def_xi_gamma}
                \begin{cases}
                    \xi_1 = c_1\\
                    \gamma_1 = 1 - c_1
                \end{cases}
                \quad \text{and} \quad
                \begin{cases}
                    \xi_n = (1-c_{n-1})c_n\\
                    \gamma_{n} = c_{n-1} (1-c_n)
                \end{cases}
                \forall 2 \leq n \leq 2N-1.
            \end{equation}
            We plug these results back into \eqref{eq:prod_1-x_ratio_moment_matrices}, and conclude that
            \begin{align*}
                \prod_{n= 1}^N (1-x_n)
                    &= \frac{\lrabs{\overline{H}_{2N-1}}}
                            {\lrabs{\underline{H}_{2N-2}}}
                    = \gamma_{1}^N
                        \prod_{n=1}^{N- 1}
                            \lrb{\frac{
                                        \gamma_{2n}\gamma_{2n+1}
                                    }{
                                        \xi_{2n-1}\xi_{2n}
                                    }
                                }^{N- n}\\
                    &= (1-c_1)^N
                        \lrb{\frac{\cancel{c_1}(1-c_2)\cancel{c_2}(1-c_3)}
                                  {\cancel{c_1}(1-c_1)\cancel{c_2}}}^{N-1}
                        \prod_{n=2}^{N- 1}
                            \lrb{\frac{
                                        \cancel{c_{2n-1}}(1-c_{2n})
                                        \cancel{c_{2n}}(1-c_{2n+1})
                                    }{
                                        (1-c_{2n-2})\cancel{c_{2n-1}}
                                        (1-c_{2n-1})\cancel{c_{2n}}
                                    }
                                }^{N- n}\\
                    &= (1-c_1)^N
                        \prod_{n=1}^{N- 1}
                            \lrb{\frac{
                                        1-c_{2n+1}
                                    }{
                                        1-c_{2n-1}
                                    }
                                }^{N- n}
                        ~
                        (1-c_{2})^{N-1}
                        \prod_{n=2}^{N- 1}
                            \lrb{\frac{
                                        1-c_{2n}
                                    }{
                                        1-c_{2n-2}
                                    }
                                }^{N- n}\\
                    &=
                        \prod_{n=1}^N
                            \lrp{1-c_{2n-1}}
                        \prod_{n=1}^{N-1}
                            \lrp{1-c_{2n}}.
            \end{align*}
        \end{proof}

        \begin{proof}[Proof of \Thref{th:jacobi_beta_ensemble}]
            Considering the potential $V(x)= -[(a-1)\log(x) + (b-1)\log(1-x)]$, \Lemref{lem:prod_x,1-x_c} yields
            \begin{align}
                \exp[-\Tr V(\Jab)]
                    &= \lrb{\det J}^{a-1} \lrb{\det {I_N - J}}^{b-1}
                    \nonumber\\
                    &\oplim{=}{}{\eqref{eq:prod_x,1-x_c}}
                        \prod{n=1}{N} c_{2n-1}^{a-1}(1-c_{2n-1})^{b-1}
                        \prod{n=1}{N-1} (1-c_{2n})^{a+b-2}.
                    \label{eq:exp_-_Tr_V_Jab_c}
            \end{align}
            Starting from \eqref{eq:joint_distribution_xi_Tr-J}, \Proporef{propo:jacobian_c_xi} allows us to express the joint distribution of the canonical moments as
            \begin{align*}
                &\prod_{n=1}^{N-1}
                    \frac{\lrb{\xi_{2n-1}\xi_{2n}}^{\frac{\beta}{2}(N-n)}}
                         {\xi_{2n}}
                \e^{- \Tr V(\Jab)}
                    \lrabs{
                    \frac{\partial \xi_{1:2N-1}}
                        {\partial c_{1:2N-1}}
                    }
                    \diff c_{1:2N-1}\\
                &\oplim{=}{}{\eqref{eq:jacobian_c_xi}}
                    \prod_{n=1}^{N-1}
                    \frac{\lrb{\xi_{2n-1}\xi_{2n}}^{\frac{\beta}{2}(N-n)}}
                         {\xi_{2n}}
                    \prod_{n=1}^{2N-2} (1-c_{n})
                    \e^{- \Tr V(\Jab)}
                    \diff c_{1:2N-1}\\
                &\oplim{=}{}{\eqref{eq:def_xi_c}}
                    \frac{\lrb{c_{1}(1-c_{1})c_{2}}^{\frac{\beta}{2}(N-1)}}
                         {\cancel{(1-c_{1})}c_{2}}
                \prod_{n=2}^{N-1}
                    \frac{\lrb{(1-c_{2n-2})c_{2n-1}
                               (1-c_{2n-1})c_{2n}}^{\frac{\beta}{2}(N-n)}}
                         {\cancel{(1-c_{2n-1})}c_{2n}}\\
                &\qquad\qquad
                    \prod_{n=1}^{N-1} \cancel{(1-c_{2n-1})}(1-c_{2n})
                    \e^{- \Tr V(\Jab)}
                    \diff c_{1:2N-1}\\
                &=
                \prod_{n=1}^{N}
                    \lrb{c_{2n-1}(1-c_{2n-1})}^{\frac{\beta}{2}(N-n)}
                \prod_{n=1}^{N-1}
                    c_{2n}^{\frac{\beta}{2}(N-n)-1}
                    (1-c_{2n})^{\frac{\beta}{2}(N-n-1)+1}
                    \e^{- \Tr V(\Jab)}
                    \diff c_{1:2N-1}\\
                &\overset{\eqref{eq:exp_-_Tr_V_Jab_c}}{=}
                \prod{n=1}{N}
                    c_{2n-1}^{\frac{\beta}{2}(N-n)+a-1}
                    (1-c_{2n-1})^{\frac{\beta}{2}(N-n)+b-1}
                \prod{n=1}{N-1}
                    c_{2n}^{\frac{\beta}{2}(N-n)-1}
                    (1-c_{2n})^{\frac{\beta}{2}(N-n-1)+a+b-1}
                    \diff c_{1:2N-1}.
            \end{align*}
        \end{proof}


\section{Gibbs sampling tridiagonal models associated to polynomial potentials} 
\label{sec:gibbs_sampling}

    As seen in \Secref{sec:proving_the_three_classical_tridiagonal_models}, for the specific potentials associated to the Hermite, Laguerre, and Jacobi $\beta$-ensembles, the successive parametrizations of the corresponding Jacobi matrix $\Jab$ yield independent coefficients with easy-to-sample distributions.
    Thus, computing the eigenvalues of the corresponding randomized tridiagonal Jacobi matrices gives $\calO(N^2)$ exact samplers.
    However, when the potential $V$ is generic, these Jacobi parameters may not be independent anymore.
    For polynomial potentials, this dependence remains mild, in the sense that each parameter remains independent from the rest conditionally on a few ``neighboring" parameters.
    As we shall see in this section, simple Gibbs samplers in the space of these Jacobi parameters can provide surprisingly fast-mixing approximate samplers for $\beta$-ensembles.
    In short, we study a Gibbs sampler on tridiagonal matrices, which we can diagonalize in $\calO(N^2)$ to obtain approximate samples from a given $\beta$-ensemble.
    This approach is in contrast with that of \citet{LiMe13} and \citet{ChFe18}, who used MCMC directly on the original space where the particles $\lrcb{x_{1}, \dots, x_{N}}$ live.

    Our starting point is Proposition 2 of \citet{KrRiVi16}, which we rederived as \Thref{th:joint_distribution_recurrence_coefficients_Tr_J}.
    In short, a Jacobi matrix $\Jab$ with coefficients distributed as
    \begin{equation}
        (a_1, b_1, \dots, a_{N-1}, b_{N-1}, a_N)
            \sim
            \prod_{i=1}^{N-1}
            b_{i}^{\frac{\beta}{2}(N-i)-1}
            \exp^{-\Tr V(J_{a, b})}
                \diff a_{1:N}, b_{1:N-1},
    \label{eq:joint_jacobi_coefficients}
    \end{equation}
    has eigenvalues distributed according to the $\beta$-ensemble \eqref{eq:joint_distribution_beta_ensemble} with potential $V$.
    \citet{KrRiVi16} already mention their intuition that a Gibbs chain with invariant measure \eqref{eq:joint_jacobi_coefficients} and a polynomial potential would mix fast, in $\calO(\log N)$, due to the short range interaction between the coefficients.
    From an algorithmic point of view, the explicit conditionals in \eqref{eq:joint_jacobi_coefficients} similarly invite to use a Gibbs sampler, which we investigate in this section.
    For the sake of presentation, we fix the potential to be a polynomial with even degree at most 6 and positive leading coefficient, i.e.
    \begin{equation}
    \label{eq:sextic_potential}
        V(x)
            = g_6 x^6
            + \cancel{g_5 x^5}
            + g_4 x^4
            + g_3 x^3
            + g_2 x^2
            + g_1 x.
    \end{equation}
    The absence of a term of degree $5$ in \eqref{eq:sextic_potential} comes from practical reasons detailed in \Secref{sub:sampling_from_the_conditionals}.
    While the method applies more generally, we restrict ourselves to potentials of the form \eqref{eq:sextic_potential} because
    $(i)$ it already goes beyond the numerical state-of-the-art,
    $(ii)$ it is rich enough to require different sampling schemes for different conditionals depending on the coefficients in \eqref{eq:sextic_potential}, and
    $(iii)$ the theory of sextic potentials is advanced enough that we have means to empirically assess the convergence of our samplers.

    The associated implementation is available in our \DPPy~toolbox.
    We also provide a companion Python notebook where we illustrate our sampler on various potentials, see \href{https://github.com/guilgautier/DPPy/tree/master/notebooks}{https://github.com/guilgautier/DPPy/tree/master/notebooks}.

    \subsection{Sampling from the conditionals} 
    \label{sub:sampling_from_the_conditionals}

        We implement a systematic scan Gibbs sampler \citep[Chapter 10]{RoCa04} to approximately sample from the distribution \eqref{eq:joint_jacobi_coefficients} on the Jacobi coefficients.
        Writing the conditionals in closed form for the generic sextic potential \eqref{eq:sextic_potential} is cumbersome, but we do it for a specific instance in \Exref{ex:quartic} below.
        The expansion of $\Tr V(\Jab)$ in \eqref{eq:joint_jacobi_coefficients} reveals that the size of the Markov blanket of each coefficient grows with $\degree V$.
        Quoting \citet[Section 1]{KrRiVi16}, \emph{variables with indices that are $\degree V / 2$ apart are conditionally independent given the variables in between}.
        In other words, the Jacobi coefficients $a_1, \dots, a_N, b_1, \dots, b_{N-1}$ have a more short-range interaction than the corresponding particles $x_{1}, \dots, x_{N} $.
        Gibbs sampling can leverage that property.

        For $1\leq i\leq N$, let $\bfa_{\setminus i} = (a_1,\dots,a_{i-1},a_{i+1},\dots,a_N)$.
        Similarly, for $1\leq j\leq N-1$, let $\bfb_{\setminus j} = (b_1,\dots,b_{j-1},b_{j+1},\dots,b_{N-1})$.
        In practice, we define one complete Gibbs pass as sampling from
            $a_n \mid \bfa_{\setminus n}, \bfb$, and then
            $b_n \mid \bfa, \bfb_{\setminus n}$,
        for each $n$ in turn.
        We avoid the term of degree $5$ in \eqref{eq:sextic_potential} to make sure that the conditionals
        $b_n \mid \bfa, \bfb_{\setminus n}$ are always log-concave, while the conditionals
        $a_n \mid \bfa_{\setminus n}, \bfb$ are log-concave if $g_2>0$ and $g_3=g_6=0$.
        Univariate log-concave densities are interesting from a sampling point of view, since they are usually amenable to efficient rejection sampling.

        In our case, for every log-concave conditional, the mode of the corresponding density can be derived analytically.
        We can thus use the tailored rejection sampler of \citet{Dev12}, with an expected $5$ rejection steps per draw; see \Exref{ex:quartic} for details.
        The overall algorithm is given in \Algoref{alg:gibbs_sampler}.

        When the conditionals $a_n \mid \bfa_{\setminus n}, \bfb$ are not log-concave, we switch from a Gibbs algorithm to a Metropolis-within-Gibbs algorithm, and replace exact sampling of the corresponding conditional by a draw from a Metropolis-Hastings kernel.
        More specifically, since the log of the conditional densities $a_n \mid \bfa_{\setminus n}, \bfb$ are polynomials, they are easy to differentiate, and we use their gradient in a Metropolis-adjusted Langevin kernel (MALA, see, e.g., \citealp[Section 7.8.5]{RoCa04}).
        \begin{algorithm}[tb]
           \caption{Gibbs sampler to sample from \eqref{eq:joint_jacobi_coefficients} with $\beta>0$ and $V$ as in \eqref{eq:sextic_potential}}
           \label{alg:gibbs_sampler}
            \begin{algorithmic}
               \STATE {\bfseries Input:}
               inverse temperature $\beta$, potential $V$, number of MCMC steps $T$
               \STATE Initialize $a_1 = \cdots = a_N = b_1 = \cdots = b_{N-1} = 0$
                \FOR{$t=1$ {\bfseries to} $T$}
                    \FOR{$n=1$ {\bfseries to} $N$}
                        \STATE Sample $a_n \mid \bfa_{\setminus n}, \bfb$
                        \IF{$n < N$}
                            \STATE Sample $b_n \mid \bfa, \bfb_{\setminus n}$
                        \ENDIF
                    \ENDFOR
                \STATE $x_{1}^{t}, \dots, x_{N}^{t} = \operatorname{eigvals}(\Jab)$
                \ENDFOR
            \end{algorithmic}
        \end{algorithm}

        \clearpage
        \begin{example}[Quartic potential]
            \label{ex:quartic}
            Let $V(x) = g_4 x^4 + g_2 x^2$.
            With the convention $a_0=a_{N+1}=b_0=b_{N}=0$, the conditionals write as follows.\\
            For all $n \in \{1,\dots, N\}$,
            \begin{gather}
            \begin{aligned}
                &a_n \mid \bfa_{\setminus n}, \bfb\\
                &\sim \exp
                        \lrb{
                            -\lrb{
                                g_4 a_n^4
                                + a_n^2
                                    \lrb{g_2 + 4g_4(b_{n-1}+b_{n}) }
                                + 4g_4 a_n(a_{n-1}b_{n-1}+a_{n+1}b_{n})
                            }
                        },
            \label{eq:a_conditional_quartic}
            \end{aligned}
            \end{gather}
            For all $n \in \{1,\dots, N-1\}$,
            \begin{gather}
            \begin{aligned}
                &b_n \mid \bfa, \bfb_{\setminus n}\\
                &\sim b_n^{\frac{\beta}{2}(N-i)-1}
                        \exp
                            \lrb{-2
                                 \lrb{g_4 b_n^2
                                    + b_n \lrb{g_2
                                                + 2g_4(a_{n}^2
                                                       +a_{n}a_{n+1}
                                                       +a_{n+1}^2
                                                       +b_{n-1}+b_{n+1})
                                              }
                                     }
                                }.
            \label{eq:b_conditional_quartic}
            \end{aligned}
            \end{gather}
            In this case, for $g_2, g_4 > 0$, the conditionals given in Equations~\ref{eq:a_conditional_quartic} and \ref{eq:b_conditional_quartic} are unnormalized and $\log$-concave, with easy-to-find modes.
            Thus, the rejection sampling technique of \citet{Dev12} applies, with an expected number of rejections equal to $5$.
            Given an unnormalized and $\log$-concave target density $\pi$ with mode $m = \argmax_y \pi(y)$, \citet{Dev12} constructs a piecewise dominating function $h$ comprising 3 plateaus and 2 exponential tails such that $\int h / \int \pi \leq 5$.
            The breakpoints $m+2u, m+u, m+v, m+2v$ are located on both sides of the mode, where $u < 0 < v$ satisfy $\pi(m+x) \geq \pi(m) / 4 \geq \pi(m+2x)$.
            Such $u$ and $v$ can be found using a simple bisection method.
            In practice, we compute $u' < 0 < v'$ solutions of $\pi(m + x) = \pi(m) / 4$ and assign $u=u'/2$ and $v=v'/2$, see \Figref{fig:illustration_of_Devroye_rejection_sampler}.

                    %
                    %

            \begin{figure}[!ht]
                \centering
                \begin{subfigure}[b]{0.49\textwidth}
                    \centering
                    \includegraphics[width=\textwidth]{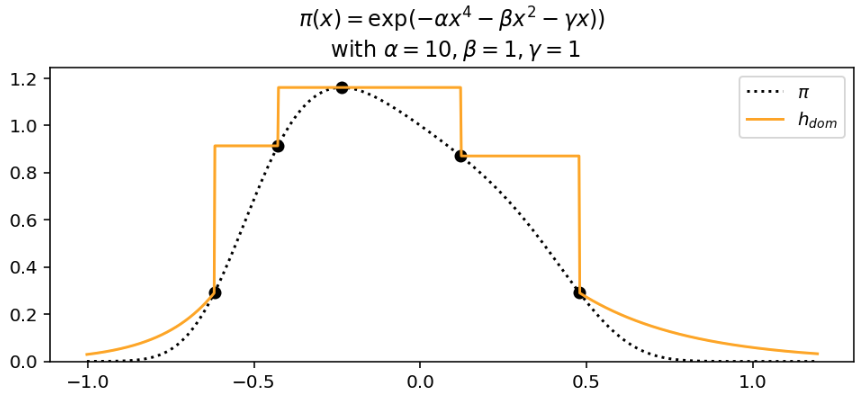}
                    \caption{An example conditional \eqref{eq:a_conditional_quartic}}
                    \label{fig:illustration_Dev12_sampler_quartic}
                \end{subfigure}
                \hfill
                \begin{subfigure}[b]{0.49\textwidth}
                    \centering
                    \includegraphics[width=\textwidth]{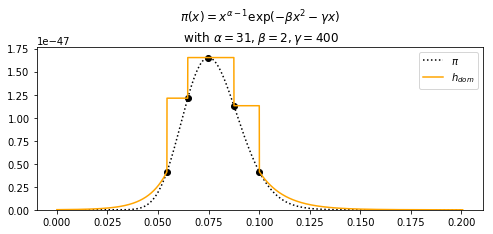}
                    \caption{An example conditional \eqref{eq:b_conditional_quartic}}
                    \label{fig:illustration_Dev12_sampler_gen_gamma}
                \end{subfigure}
                \caption{
                Construction of the dominating function $h$ (solid line) of \citet{Dev12} to perform rejection sampling with a log-concave target $\pi$ (dashed line), the mode of which needs to be analytically tractable.
                }
                \label{fig:illustration_of_Devroye_rejection_sampler}
            \vspace{-2em}
            \end{figure}
        \end{example}



    \subsection{Example simulations and empirical study of the convergence} 
    \label{sub:empirical_study_of_the_mixing_time}

        In this section, we investigate the convergence of the Gibbs sampler detailed in \Secref{sub:sampling_from_the_conditionals}.
        We sample from $\beta$-ensembles with potential $W(x) = \frac{\beta N}{2} V(x)$, for various choices of $V$ of the form \eqref{eq:sextic_potential}.
        The rescaling in $W$ is applied to capture the weak convergence of the empirical distribution of the particles towards the corresponding equilibrium measure $\mu_{\text{eq}}$; see e.g., \citet[Section 6.1]{Dei00}.
        Intuitively, the rescaling balances the effect of the Vandermonde determinant and that of the potential $V$ in \eqref{eq:joint_distribution_beta_ensemble}.

        \subsubsection{Convergence of the marginals} 
        \label{ssub:marginal_mixing}

            Let $(x_n^t)_{1\leq n\leq N}$ be the vector of ordered particles after $t$ full Gibbs passes, that is, after $t$ outer iterations of \Algoref{alg:gibbs_sampler}.
            A first quantity to monitor is how well the empirical distribution
            $\hat{\mu}_N^t = N^{-1}\sum_{n=1}^{N} \delta_{x_n^t}$ approximates, as $t$ grows, the empirical distribution of the target $\beta$-ensemble
            $$
            \hat{\mu}_N = \frac1N \sum\limits_{n=1}^{N} \delta_{x_n}, \quad \text{where }\{x_1,\dots,x_N\} \text{ is drawn from \eqref{eq:joint_distribution_beta_ensemble}.}
            $$
            It turns out that, under assumptions on the potential $V$ that are satisfied by \eqref{eq:sextic_potential}, the random measure $\hat\mu_N$ is itself well approximated when $N\gg 1$ by a (deterministic) measure $\mueq$ called the \emph{equilibrium measure} of the potential.
            This statement can be made rigorous; see for instance the large deviation principle with fast rate $1/N^2$ in \cite[Theorem 2.3]{Ser15}.

            Two observations are in order.
            First, the fast rate hints that the approximation should hold even for moderate values of $N$, as we shall confirm later on in our simulations.
            Second, $\mueq$ is known analytically for a few choices of polynomial potentials.
            Thus, for these potentials, we compare draws from $\hat{\mu}_N^t = N^{-1}\sum_{n=1}^{N} \delta_{x_n^t}$ with $\mu_{\text{eq}}$, to assess convergence of our marginals $\hat\mu^t_N$ as $t$ grows.
            This is in line with the experiments of \citet{LiMe13}, \citet{OlNaTr15} and \citet{ChFe18}.

            \paragraph{The quartic potential.} 
            \label{par:the_quartic_potential}

                The equilibrium measure $\mu_{\text{eq}}$ is available in closed form for potentials proportional to $x^{2d}$ \citep[Proposition 6.156]{Dei00}.
                We consider again $V(x) = \frac14 x^4$, as in \Exref{ex:quartic}.
                In this case all conditionals are log-concave, and can thus be sampled exactly, cf.\,\Secref{sub:sampling_from_the_conditionals}.
                \Figref{fig:quartic_marginal_convergence} shows the agregation of the marginal histograms of $1000$ independent runs, after each of the first few Gibbs passes.

                Observe that convergence to the equilibrium measure is extremely fast: beyond $t=3$ Gibbs passes, the histograms are visually indistinguishable from the equilibrium measure.
                This observation is quantitatively monitored in \Figref{fig:quartic_marginal_convergence}, where we plot the logarithm of the $L_\infty$ distance between the empirical cdf of $\hat\mu_N^t$ and the cdf of $\mueq$.

                \begin{figure}[!ht]
                    \centering
                    \begin{subfigure}[b]{.45\textwidth}
                        \centering
                        \includegraphics[width=\textwidth]{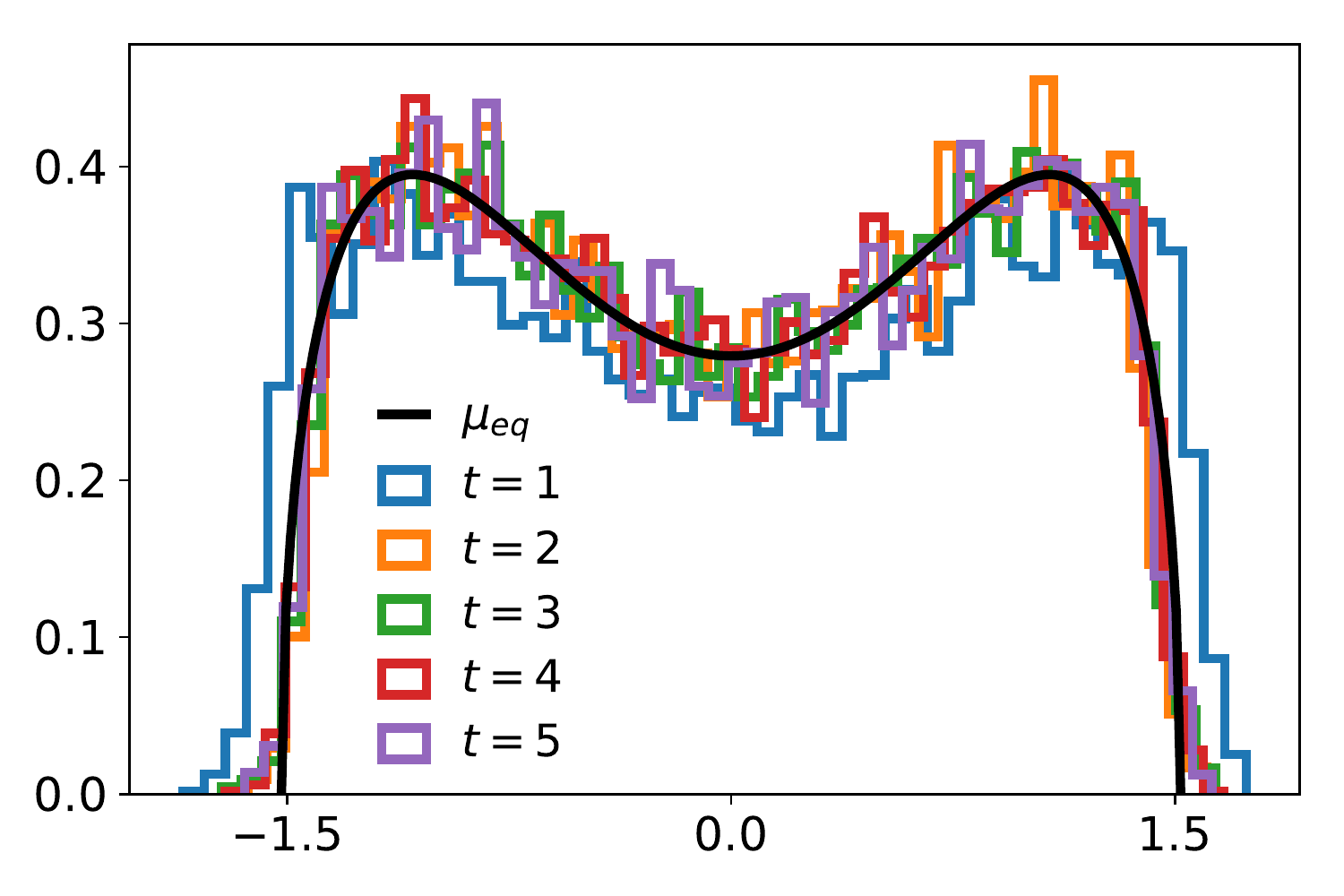}
                        \caption{$N=10$}
                        \label{fig:quartic_marginal_N=10}
                    \end{subfigure}
                    \hfill
                    \begin{subfigure}[b]{0.45\textwidth}
                        \centering
                        \includegraphics[width=\textwidth]{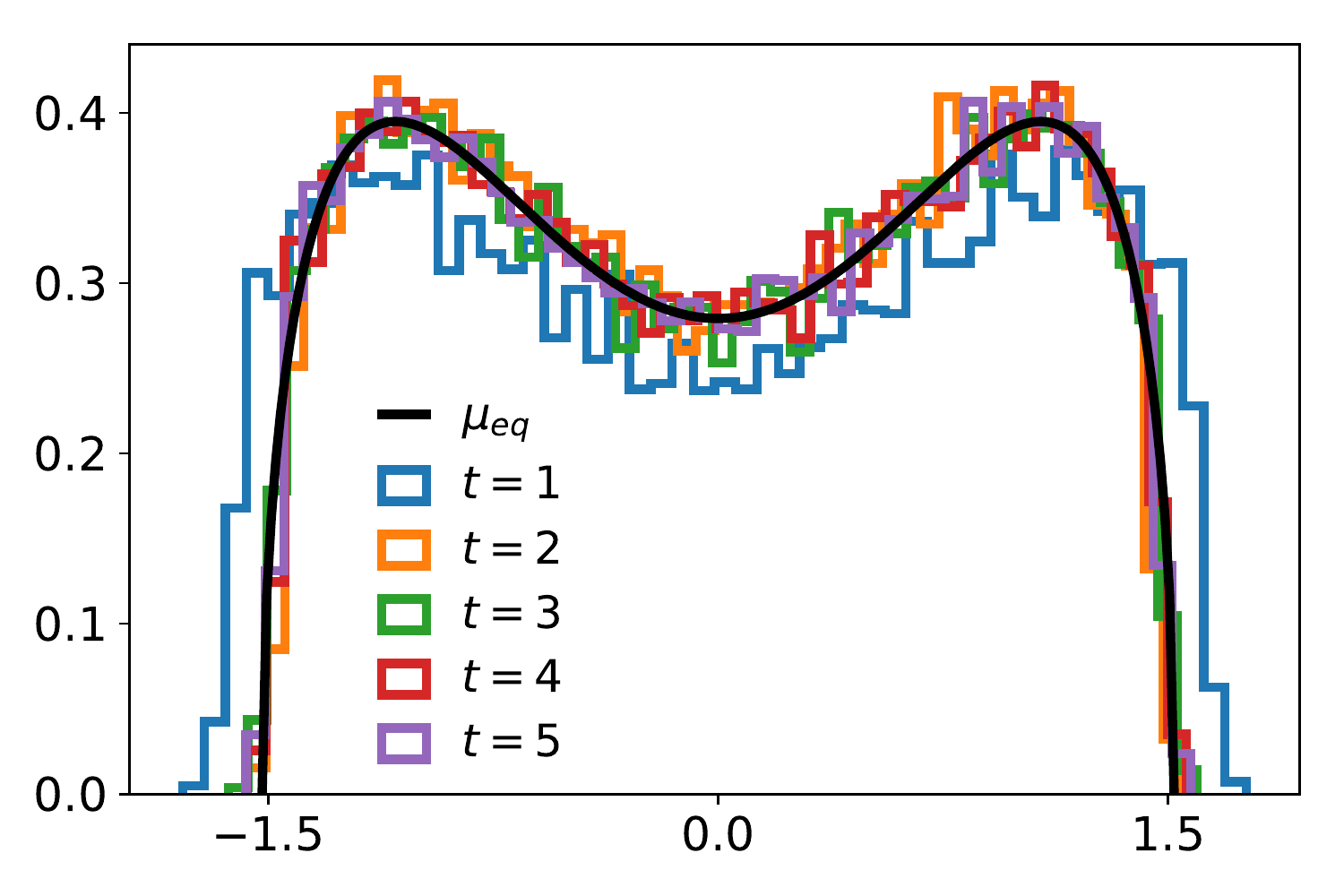}
                        \caption{$N=20$}
                        \label{fig:quartic_marginal_N=20}
                    \end{subfigure}
                    \\
                    \begin{subfigure}[b]{0.45\textwidth}
                        \centering
                        \includegraphics[width=\textwidth]{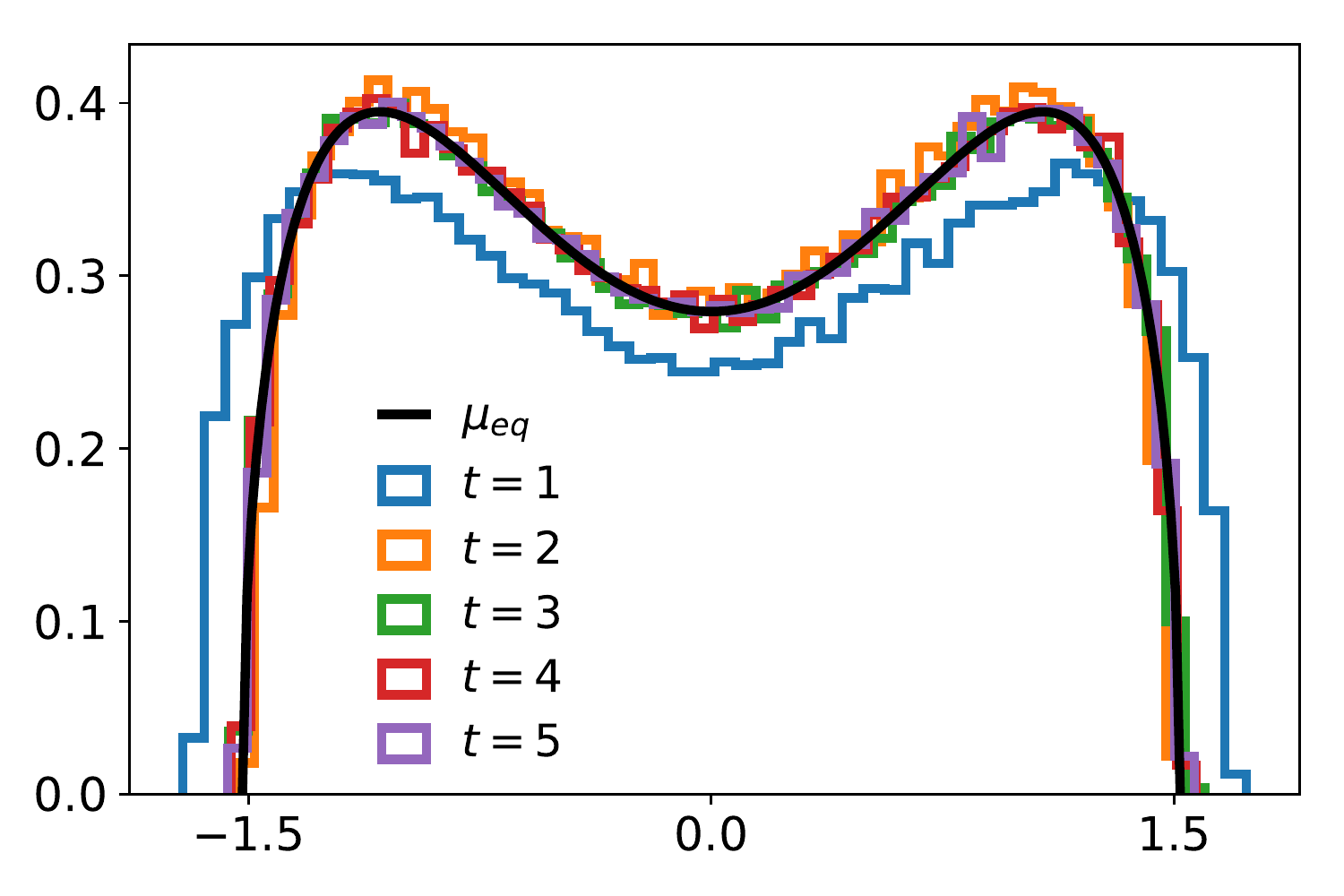}
                        \caption{$N=50$}
                        \label{fig:quartic_marginal_N=50}
                    \end{subfigure}
                    \hfill
                    \begin{subfigure}[b]{0.45\textwidth}
                        \centering
                        \includegraphics[width=\textwidth]{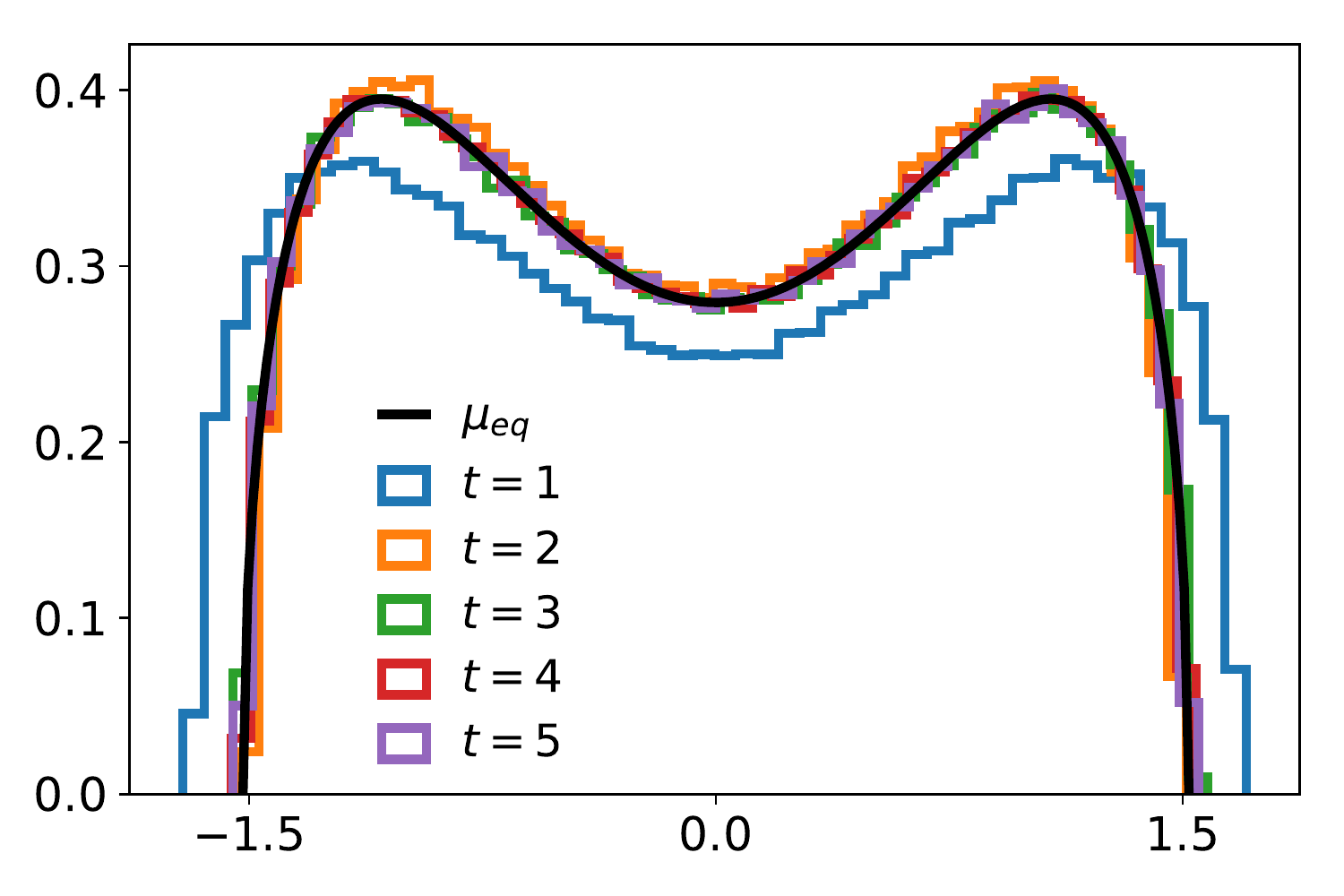}
                        \caption{$N=100$}
                        \label{fig:quartic_marginal_N=100}
                    \end{subfigure}
                    \\
                    \begin{subfigure}[b]{0.45\textwidth}
                        \centering
                        \includegraphics[width=\textwidth]{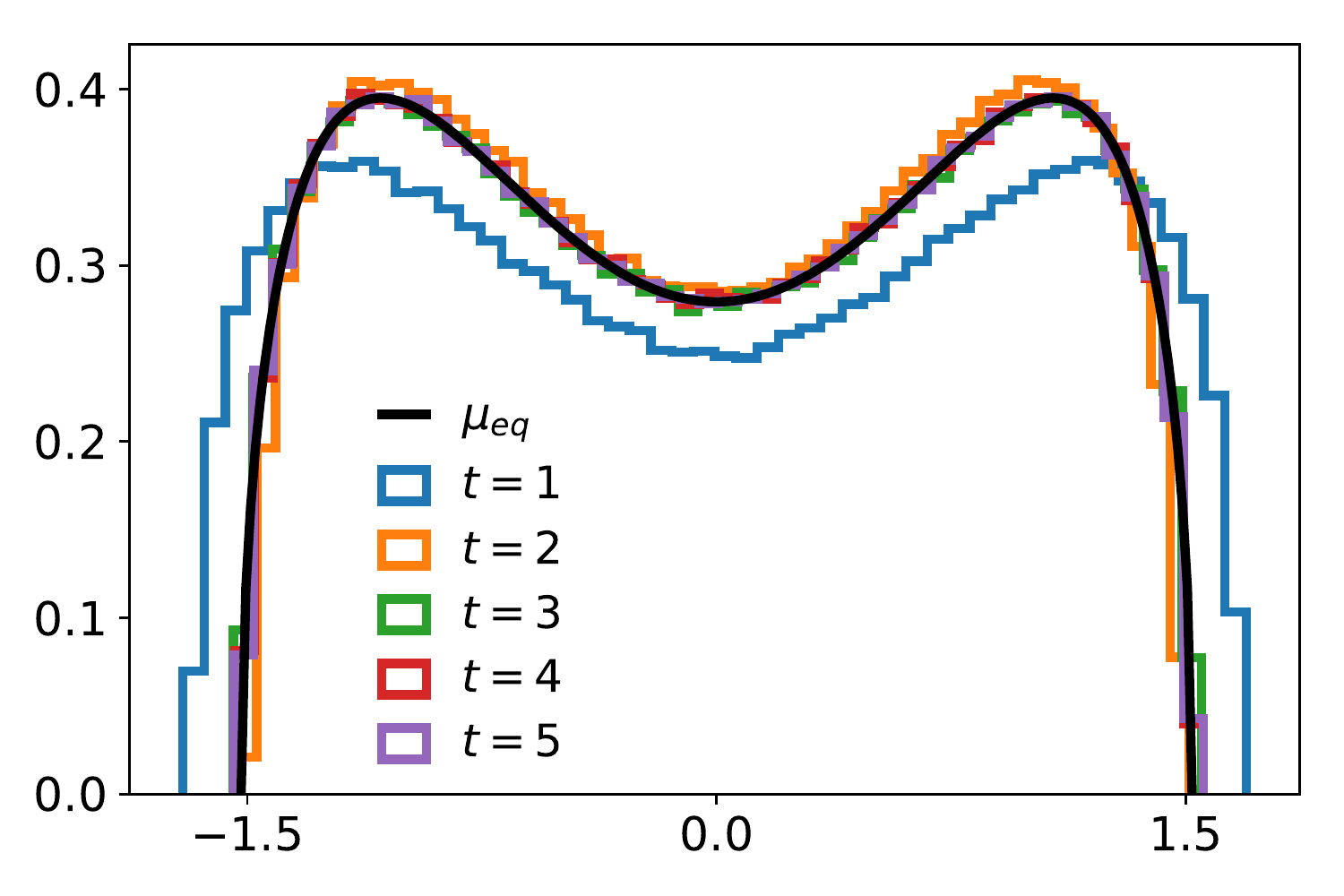}
                        \caption{$N=150$}
                        \label{fig:quartic_marginal_N=150}
                    \end{subfigure}
                    \hfill
                    \begin{subfigure}[b]{0.45\textwidth}
                        \centering
                        \includegraphics[width=\textwidth]{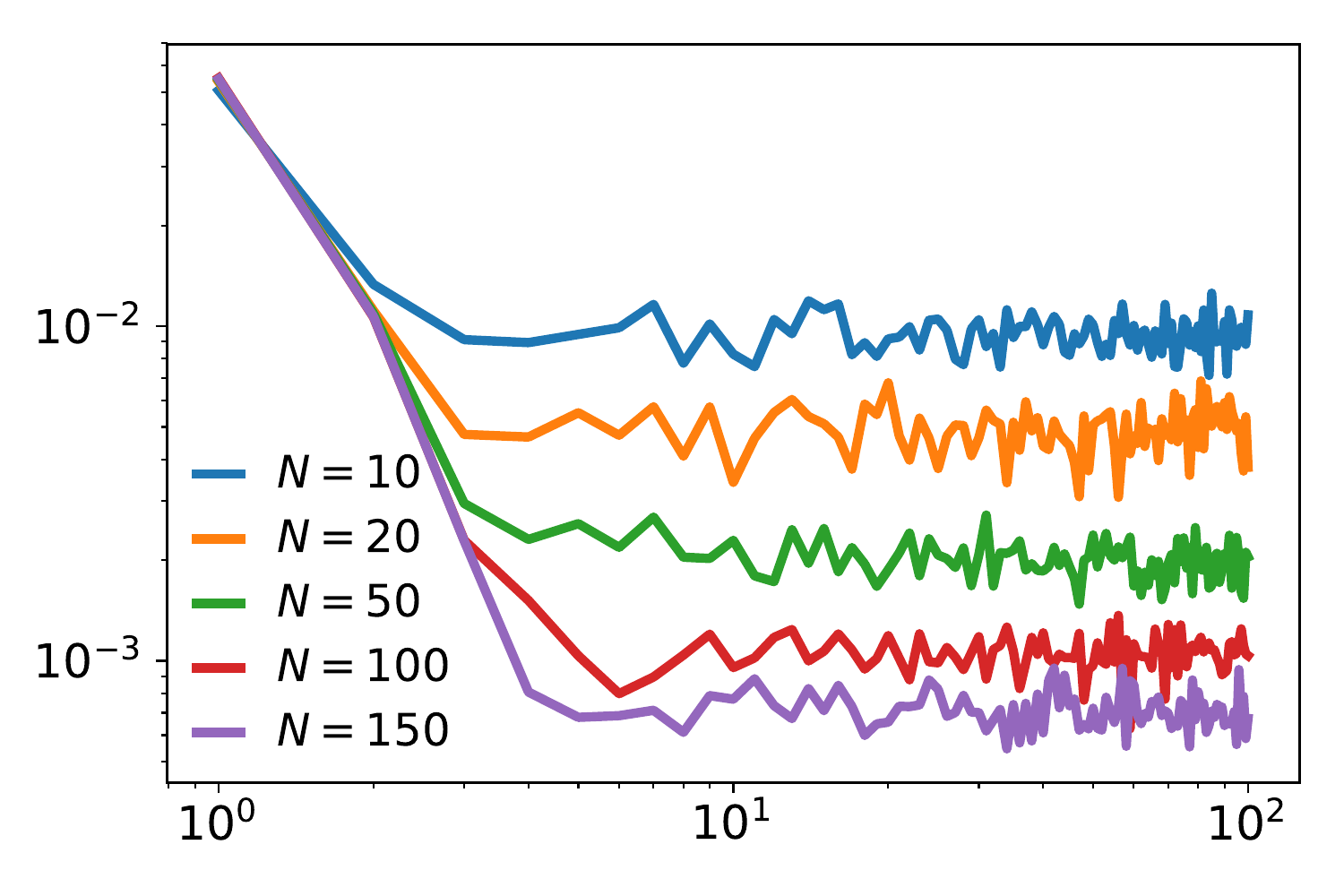}
                        \caption{$N=150$}
                        \label{fig:quartic_marginal_cdf_KS}
                    \end{subfigure}
                    \caption{
                    For $\beta=2$ and $V(x)=\frac14 x^4$, panels~\subref{fig:quartic_marginal_N=10}-\subref{fig:quartic_marginal_N=150} give a visual display of the convergence of the empirical marginal distribution $\mu_N^t$ of the eigenvalues constructed from $1000$ independent chains.
                    Each colored line corresponds to a Gibbs pass $t\in\{1,2,3,...\}$, while the equilibrium pdf is shown as a black line on each panel.
                    Different panels correspond to increasing values of $N$.
                    Panel~\subref{fig:quartic_marginal_cdf_KS} shows the supremum norm of the difference between the cdf of $\mueq$ and the cdf of $\hat\mu_N^t$ as a function of the number $t$ of Gibbs passes.}
                    \label{fig:quartic_marginal_convergence}
                \end{figure}



            \paragraph{Other potentials of degree $4$.} 
            \label{par:other_potentials_of_degree_4}

                We also consider the potential $V (x) = \frac{1}{20} x^4 - \frac{4}{15}x^3 + \frac{1}{5}x^2 +\frac{8}{5}x$ and potentials of the form $V(x)=g_2 x^2 + \frac14 x^4$ where we vary $g_2$.
                Except the case where $g_2 \geq 0$, the conditionals $a_n \mid \bfa_{\setminus n}, \bfb$ are not log concave and we sample from them using a few steps of MALA.
                This allows us to select various qualitative behaviors of $\mueq$, which may become dissymmetric \citep[Example 1.2; Section 3.2]{ClKrIt09,OlNaTr15}, or supported by more than one connected component \citep[Figure 4]{Mol18}.
                Our approach allows to simulate from the corresponding $\beta$-ensembles in regimes yet unexplored.
                \Figref{fig:other_quartic_marginal} shows good agreement of marginal histograms of a single sample of $N=1000$ points   with the equilibrium distribution after only $t=10$ Gibbs passes.

                \begin{figure}[!ht]
                    \centering
                    \begin{subfigure}[b]{0.48\textwidth}
                        \centering
                        \includegraphics[width=\textwidth]{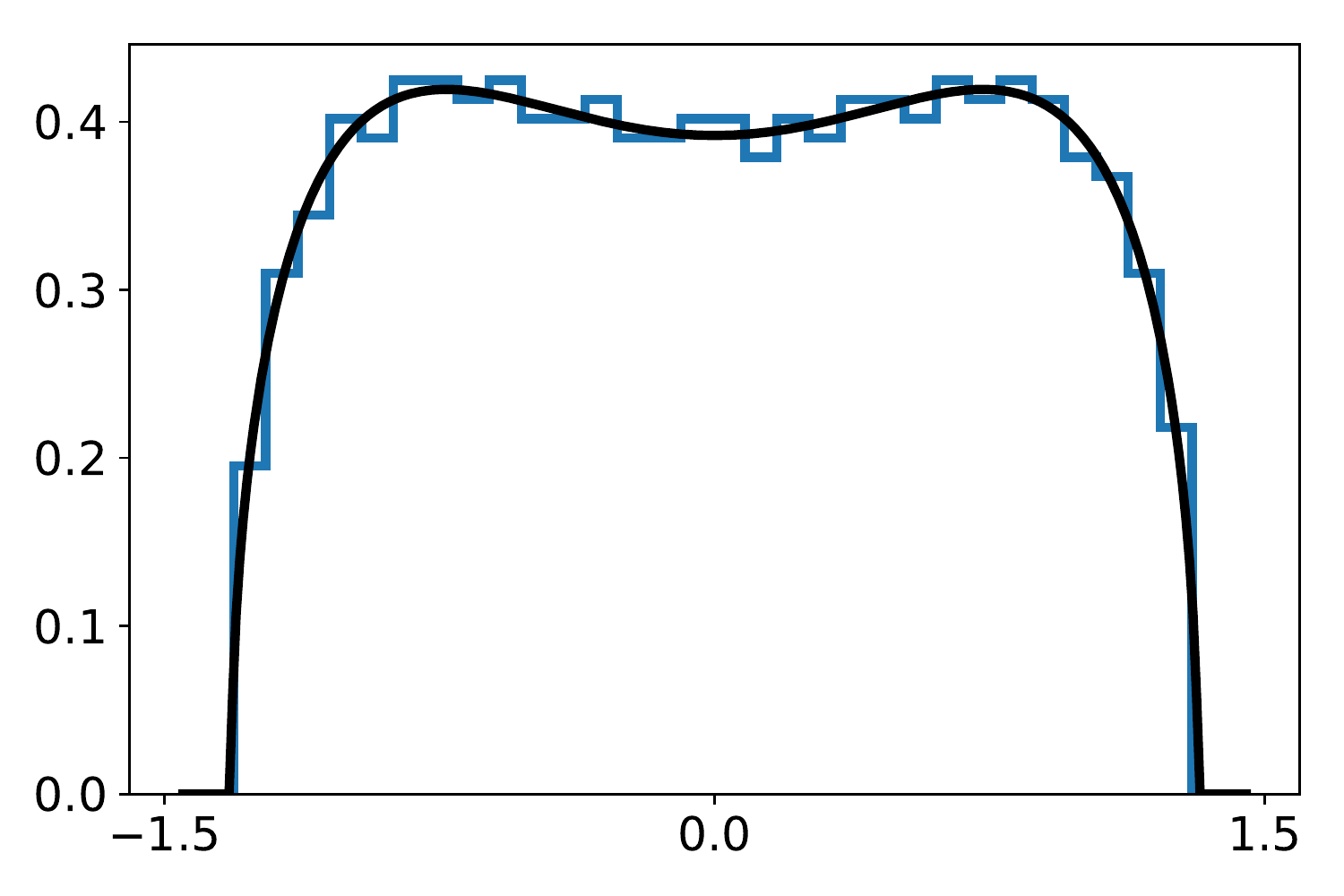}
                        \caption{$V(x)=\frac{1}{4} x^4 + \frac{1}{2} x^2$}
                        \label{fig:quartic_x4_x2}
                    \end{subfigure}
                    \hfill
                    \begin{subfigure}[b]{0.48\textwidth}
                        \centering
                        \includegraphics[width=\textwidth]{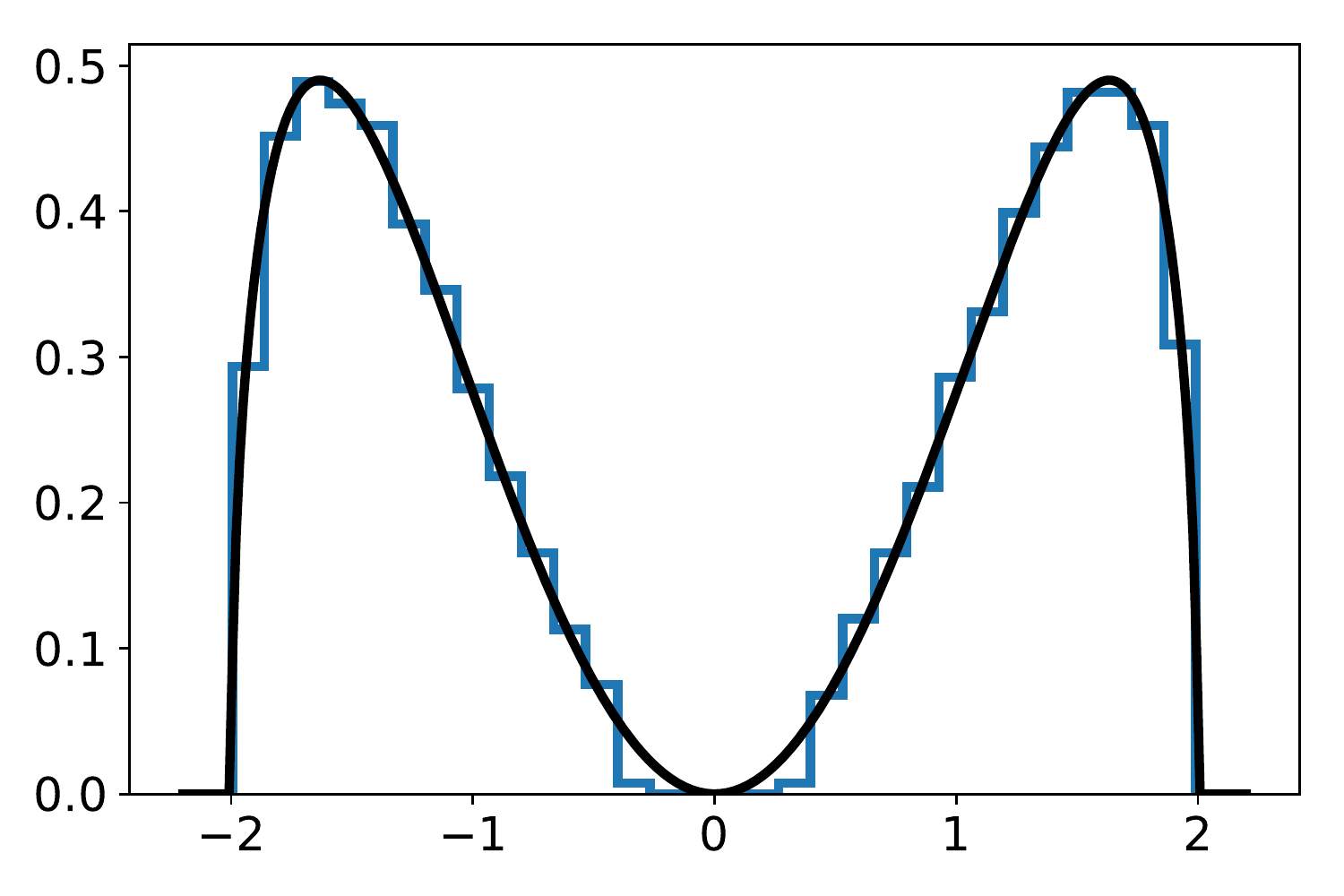}
                        \caption{$V(x)=\frac{1}{4} x^4 - x^2$}
                        \label{fig:quartic_x4_x2_onset_2cut}
                    \end{subfigure}
                    \\
                    \begin{subfigure}[b]{0.48\textwidth}
                        \centering
                        \includegraphics[width=\textwidth]{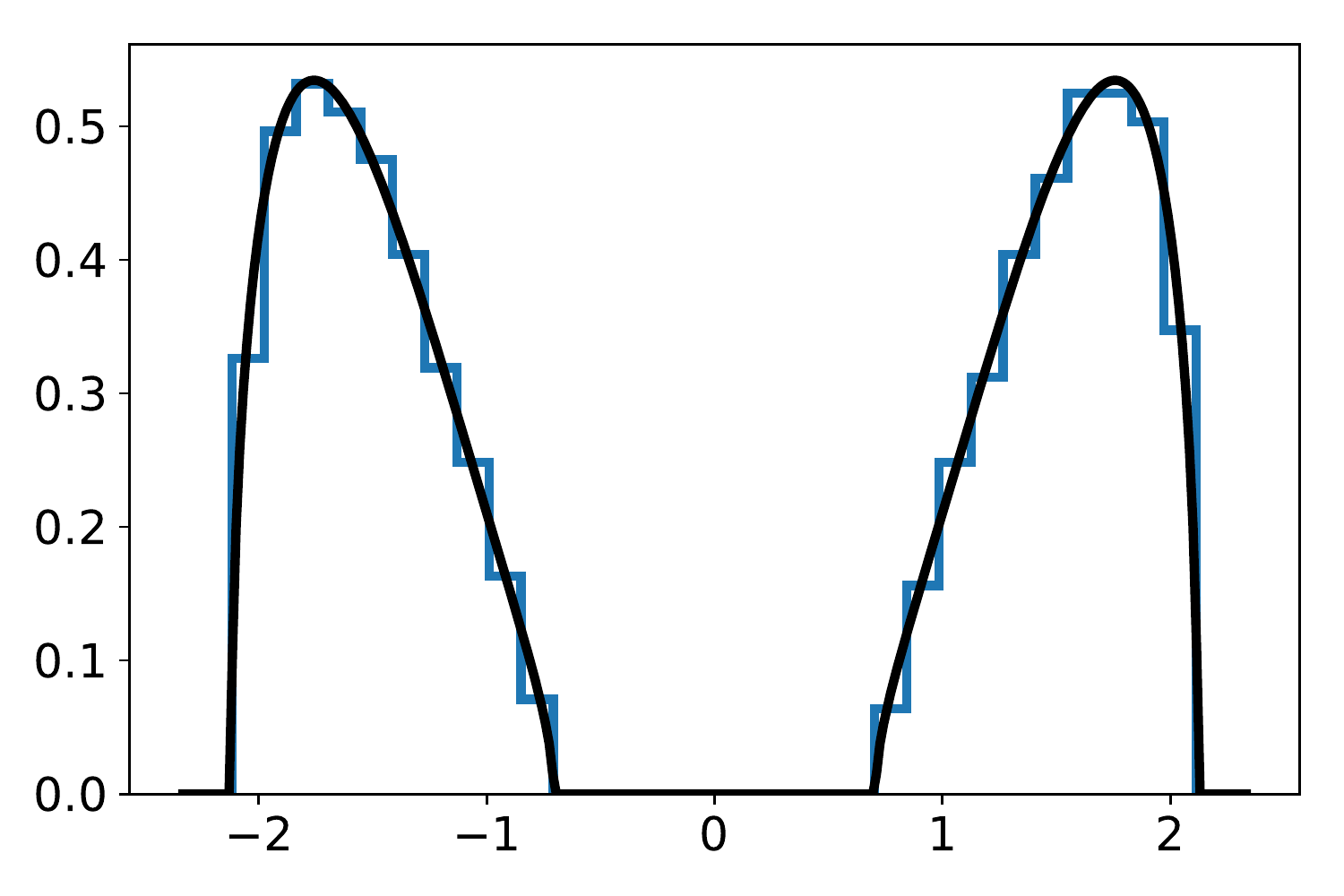}
                        \caption{$V(x)=\frac{1}{4} x^4 -\frac{5}{4} x^2$}
                        \label{fig:quartic_x4_x2_2cut}
                    \end{subfigure}
                    \hfill
                    \begin{subfigure}[b]{0.48\textwidth}
                        \centering
                        \includegraphics[width=\textwidth]{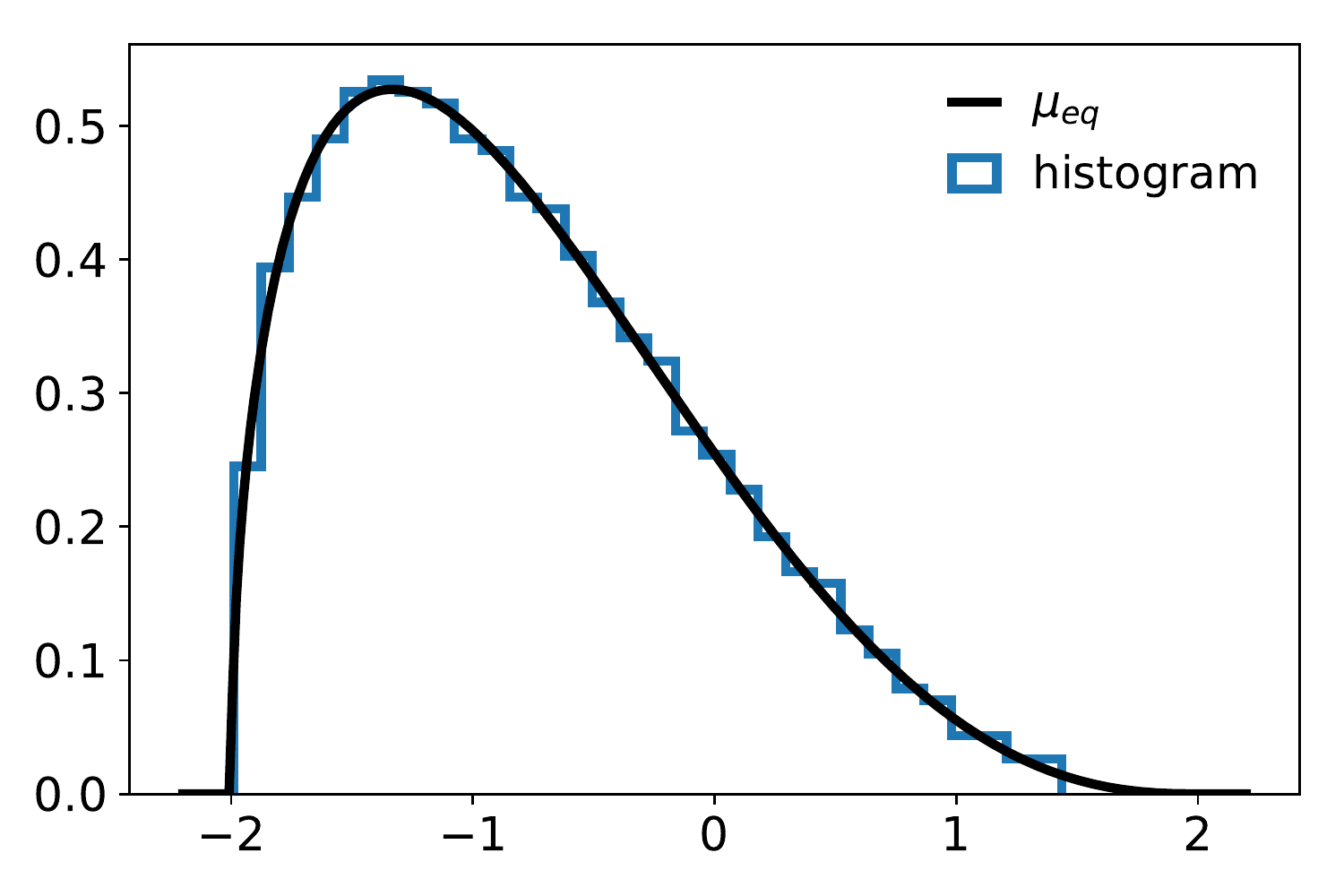}
                        \caption{$V (x) = \frac{1}{20} x^4 - \frac{4}{15}x^3 + \frac{1}{5}x^2 +\frac{8}{5}x$}
                        \label{fig:quartic_x4_x3_x2_x}
                    \end{subfigure}
                    \caption{
                    For various choices of potentials $V$ of degree $4$ and $\beta=2$, each panel shows the histogram of a sample from $\mu_N^t$, with $N=1000$ points after $t=10$ Gibbs passes. The corresponding equilibrium measures are superimposed in black.}
                    \label{fig:other_quartic_marginal}
                \end{figure}


            \vspace{3em}
            \paragraph{The sextic potential.} 
            \label{par:the_sextic_potential}

                We extend the derivations of \Exref{ex:quartic} and consider sampling from the $\beta$-ensemble with potential $V(x)=\frac{1}{6} x^6$.
                The corresponding equilibrium distribution can be derived from \citet[Proposition 6.156]{Dei00}.
                In this case, the conditionals $a_n \mid \bfa_{\setminus n}, \bfb$ are not log-concave and we cannot use the exact rejection sampler of \citet{Dev12}.
                Instead, we switch to a Metropolis-within-Gibbs sampler and make a few steps of MALA; see \Secref{sub:sampling_from_the_conditionals}.

                One free parameter is the number of MALA steps in one Gibbs pass.
                We empirically observed (not shown) that this number has an influence on the number of Gibbs passes needed to reach the plateau in \Figref{fig:sextic_marginal_cdf_KS}. Manually setting the number of MALA steps per Gibbs pass to 100 was enough for rapid overall convergence in our experiments, and larger values did not significantly influence the fit in \Figref{fig:sextic_marginal_convergence}, which is already striking after less than 10 Gibbs passes.

                \begin{figure}[ht]
                    \centering
                    \begin{subfigure}[b]{0.49\textwidth}
                        \centering
                        \includegraphics[width=\textwidth]{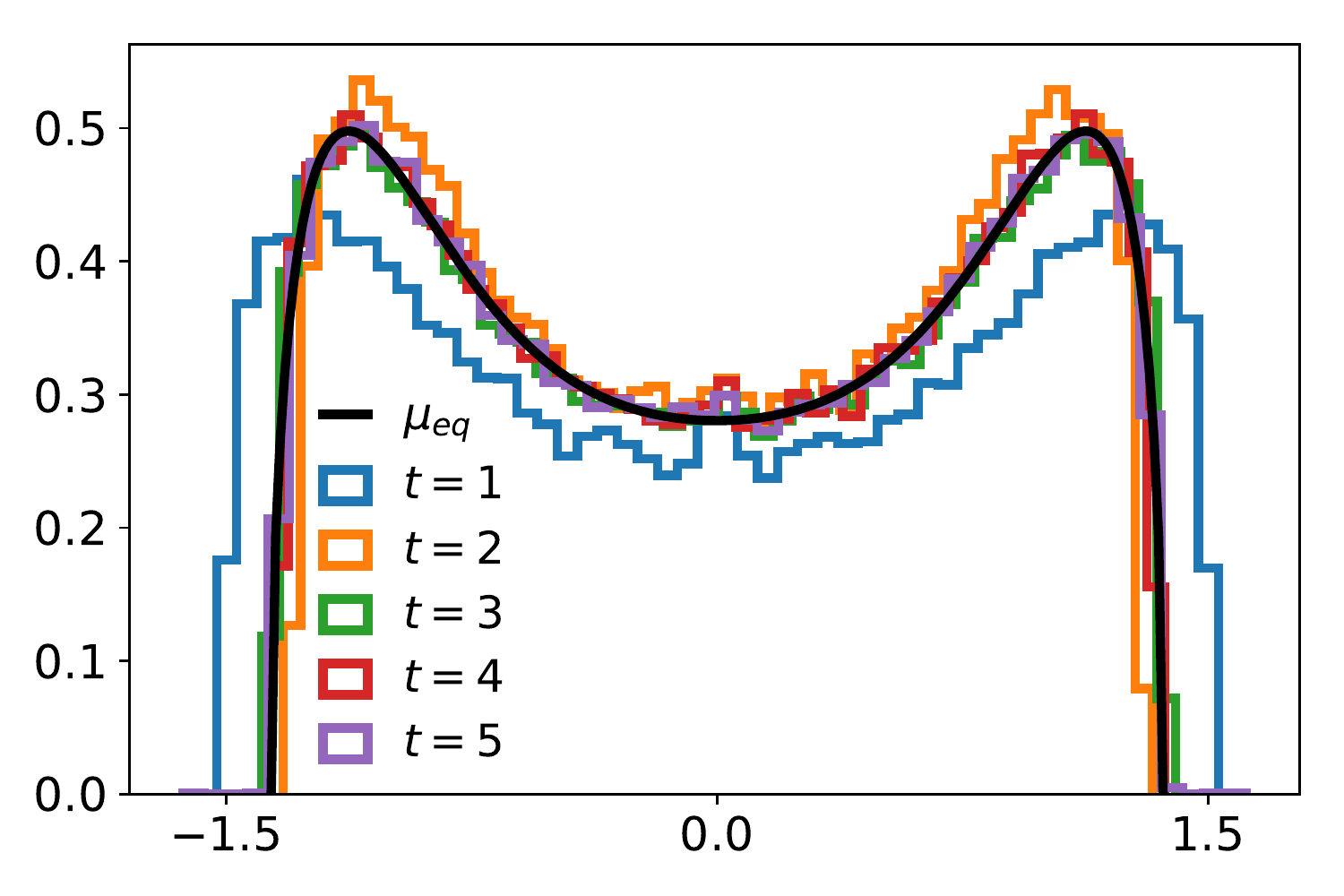}
                        \caption{$N=50$}
                        \label{fig:sextic_marginal_N=50}
                    \end{subfigure}
                    \hfill
                    \begin{subfigure}[b]{0.49\textwidth}
                        \centering
                        \includegraphics[width=\textwidth]{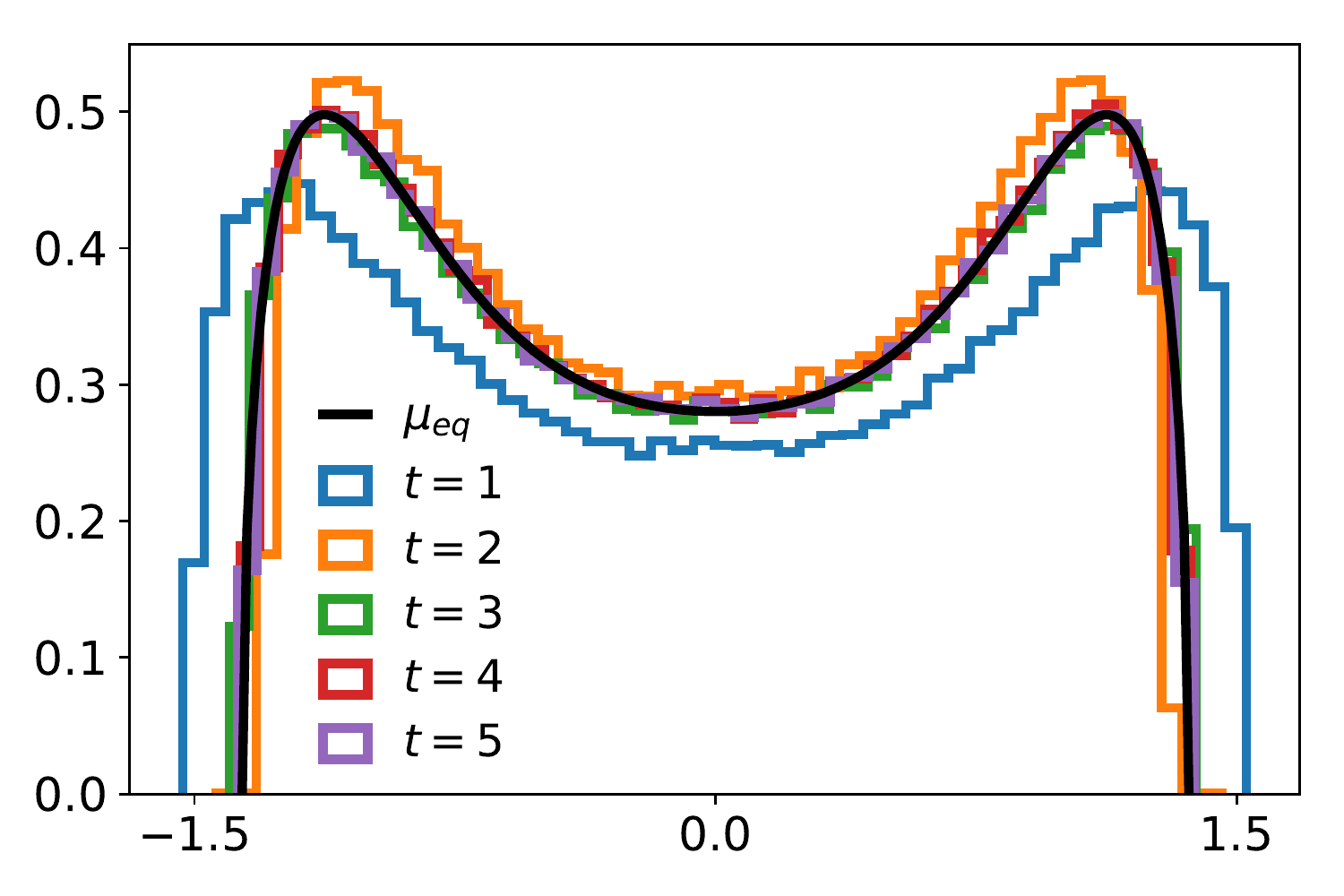}
                        \caption{$N=100$}
                        \label{fig:sextic_marginal_N=100}
                    \end{subfigure}\\
                    \begin{subfigure}[b]{0.49\textwidth}
                        \centering
                        \includegraphics[width=\textwidth]{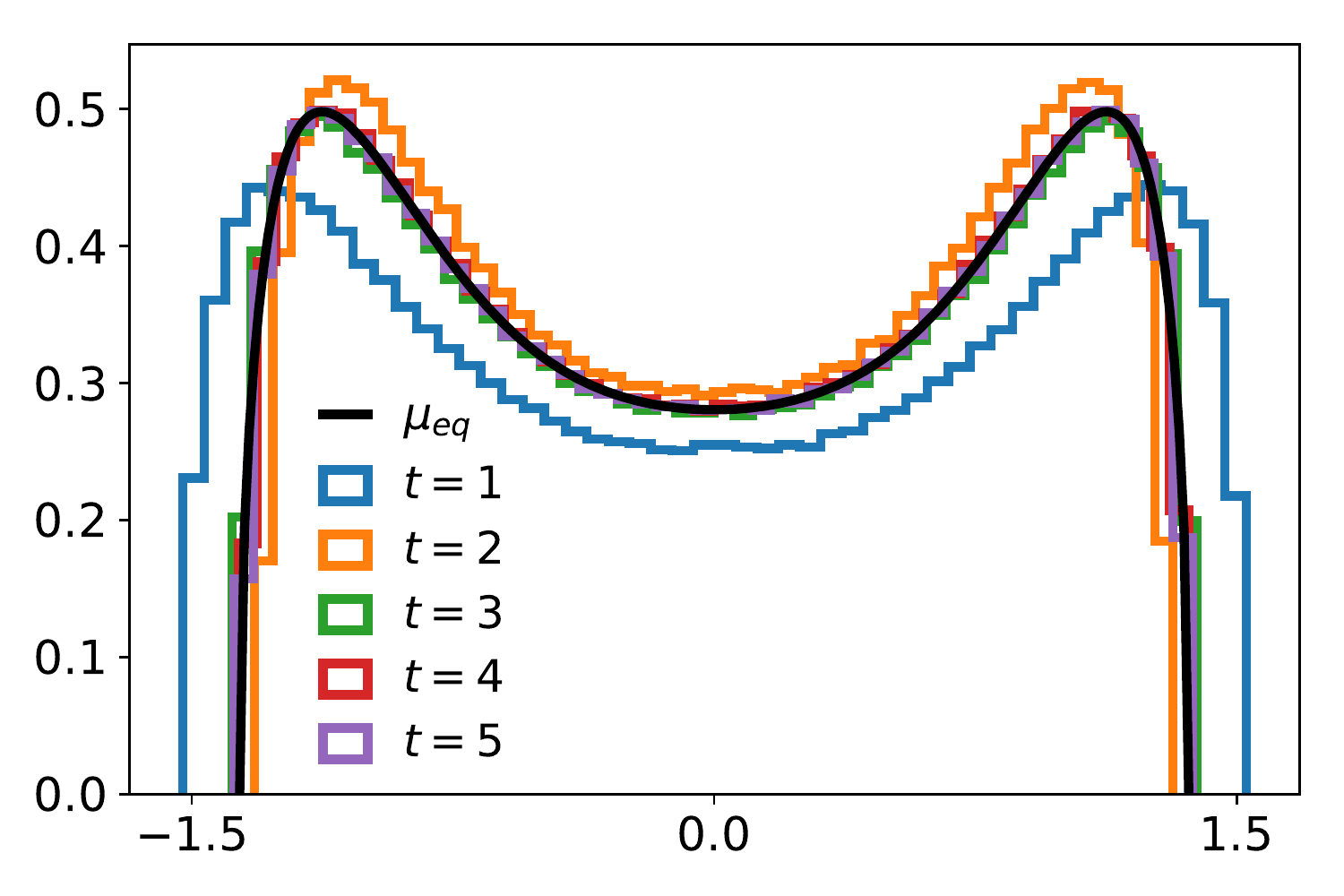}
                        \caption{$N=200$}
                        \label{fig:sextic_marginal_N=200}
                    \end{subfigure}
                    \hfill
                    \begin{subfigure}[b]{0.49\textwidth}
                        \centering
                        \includegraphics[width=\textwidth]{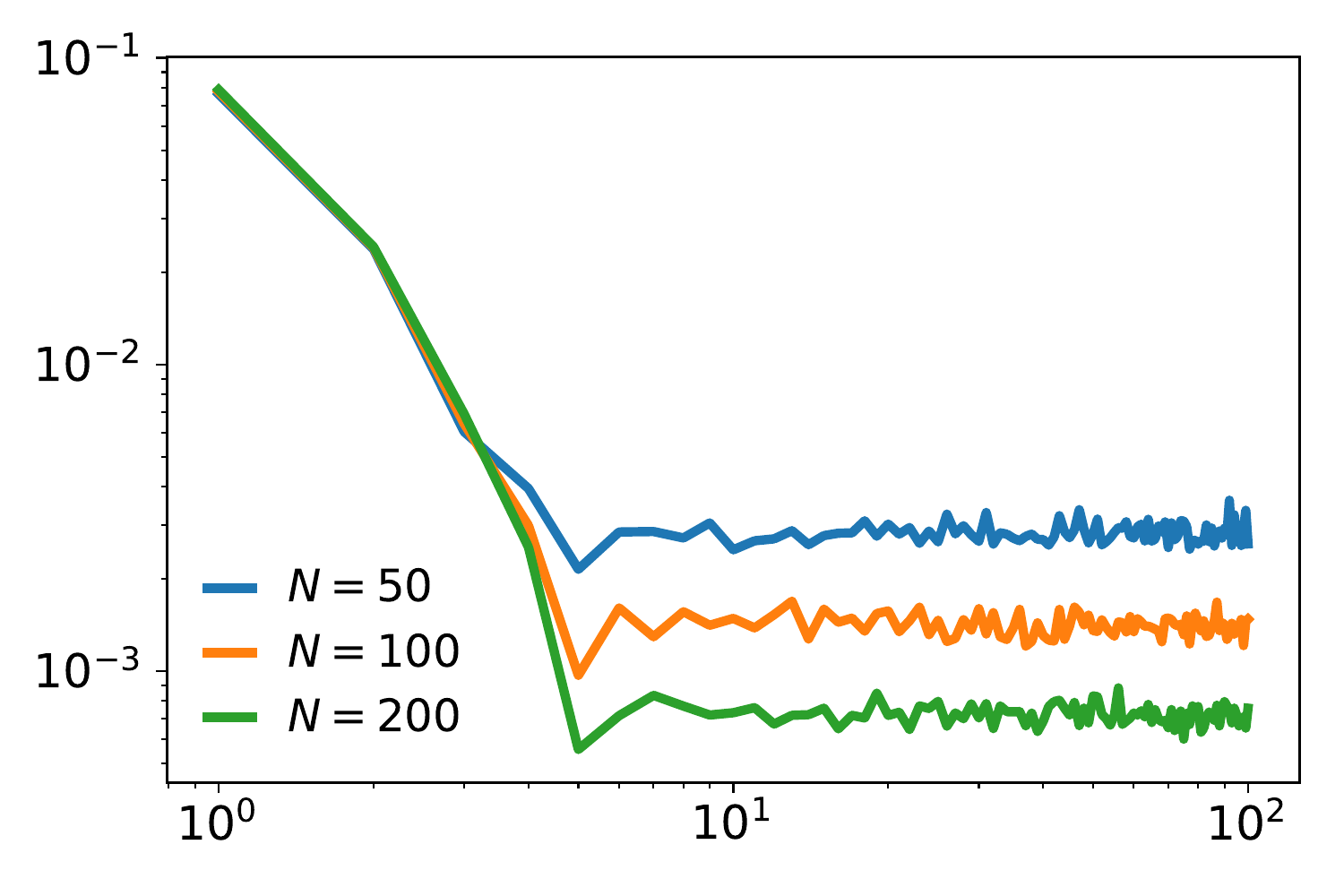}
                        \caption{Convergence of the empirical cdf}
                        \label{fig:sextic_marginal_cdf_KS}
                    \end{subfigure}
                    \caption{
                    For $\beta=2$ and $V(x)=\frac16 x^6$, panels~\subref{fig:sextic_marginal_N=50}-\subref{fig:sextic_marginal_N=200} give a visual display of the convergence of the empirical marginal distribution $\mu_N^t$ of the eigenvalues constructed from $1000$ independent chains.
                    Each colored line corresponds to a Gibbs pass $t\in\{1,2,3,...\}$, while the equilibrium pdf is shown as a black line on each panel.
                    Different panels correspond to increasing values of $N$.
                    Panel~\subref{fig:sextic_marginal_cdf_KS} shows the supremum norm of the difference between the cdf of $\mueq$ and the cdf of $\hat\mu_N^t$ as a function of the number $t$ of Gibbs passes.
                    }
                    \label{fig:sextic_marginal_convergence}
                \end{figure}



        \subsubsection{Fluctuations of the largest eigenvalue} 
        \label{ssub:fluctuations_of_the_largest_eigenvalue}

            After looking at the global behavior of the eigenvalues, we zoom at the right edge of the support of $\mu$ to study the local behavior of our approximate samples.
            We do this for the quartic and sextic potentials.
            When $\beta=2$, the target $\beta$-ensemble is determinantal and we can test the adequation of the largest atom of $\hat\mu_N^t$ to the universal Tracy-Widom limiting distribution \citep[Corollary 1.3]{DeGi05}.
            We implemented the cumulative distribution function (cdf) of the Tracy-Widom law following the work of \citet{Bor09}, and rescale the eigenvalues as \citet[Section 3.2]{OlTr14}.

            For each potential, we run $1000$ independent chains and record only the largest eigenvalue of each chain after each Gibbs pass.
            \Figref{fig:quartic_lambda_max_visual} and \Figref{fig:sextic_lambda_max_visual} show the histograms of the rescaled largest particles after a few Gibbs passes, respectively for the quartic and sextic potential.
            More quantitatively, in \Figref{fig:quartic_sextic_lambda_max} we monitor the convergence to the Tracy-Widom distribution across Gibbs passes, by computing the supremum distance between the empirical cdf of the largest eigenvalue and the cdf of the Tracy-Widom law.

            For the quartic potential, we observe that the adequation with the cdf of the Tracy-Widom law gets tighter as $N$ grows, and again only a few passes of the Gibbs sampler are sufficient to reach a plateau.

            In contrast, for the sextic ensemble there seems to be an impassable gap, as if the rescaling was not adequate or the Tracy-Widom law was not the proper limiting distribution.
            This is despite the square root singularity at the right edge of the equilibrium distribution \citep[Section 6.1]{Dei00}.
            In particular, a simple Kolmogorov-Smirnov test at level $0.05$ would reject the adequation to Tracy-Widom.
            Part of this effect might be due to the fact that the conditionals are not sampled exactly in the sextic case, though.

            \begin{figure}[ht]
                \centering
                \begin{subfigure}[b]{0.49\textwidth}
                    \centering
                    \includegraphics[width=\textwidth]{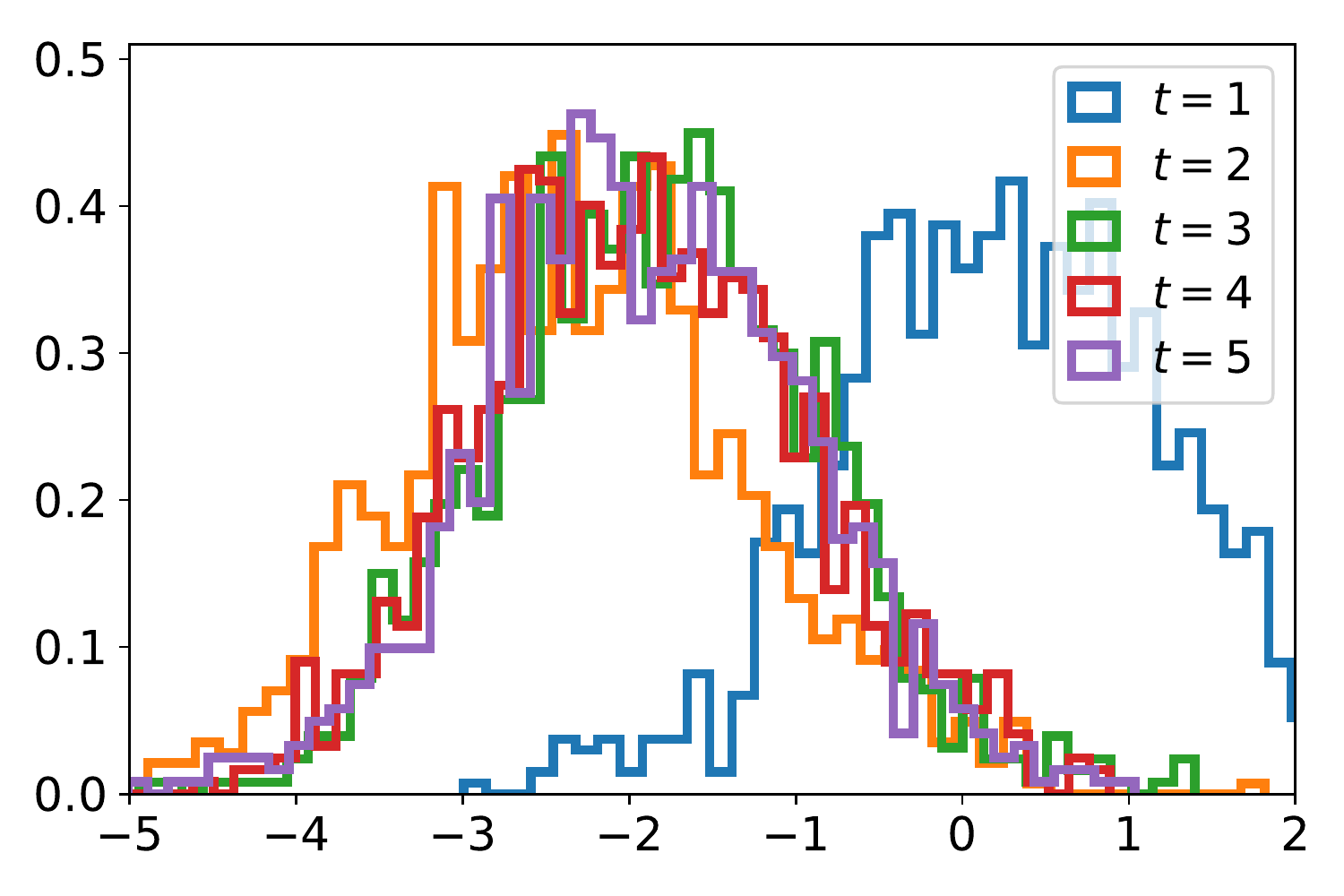}
                    \caption{$N=20$}
                    \label{fig:quartic_lambda_max_N=20}
                \end{subfigure}
                \begin{subfigure}[b]{0.49\textwidth}
                    \centering
                    \includegraphics[width=\textwidth]{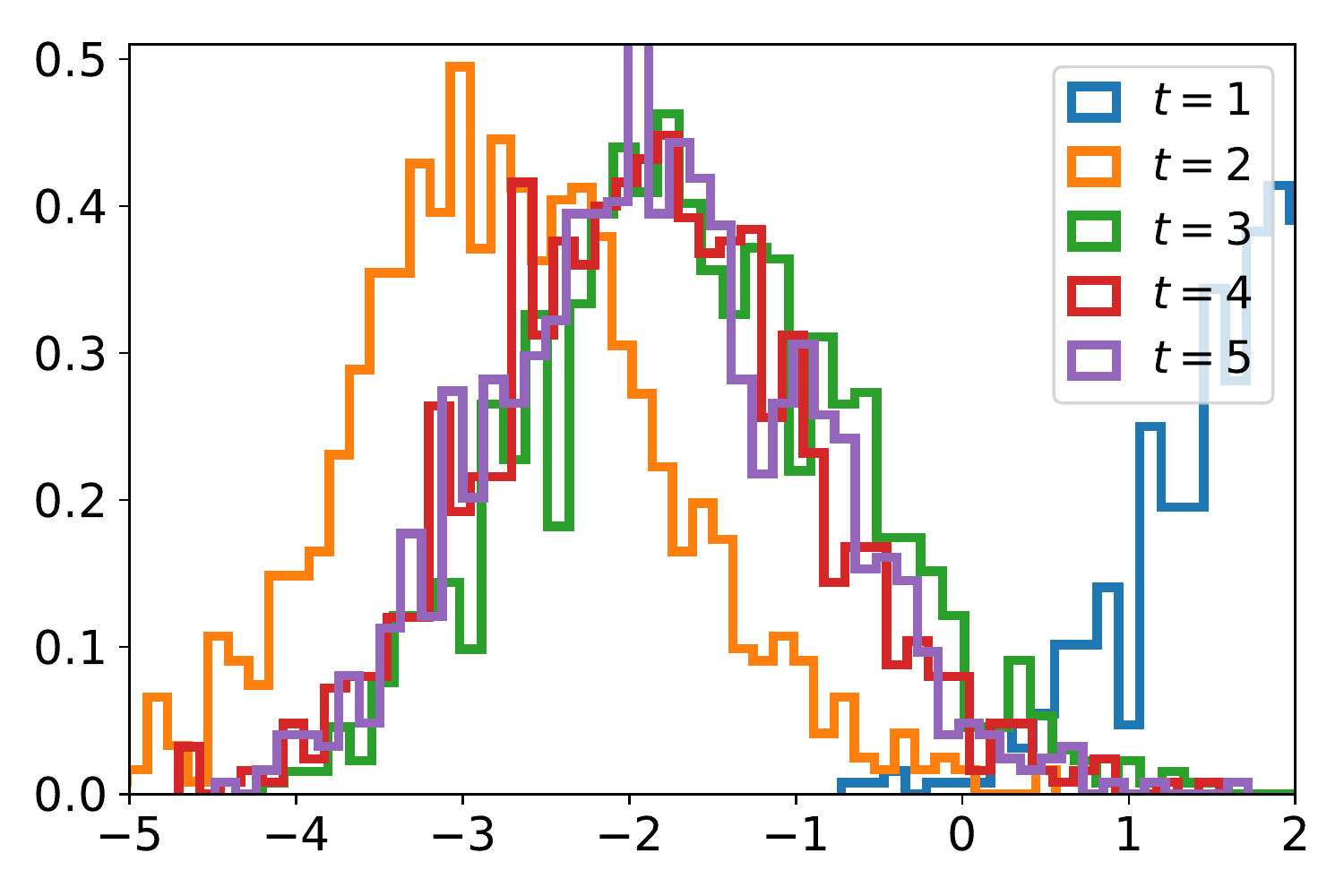}
                    \caption{$N=50$}
                    \label{fig:quartic_lambda_max_N=50}
                \end{subfigure}
                \\
                \begin{subfigure}[b]{0.49\textwidth}
                    \centering
                    \includegraphics[width=\textwidth]{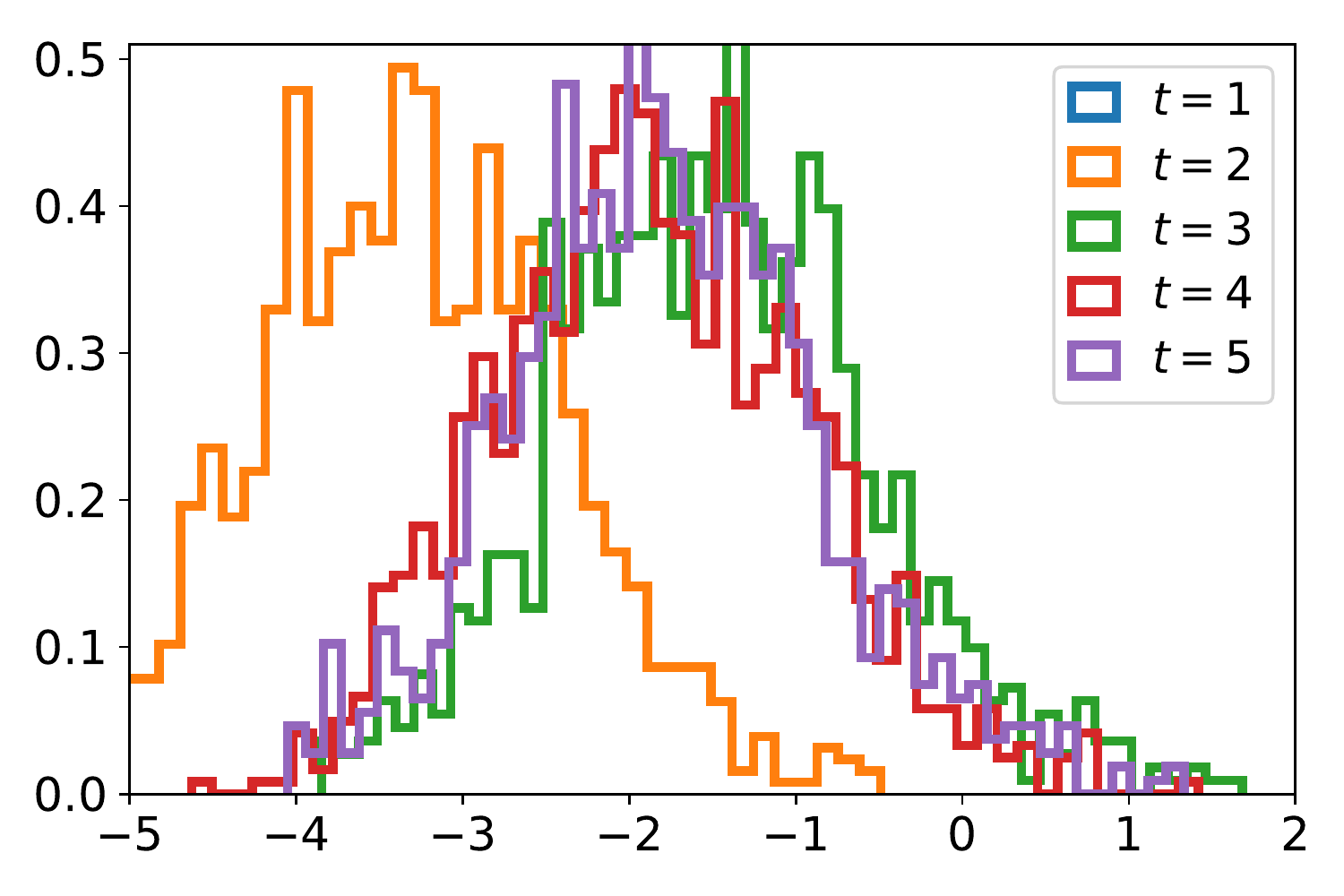}
                    \caption{$N=100$}
                    \label{fig:quartic_lambda_max_N=100}
                \end{subfigure}
                \begin{subfigure}[b]{0.49\textwidth}
                    \centering
                    \includegraphics[width=\textwidth]{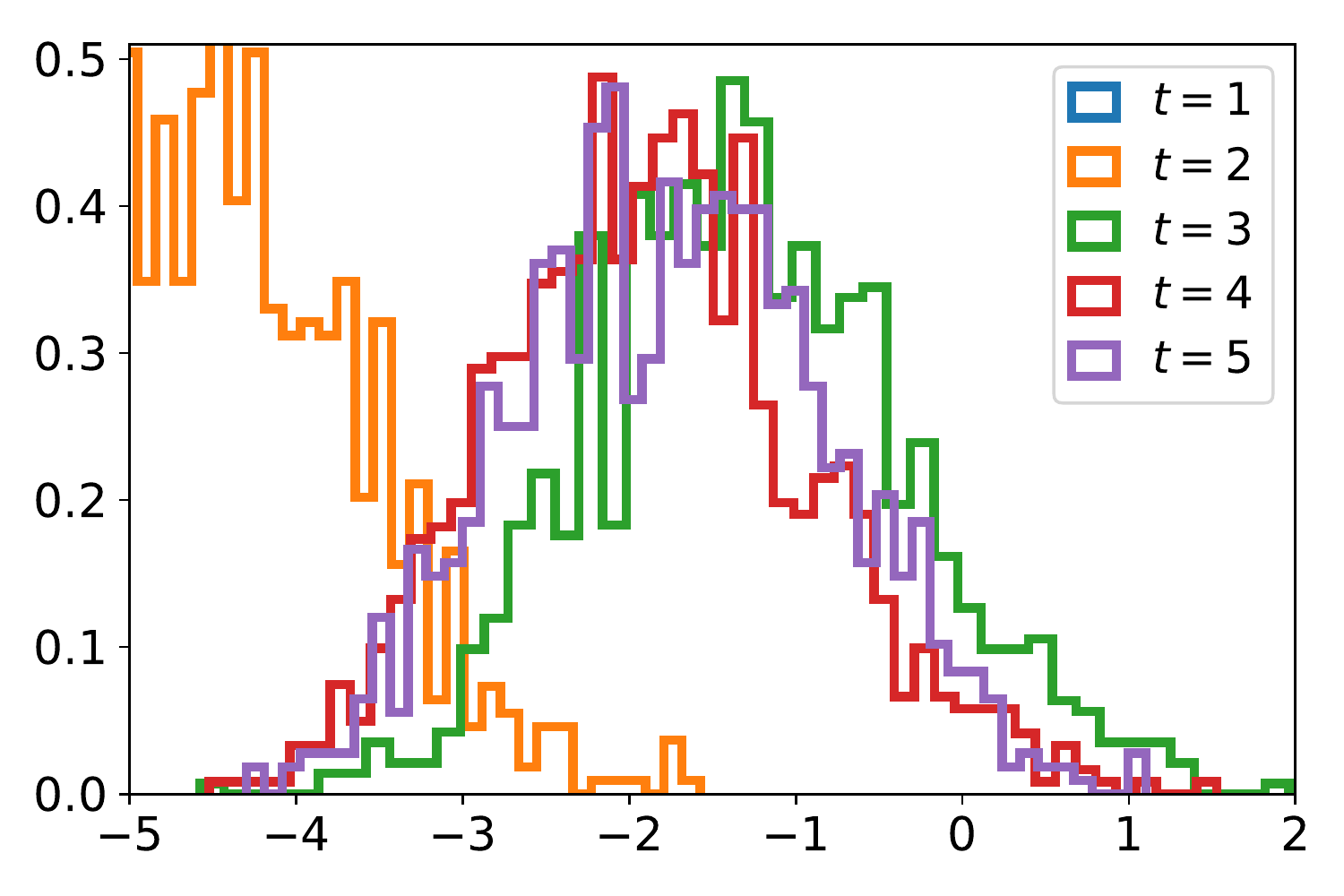}
                    \caption{$N=250$}
                    \label{fig:quartic_lambda_max_N=250}
                \end{subfigure}
                \caption{%
                    For $\beta=2$ and $V(x)=\frac14 x^4$, we give a visual display of the convergence of the empirical distribution of the largest atom of $\mu_N^t$ constructed from $1000$ independent chains to the Tracy-Widom distribution.
                    Each colored line corresponds to a Gibbs pass $t\in\{1,2,3,...\}$.
                    Different panels correspond to increasing values of $N$.
                    As $N$ increases the histogram of passes $t=1$ and $t=2$ are farther from matching the high density regions of the Tracy-Widom law but then the fit is good and fast; see also Figures~\ref{fig:quartic_lambda_max_cdf_convergence_visual}-\subref{fig:quartic_lambda_max_cdf_smoothed} for a more quantitative monitoring.
                }
                \label{fig:quartic_lambda_max_visual}
            \end{figure}

            \begin{figure}[ht]
                \centering
                \begin{subfigure}[b]{0.49\textwidth}
                    \centering
                    \includegraphics[width=\textwidth]{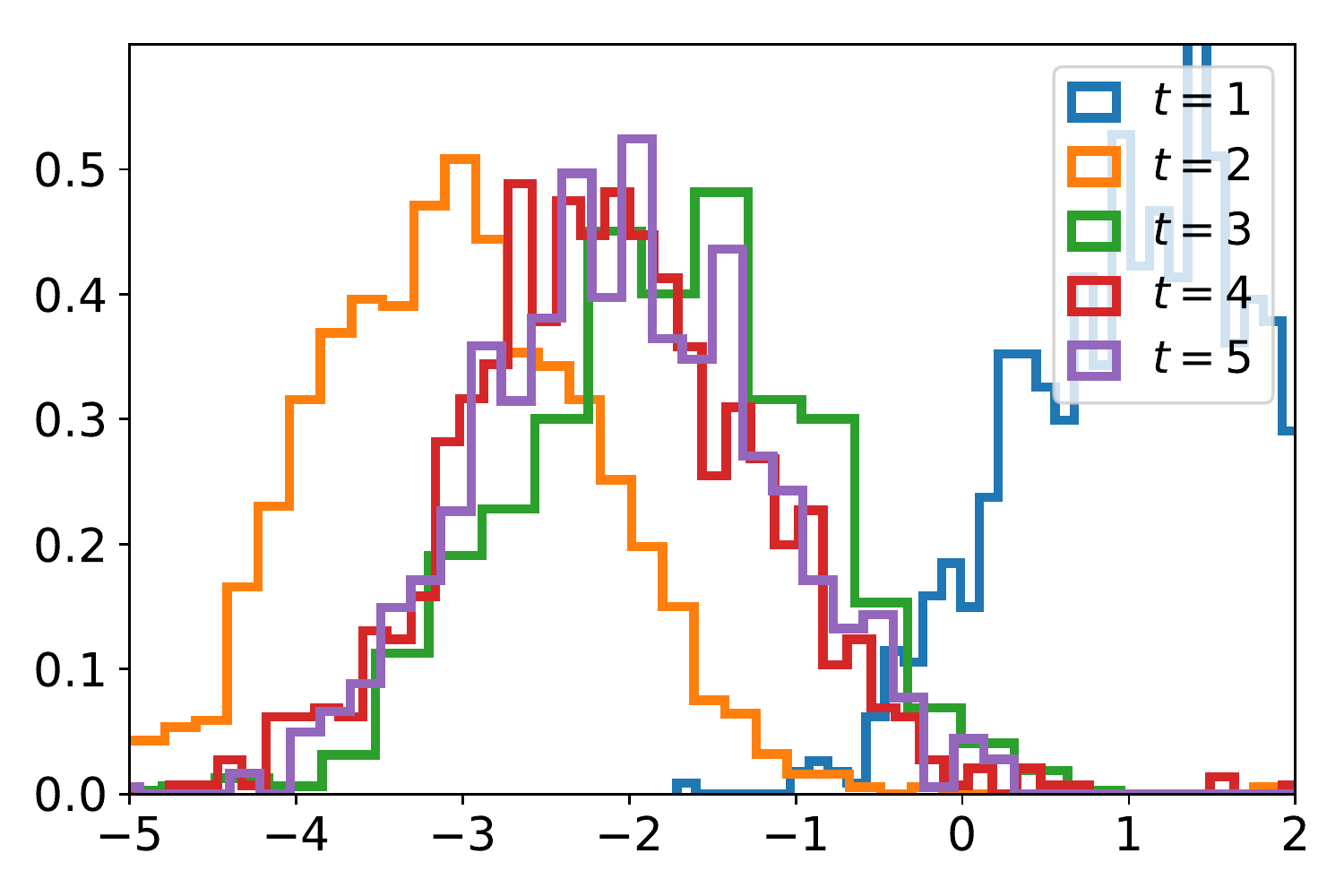}
                    \caption{$N=20$}
                    \label{fig:sextic_lambda_max_N=20}
                \end{subfigure}
                \hfill
                \begin{subfigure}[b]{0.49\textwidth}
                    \centering
                    \includegraphics[width=\textwidth]{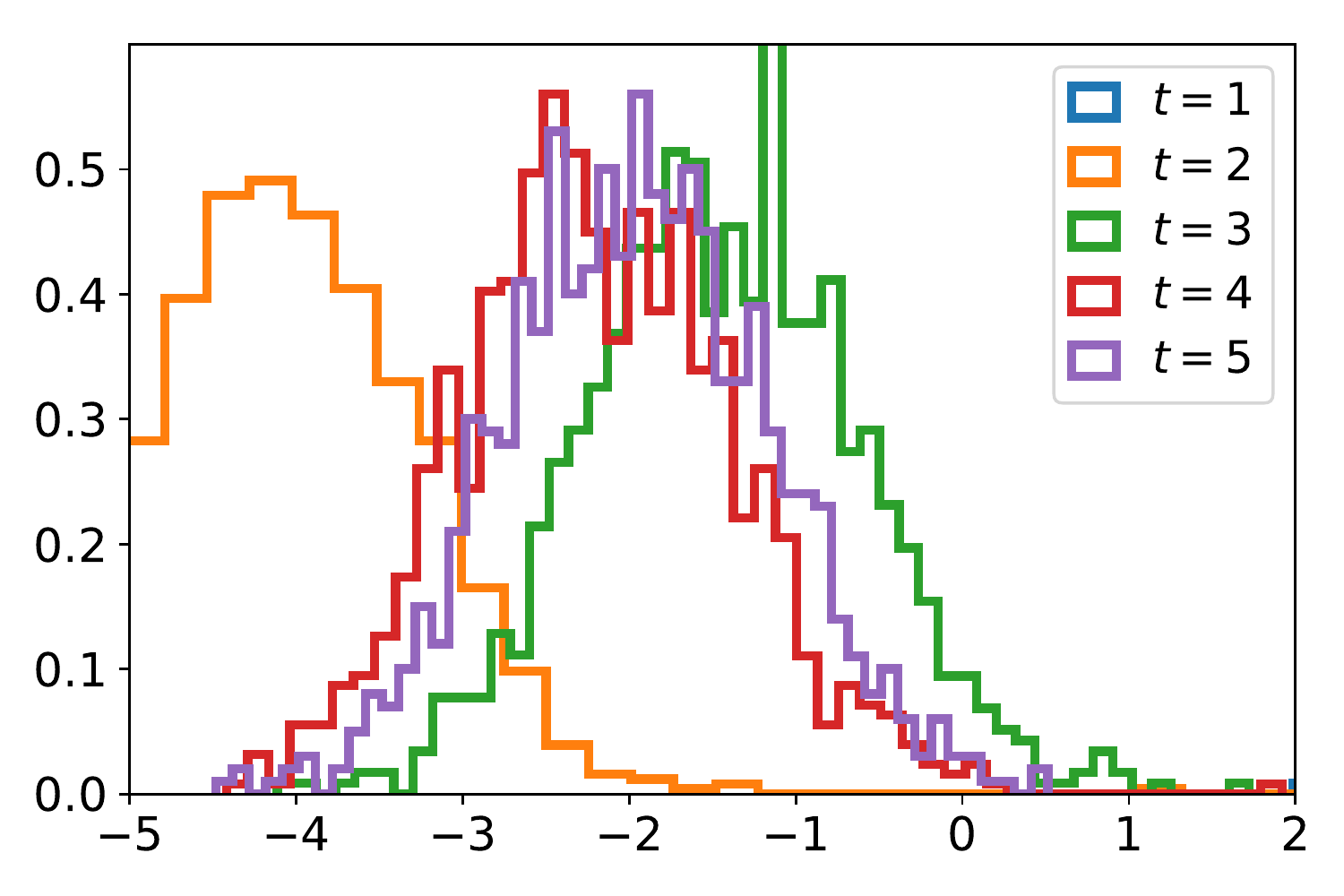}
                    \caption{$N=50$}
                    \label{fig:sextic_lambda_max_N=50}
                \end{subfigure}
                \\
                \begin{subfigure}[b]{0.49\textwidth}
                    \centering
                    \includegraphics[width=\textwidth]{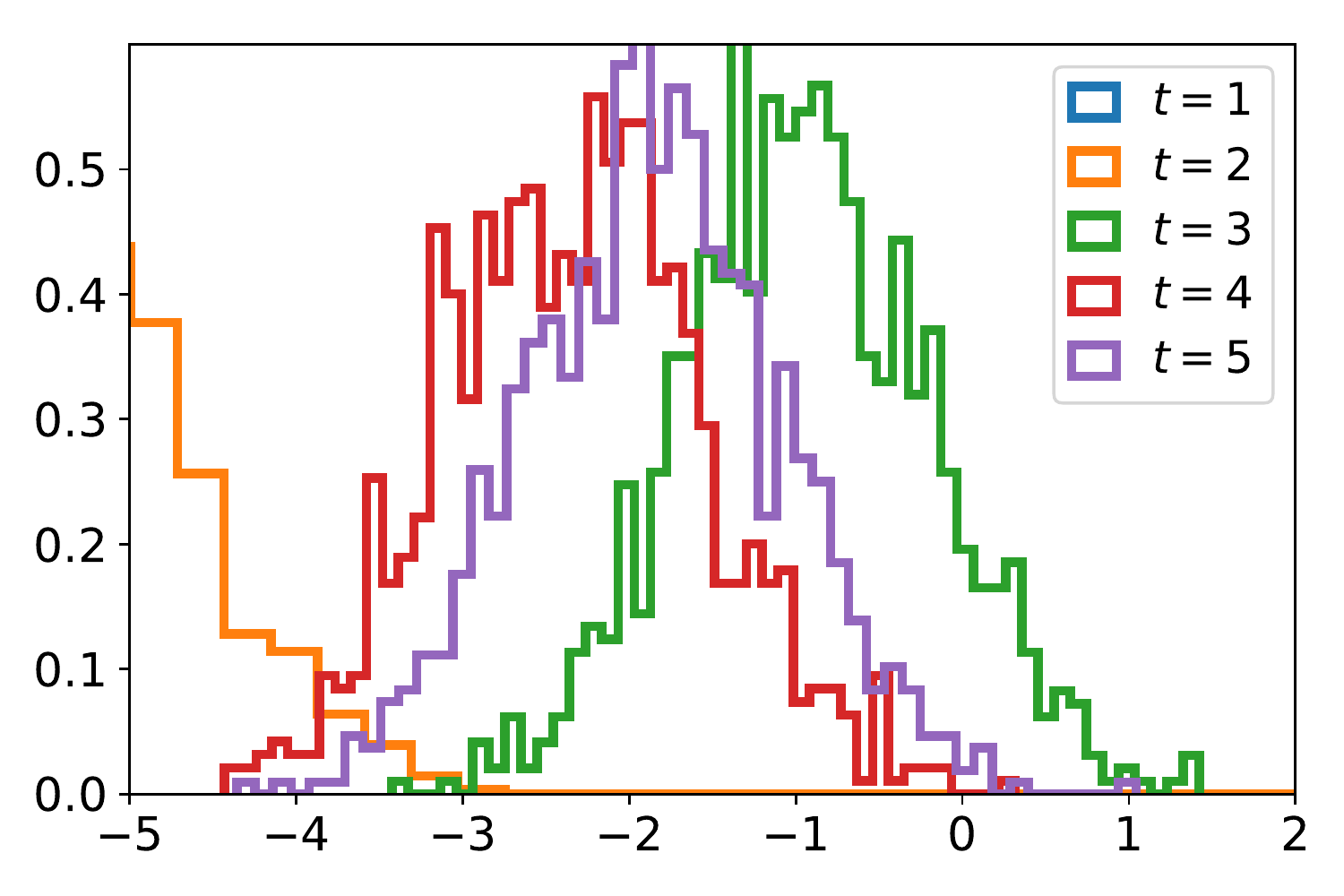}
                    \caption{$N=100$}
                    \label{fig:sextic_lambda_max_N=100}
                \end{subfigure}
                \hfill
                \begin{subfigure}[b]{0.49\textwidth}
                    \centering
                    \includegraphics[width=\textwidth]{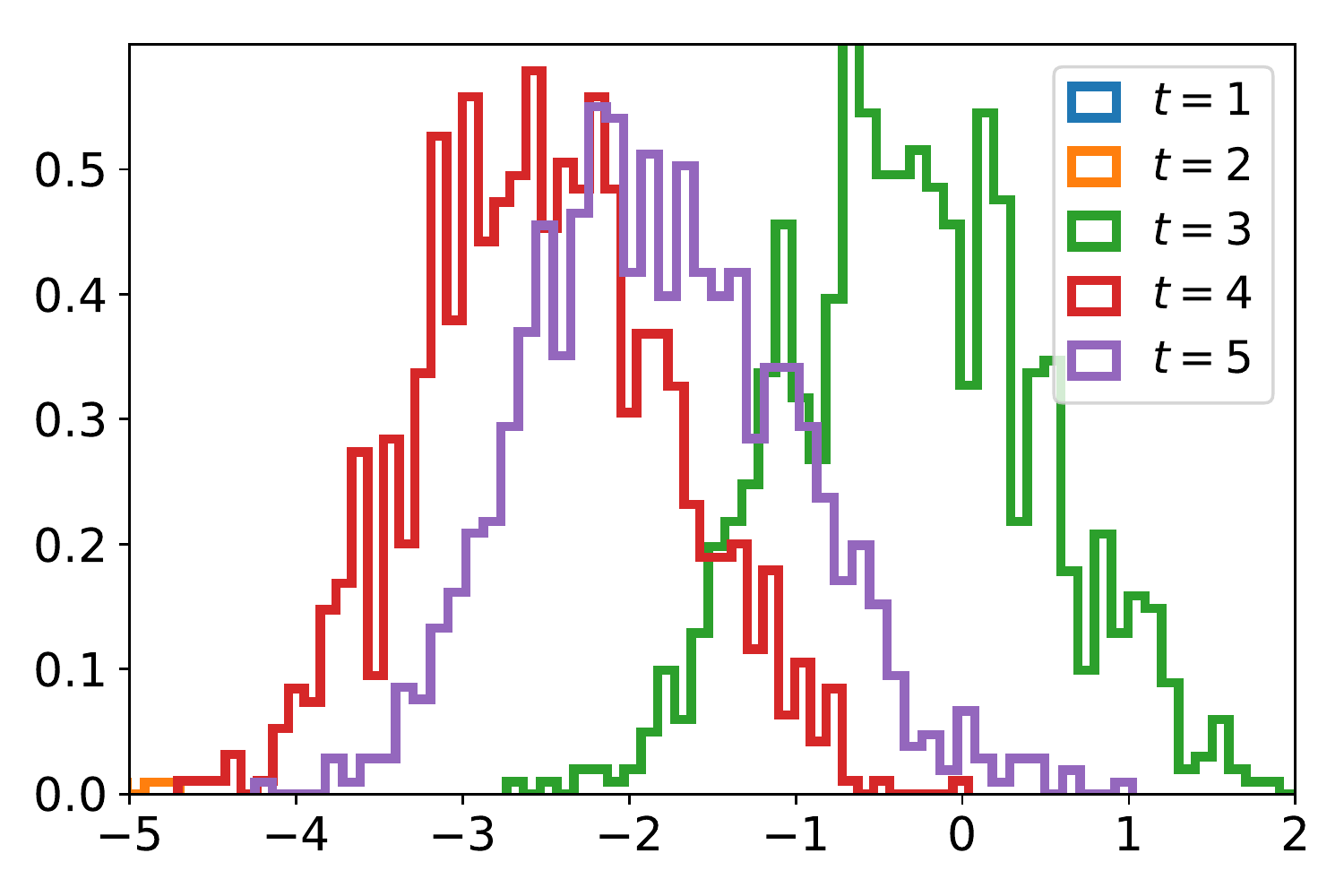}
                    \caption{$N=200$}
                    \label{fig:sextic_lambda_max_N=200}
                \end{subfigure}
                \caption{%
                    For $\beta=2$ and $V(x)=\frac16 x^6$, we give a visual display of the convergence of the empirical distribution of the largest atom of $\mu_N^t$ constructed from $1000$ independent chains to the Tracy-Widom distribution.
                    Each colored line corresponds to a Gibbs pass $t\in\{1,2,3,...\}$, while the equilibrium pdf is shown as a black line on each panel.
                    Different panels correspond to increasing values of $N$.
                    As $N$ increases the histogram of passes $t=1$ and $t=2$ are farther from matching the high density regions of the Tracy-Widom law and convergence is not as fast as for the quartic case; see also Figures~\ref{fig:sextic_lambda_max_cdf_convergence_visual}-\subref{fig:sextic_lambda_max_cdf_smoothed} for a more quantitative monitoring.
                    }
                \label{fig:sextic_lambda_max_visual}
            \end{figure}

            \begin{figure}[ht]
                \centering
                \begin{subfigure}[b]{0.49\textwidth}
                    \centering
                    \includegraphics[width=\textwidth]{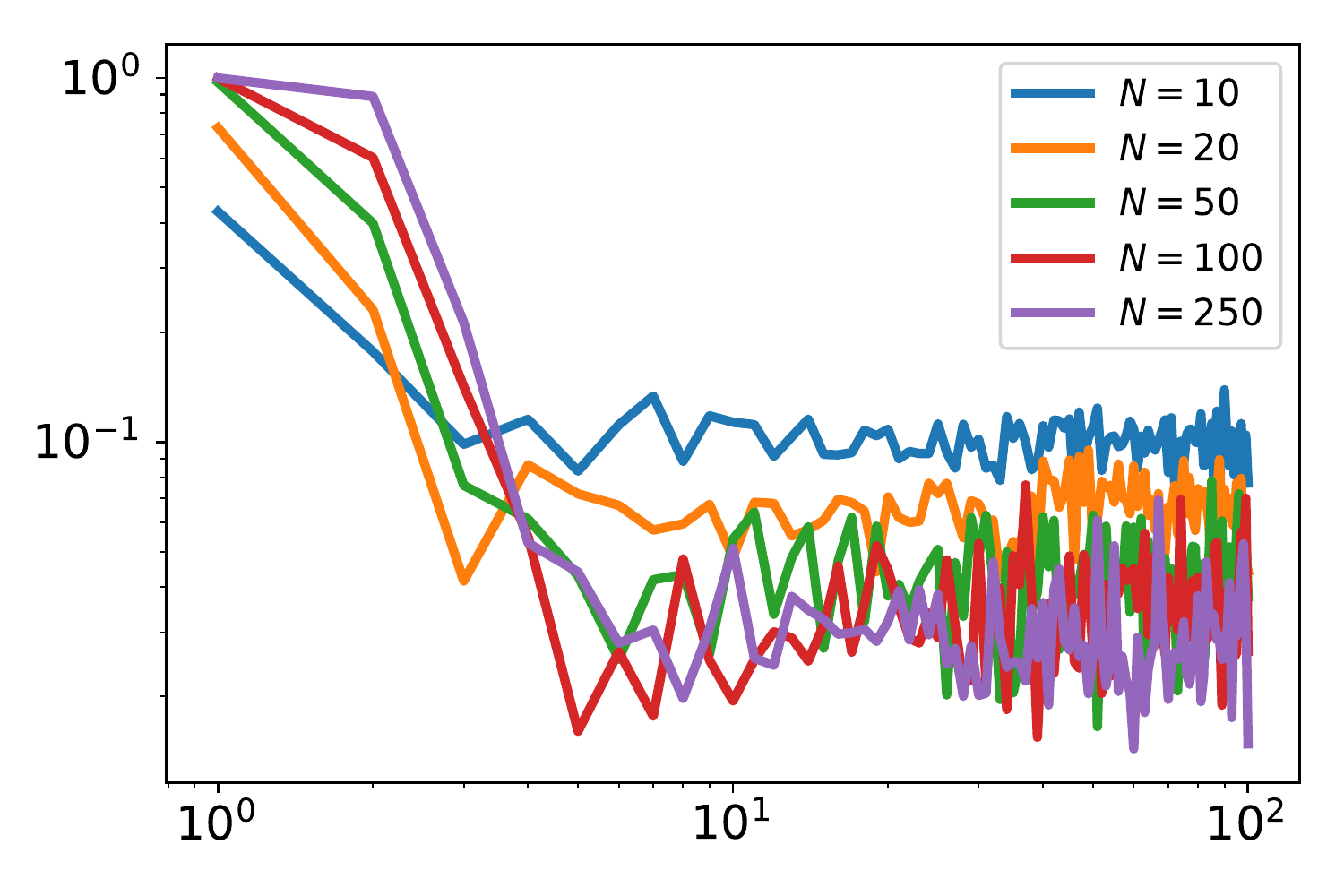}
                    \caption{$V(x)=\frac14 x^4$}
                    \label{fig:quartic_lambda_max_cdf_convergence_visual}
                \end{subfigure}
                \hfill
                \begin{subfigure}[b]{0.49\textwidth}
                    \centering
                    \includegraphics[width=\textwidth]{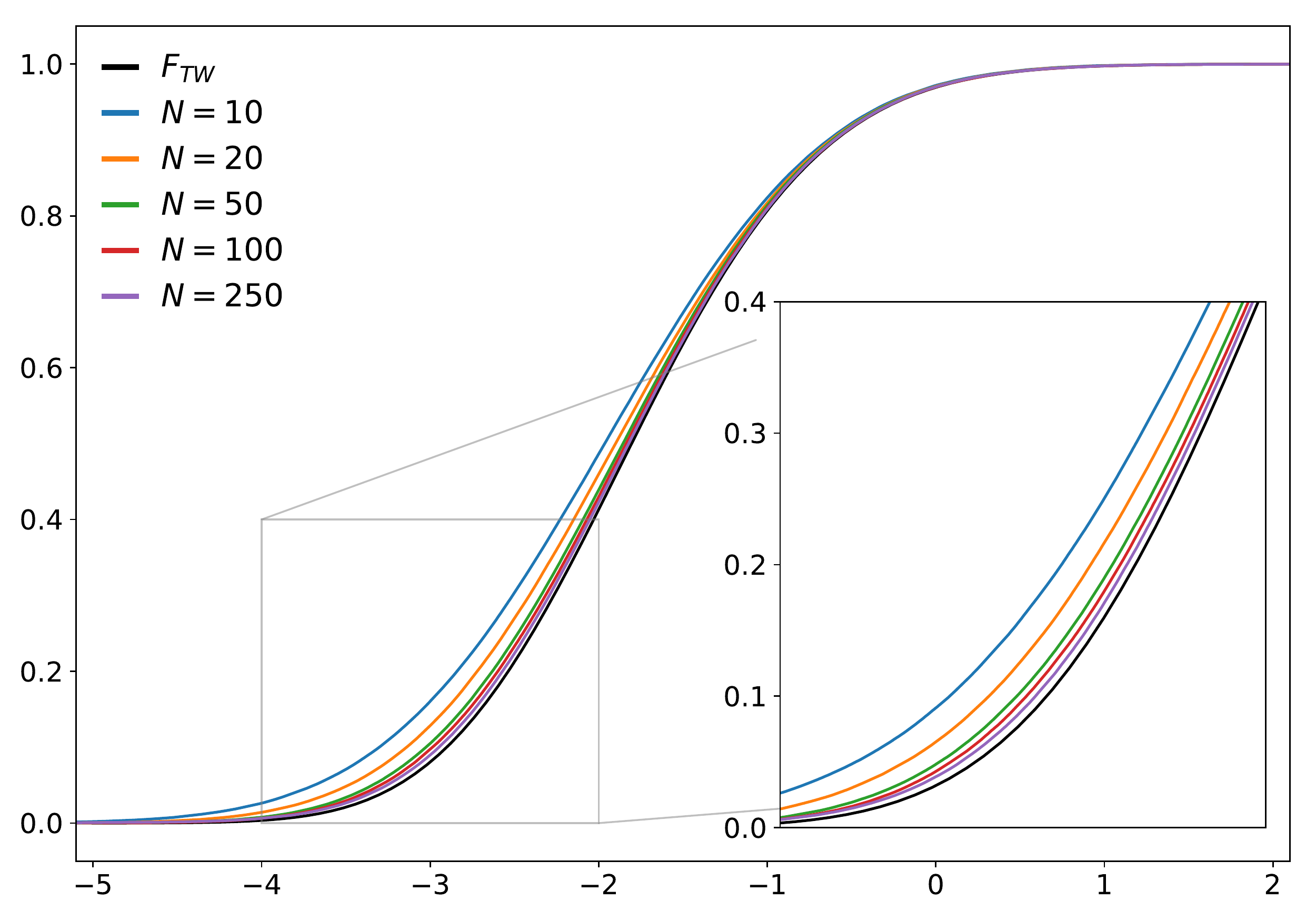}
                    \caption{$V(x)=\frac14 x^4$}
                    \label{fig:quartic_lambda_max_cdf_smoothed}
                \end{subfigure}
                \\
                \begin{subfigure}[b]{0.49\textwidth}
                    \centering
                    \includegraphics[width=\textwidth]{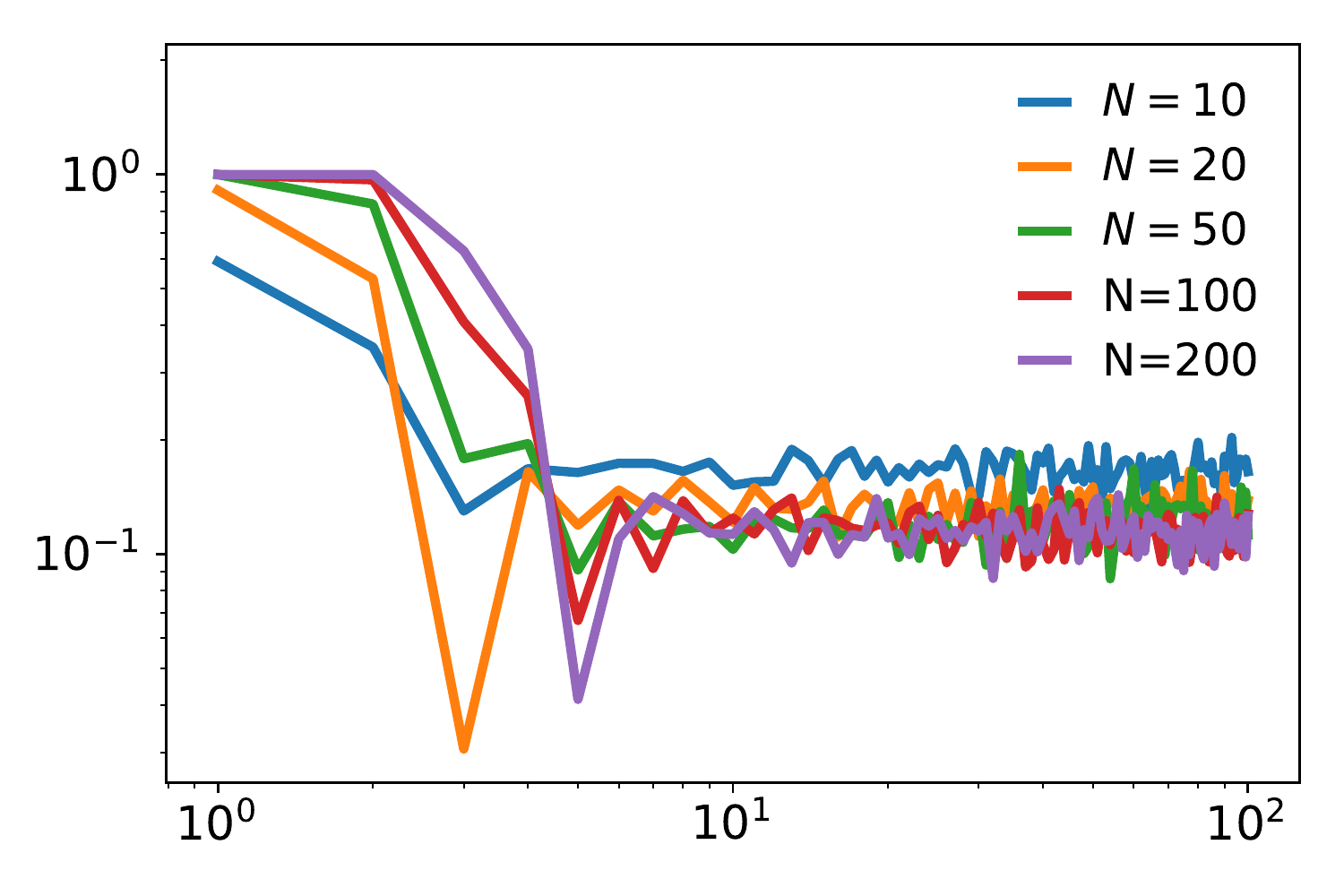}
                    \caption{$V(x)=\frac16 x^6$}
                    \label{fig:sextic_lambda_max_cdf_convergence_visual}
                \end{subfigure}
                \hfill
                \begin{subfigure}[b]{0.49\textwidth}
                    \centering
                    \includegraphics[width=\textwidth]{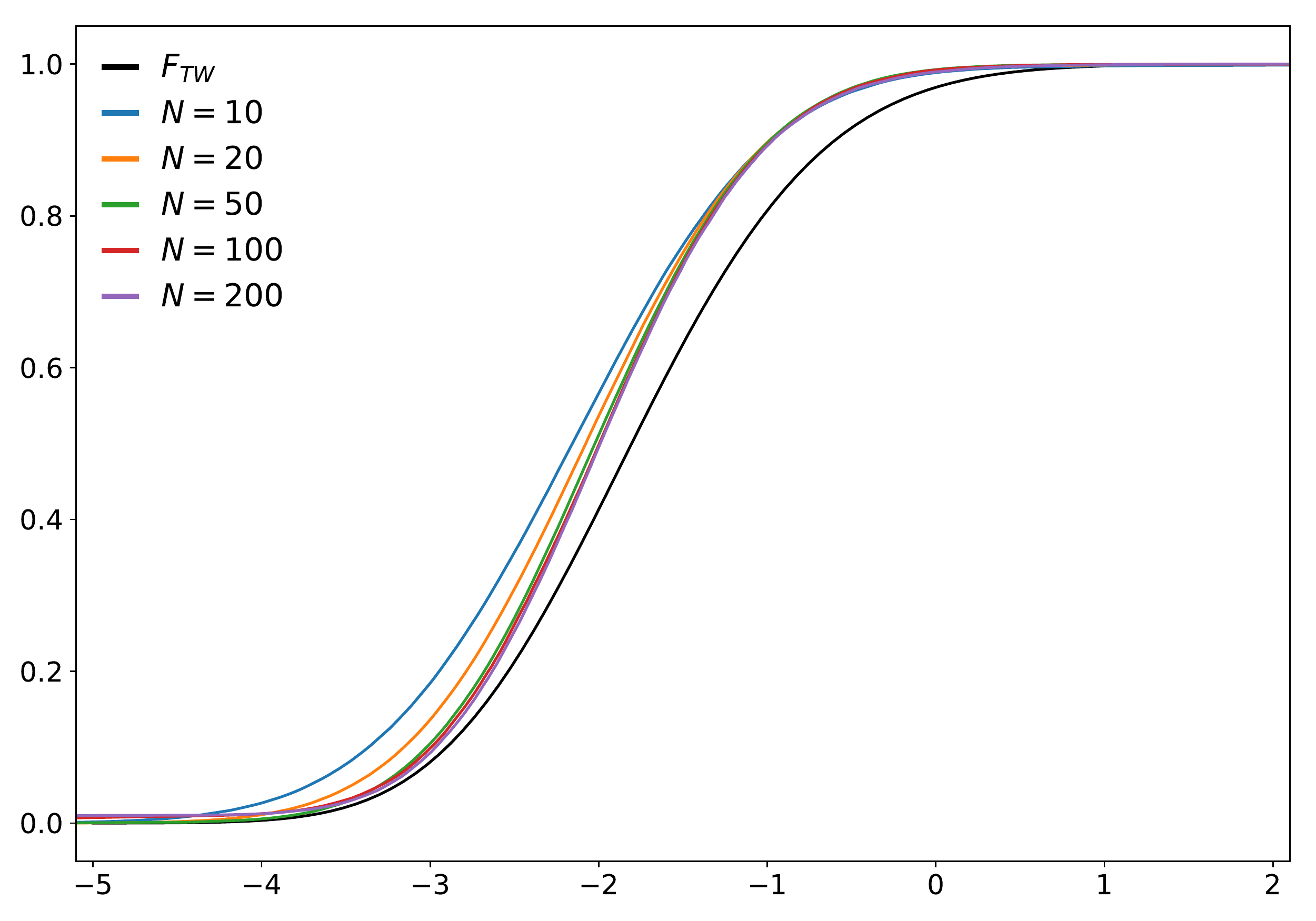}
                    \caption{$V(x)=\frac16 x^6$}
                    \label{fig:sextic_lambda_max_cdf_smoothed}
                \end{subfigure}
                \caption{%
                    For $\beta=2$ and $V(x)=\frac14 x^4$ and $\frac16 x^6$, we monitor the convergence of the empirical cdf of the largest atom of $\hat\mu_N^t$ to the expected Tracy-Widom distribution, for several values of the number of points $N$.
                    Panels~\subref{fig:quartic_lambda_max_cdf_convergence_visual} and~\subref{fig:sextic_lambda_max_cdf_convergence_visual}
                    show the supremum norm of the difference between the cdf of the Tracy-Widom distribution and the empirical cdf of the largest atom of $\hat\mu_N^t$, constructed from $1000$ independent chains, as a function of the number $1\leq t\leq 100$ of Gibbs passes.
                    Panels~\subref{fig:quartic_lambda_max_cdf_smoothed} and~\subref{fig:sextic_lambda_max_cdf_smoothed} show the corresponding smoothed empirical cdf constructed from the aggregation of the $1000$ independent chains over the $t=100$ passes.
                    Different panels correspond to increasing values of $N$.
                    }
                \label{fig:quartic_sextic_lambda_max}
            \end{figure}




\section{Conclusion}

    First, we wrote down the details of an elementary proof of the three classical tridiagonal models for $\beta$-ensembles.
    Most arguments of the proof already appeared in work by \citet{DuEd02,KiNe04,FoRa06,GaRo10,DeNa12,KrRiVi16} and we take no credit for the originality of the proof, only for a stand-alone and elementary version, akin to a survey.
    We hope that this version will help share the ideas of parametrizing a measure through its recurrence coefficients to computational scientists interested in interacting particle systems.
    Indeed, throughout the proof, we outline natural reparametrizations of $\beta$-ensembles through tridiagonal Jacobi matrices, in which the Vandermonde interaction disappears and leaves only a stream of easy-to-sample, independent matrix entries.
    Coupled with diagonalization of the underlying tridiagonal matrix, this gives a rejection-free, $\mathcal{O}(N^2)$ exact sampler for the three classical $\beta$-ensembles.

    Second, when the potential is more generic, independence is lost, but the new interaction can be short-range.
    We exploited this property to implement a Gibbs kernel and a Metropolis-within-Gibbs variant, which sample $\beta$-ensembles with polynomial potentials.
    This leads to simple MCMC samplers that empirically mix much faster, even for a large number of points, than more sophisticated MCMC kernels working in the original domain of the particles \citep{LiMe13,ChFe18}
    In particular, marginal behavior that matches known theoretical results can be obtained in a few Gibbs passes, totaling a few seconds on a laptop for hundreds of points.
    However, local behavior, such as the law of the largest particle in the $\beta$-ensemble, remains harder to approximate as the degree of the potential grows.
    Finally, to be fair, we note that the sampler of \citet{ChFe18} applies much more generally than ours, and in particular to multivariate $\beta$-ensembles.

    Finally, we want to stress a third related approach, which we leave for future work.
    As we have seen, diagonalizing a random Jacobi matrix is equivalent to solving a randomized moment problem.
    One can thus cast sampling $\beta$-ensembles as a constrained optimization problem, namely a linear program, as in the work of \citet{RyBo15}, but with randomized constraints.
    Our own interest in tridiagonal models actually came from trying to generalize a sampler for finite determinantal point processes \citep*{GaBaVa17} of this very form.
    It is then tempting to look for multivariate versions of the corresponding randomized linear program.
    We conjecture that the semidefinite relaxations of \citet{Las10} of multivariate moment problems, with properly randomized constraints, would lead to efficient samplers for multidimensional $\beta$-ensembles.
    This is a technically difficult next step, both in mathematical and computational terms, but it would be useful for Monte Carlo integration \citep{BaHa19}.


\section*{Acknowledgements}
    We acknowledge support from ANR grant \textsc{BoB} (ANR-16-CE23-0003) and ERC grant \textsc{Blackjack} (ERC-2019-STG-851866).
\endgroup

  \clearpage
  \bibliographystyle{icml2017}
  \bibliography{biblio}

\end{document}